%% file: main.tex
\documentclass[reqno, 11pt]{amsart}

\usepackage{geometry}
\usepackage{caption}
\usepackage[english]{babel} 
\usepackage[utf8]{inputenc}
\usepackage[T1]{fontenc}
\usepackage{lmodern}
\usepackage{setspace}

\usepackage[
backend=biber, 
style=numeric,
sorting=nyt,
giveninits=true, 
doi=true,
url=false,
isbn=false,
eprint=false
]{biblatex} 

\usepackage[autostyle]{csquotes}
\usepackage{url}
\usepackage[colorlinks, allcolors=blue]{hyperref}
\hypersetup{breaklinks=true}
\AtEveryBibitem{
	\clearfield{note}
}
\usepackage[noorphans,font=itshape]{quoting}

\AtEveryBibitem{
	\clearfield{month}
	\clearfield{urldate}
	\clearfield{language}
	\clearlist{language}
}

\renewbibmacro*{volume+number+eid}{
	\printfield{volume}
	\iffieldundef{number}{}{
		\printtext{\addspace(\printfield{number})}
	}
	\setunit{\addspace}
	\printfield{eid}}

\DeclareFieldFormat[article]{title}{#1}     
\DeclareFieldFormat[incollection]{title}{#1}   
\DeclareFieldFormat[book]{title}{\mkbibemph{#1}} 
\DeclareFieldFormat{pages}{#1} 

\DeclareBibliographyDriver{article}{
	\usebibmacro{author/translator+others}
	\setunit{\addspace}
	\printtext[parens]{\printfield{year}}
	\setunit{\addcomma\space}
	\newunit
	\printfield{title}
	\setunit{\addcomma\space}
	\newunit
	\printfield{journaltitle}
	\setunit{\addcomma\space}
	\usebibmacro{volume+number+eid}
	\setunit{\addcolon\space}
	\printfield{pages}
	\finentry
}

\DeclareBibliographyDriver{book}{
	\usebibmacro{author/translator+others}
	\setunit{\addspace}%
	\printtext[parens]{\printfield{year}}
	\setunit{\addcomma\space}
	\newunit
	\printfield[booktitle]{title}
	\setunit{\addcomma\space}
	\printlist{publisher}
	\finentry
}

\DeclareBibliographyDriver{misc}{
	\usebibmacro{author/translator+others}
	\setunit{\addspace}
	\printtext[parens]{\printfield{year}}
	\setunit{\addcomma\space}
	\newunit
	\printfield{title}
	\setunit{\addcomma\space}
	\newunit
	\printfield{type}
	\setunit{\addcomma\space}
	\printlist{address}
	\finentry
}

\DeclareFieldFormat[thesis]{title}{#1}

\DeclareBibliographyDriver{phdthesis}{
	\usebibmacro{author/translator+others}
	\setunit{\addspace}
	\printtext[parens]{\printfield{year}}
	\setunit{\addcomma\space}
	\newunit
	\printfield[title]{title}
	\setunit{\addcomma\space}
	\printlist{institution}
	\finentry
}

\usepackage{enumerate}
\usepackage{mathtools}
\usepackage{mathrsfs}
\usepackage{amsmath}
\usepackage{amsthm}   
\usepackage{amsfonts}
\usepackage{amssymb}
\usepackage{textcomp} 
\usepackage{wasysym}
\usepackage{color}
\usepackage{tikz}
\usetikzlibrary{fit, patterns, positioning, shapes, calc, arrows, backgrounds}

\tikzset{
	assets/.style={
		rectangle split,
		rectangle split parts=1,
		rectangle split part fill={blue!20},
		rounded corners,
		draw=black, very thick,
		minimum height=10cm,
		minimum width=5cm,
		text width=2cm,
		text centered,
	}
}
\tikzset{
	state/.style={
		rectangle split,
		rectangle split parts=2,
		rectangle split part fill={red!30,blue!20},
		rounded corners,
		draw=black, very thick,
		minimum height=2em,
		text width=3cm,
		inner sep=2pt,
		text centered,
	}
}

\usepackage{pgfplots}
\pgfplotsset{
	every tick label/.append style = {font=\tiny},
	every axis label/.append style = {font=\scriptsize}
}
\usepackage{graphicx} 
\usepackage{subcaption}
\usepackage{changes}
\usepackage{hhline}
\usepackage{makecell}

\usepackage{placeins}
\usepackage[most]{tcolorbox}

\makeatletter
\renewenvironment{proof}[1][\proofname]{\par
	\pushQED{\qed}
	\normalfont \topsep6\p@\@plus6\p@\relax
	\trivlist
	\item[\hskip\labelsep
	\itshape
	#1\@addpunct{:}]\ignorespaces
}{
	\popQED\endtrivlist\@endpefalse
}
\makeatother

\usepackage{prettyref}
\usepackage{titleref}
\newrefformat{cha}{Chapter~\ref{#1}}
\newrefformat{sec}{Section~\ref{#1}}
\newrefformat{app}{Appendix~\ref{#1}}
\newrefformat{fig}{Figure~\ref{#1}}
\newrefformat{tab}{Table~\ref{#1}}
\newrefformat{defi}{Definition~\ref{#1}}
\newrefformat{assump}{Assumption~\ref{#1}}
\newrefformat{theorem}{Theorem~\ref{#1}}
\newrefformat{prop}{Proposition~\ref{#1}}
\newrefformat{lem}{Lemma~\ref{#1}}
\newrefformat{cor}{Corollary~\ref{#1}}
\newrefformat{rem}{Remark~\ref{#1}}
\newrefformat{exam}{Example~\ref{#1}}
\newrefformat{eq}{(\ref{#1})}

\newtheoremstyle{italic}
{15pt}
{15pt}
{\normalfont \itshape}
{}
{\normalfont \bfseries}
{.}
{ }
{}

\newtheoremstyle{nonItalic}
{15pt} 
{15pt}
{}
{}
{\normalfont \bfseries}
{.}
{ }
{}

\theoremstyle{italic}
\newtheorem{defi}{Definition}[section]
\newtheorem{theorem}[defi]{Theorem}
\newtheorem{prop}[defi]{Proposition}

\newtheorem{lem}[defi]{Lemma}
\newtheorem{cor}[defi]{Corollary}
\theoremstyle{nonItalic}
\newtheorem{rem}[defi]{Remark}
\newtheorem{exam}[defi]{Example}

\usepackage{accents}

\makeatletter
\patchcmd{\@startsection}
{\@afterindenttrue}
{\@afterindentfalse}
{}{}

\def\@seccntformat#1{
	\protect\textup{
		\protect\@secnumfont
		\expandafter\protect\csname format#1\endcsname 
		\csname the#1\endcsname
		\protect\@secnumpunct
	}
}

\DeclareFieldFormat*{citetitle}{\textit{#1}}

\makeatother

\newcommand{\defgl}{\mathrel{\mathop:\!\!=}}
\newcommand{\emptySet}{\varnothing}
\newcommand{\realNumbers}{\mathbb{R}}
\newcommand*\diff{\mathop{}\!\mathrm{d}}

\newcommand{\symbolOperatorExpectation}{E}
\newcommand{\symbolExpectation}[1]{\symbolOperatorExpectation_{#1}}
\newcommand{\expectation}[2]{\symbolExpectation{#1}[#2]}

\newcommand{\symbolOperatorQuantile}{q}
\newcommand{\symbolQuantile}[1]{\symbolOperatorQuantile_{#1}}
\newcommand{\symbolUpperQuantile}[1]{\symbolQuantile{#1}^{+}}
\newcommand{\quantileFunction}[2]{\symbolQuantile{#1}(#2)}
\newcommand{\upperQuantileFunction}[2]{\symbolUpperQuantile{#1}(#2)}

\newcommand{\level}{\lambda}

\DeclareMathOperator{\symbolOperatorValueAtRisk}{VaR}
\newcommand{\symbolValueAtRisk}[1]{\symbolOperatorValueAtRisk_{#1}}
\newcommand{\valueAtRisk}[2]{\symbolValueAtRisk{#1}(#2)}

\DeclareMathOperator{\symbolOperatorExpectedShortfall}{ES}
\newcommand{\symbolExpectedShortfall}[1]{\symbolOperatorExpectedShortfall_{#1}}
\newcommand{\expectedShortfall}[2]{\symbolExpectedShortfall{#1}(#2)}

\DeclareMathOperator{\pelveSymb}{\Pi}
\newcommand{\pelve}[1]{\pelveSymb_{#1}}
\DeclareMathOperator{\averageLetter}{A}
\newcommand{\apelve}[1]{\pelveSymb^{\averageLetter}_{#1}}
\DeclareMathOperator{\mseLetters}{MSE}
\newcommand{\msepelve}[1]{\pelveSymb^{\mseLetters}_{#1}}
\DeclareMathOperator{\wcLetters}{WC}
\newcommand{\wcpelve}[1]{\pelveSymb^{\wcLetters}_{#1}}
\DeclareMathOperator{\bcLetters}{BC}
\newcommand{\bcpelve}[1]{\pelveSymb^{\bcLetters}_{#1}}
\DeclareMathOperator{\sysLetters}{Sys}
\newcommand{\syspelve}[1]{\pelveSymb^{\sysLetters}_{#1}}
\DeclareMathOperator{\unitVec}{\mathbf{e}}

\newcommand{\tES}{G}
\DeclareMathOperator{\var}{Var}

\title{PELVE from a regulatory perspective}

\date{\today}

\usepackage{fancyhdr}
\usepackage{orcidlink}

\begin{document}
	
	\author{Christian Laudag\'e \orcidlink{0000-0002-5910-9487}}
	\address{Department of Mathematics, RPTU Kaiserslautern-Landau, Germany.}
	\email{christian.laudage@rptu.de}
	
	\author{Jörn Sass \orcidlink{0000-0003-3432-4997}}
	\address{Department of Mathematics, RPTU Kaiserslautern-Landau, Germany.}
	\email{joern.sass@rptu.de}

	\pagenumbering{arabic}
	\numberwithin{equation}{section}\numberwithin{equation}{section}
	
	\maketitle
	
	\setstretch{1.1}

	\input{./sections/abstract}

	\input{./sections/introduction}
	\input{./sections/pelve}

	\input{./sections/optimizationProblemsPelve}
	\input{./sections/exampleNormalDistributions}
	\input{./sections/caseStudy}
	\input{./sections/conclusion}
	\input{./sections/declarations}
    
    
    \printbibliography
    
    \appendix
    
    \input{./sections/appendix}

\end{document}

%% file: sections/abstract.tex
\begin{abstract}
	
	Under Solvency II, the Value-at-Risk (VaR) is applied, although there is broad consensus that the Expected Shortfall (ES) constitutes a more appropriate risk measure. Moving towards ES would necessitate specifying the corresponding ES level. The recently introduced Probability Equivalent Level of VaR and ES (PELVE) determines this by requiring that ES equals the prescribed VaR for a given future payoff, reflecting the situation of an individual insurer. We incorporate the regulator's perspective by proposing PELVE-inspired methods for \textit{multiple} insurers. We analyze existence and uniqueness of the resulting ES levels, derive expressions for elliptically distributed payoffs and establish limit results for multivariate regularly distributed payoffs. A case study highlights that the choice of method is crucial when payoffs arise from different distribution families. We provide recommendations which of our PELVE-inspired methods are most appropriate in certain scenarios.
	
	\medskip
	
	\item[\hskip\labelsep\scshape Keywords:] Elliptical distribution, Expected Shortfall, PELVE, regularly varying distribution, systemic risk, Value-at-Risk
	
	\medskip
	
\end{abstract}

%% file: sections/introduction.tex
\section{Introduction}\label{sec:introduction}

Recently,~\textcite{li_pelve_2023} defined the Probability
Equivalent Level of VaR and ES (PELVE). For a given random variable, reflecting the future payoff of an agent, it describes the ES level at which the ES equals a predefined VaR. The intended transition from VaR to ES is relevant for the insurance industry, because the Solvency II regulation is still based on the VaR, even if the common consensus is that an ES is superior, see the discussions in~\cite{artzner_coherent_1999, embrechts_quantitative_2015, kratschmer_comparative_2014}. In the banking sector, VaR has already been replaced by ES, see~\cite{embrechts_academic_2014} for a discussion regarding the Basel accords.\footnote{Moving from VaR to ES also entails certain drawbacks. For instance, under ES it is possible that more extreme default behaviors are considered acceptable, see~\cite{koch-medina_unexpected_2016}.}

PELVE is based on a \textit{one-dimensional} random variable, which makes it an appropriate measure to analyze the situation of an individual insurer. However, it may not be appropriate for a regulator who needs to determine an ES level for \textit{multiple} insurers. In this case, we cannot expect an ES level for which none of the insurers would need to adjust their capital reserves. In fact, such a scenario would rather be the exception than the rule. This raises the central question of this manuscript:
	
\begin{center}
	\textit{How should the regulator determine an appropriate ES level across multiple insurers?}
\end{center}
	
To answer this question, we model the payoffs of multiple insurers via a \textit{multi-dimensional} random variable and introduce multivariate extensions of PELVE. These extensions provide the regulator with a quantitative tool for decision-making, eliminating reliance on intuition solely based on the individual situations of insurers. 

The central question identified above has not been systematically examined in the existing actuarial literature: \cite[Figure 2]{boonen_solvency_2017} or~\cite[Figure 2.2]{rroji_risk_2013} give empirical results for an appropriate ES level only in the case of a one-dimensional random variable. \textcite{li_pelve_2023} formalised their empirical findings by introducing the innovative concept of PELVE. \textcite{assa_calibrating_2024} use prescribed PELVE values to find a suitable distribution model, that is, PELVE is utilized to  calibrate distribution functions. Furthermore, \cite{han_diversification_2023} demonstrates that PELVE is helpful to obtain a fair comparison between the diversification quotients of VaR and ES.  

Regarding the PELVE methodology, numerous studies replace ES by other risk measures: \textcite{ortega-jimenez_probability_2022} introduce the PELCoV, which is the Probability Equivalent Level between Co-Value-at-Risk (CoVaR) and VaR, where CoVaR is a VaR based on a conditional distribution function.\footnote{To be precise, the Co-Value-at-Risk $\text{CoVaR}_{v,u}(Y|X)$ is the VaR of a risk $Y$ at level $v$ conditioned that some other risk $X$ is equal to its VaR at level $u$. So, CoVaR is the (generalized) inverse of the conditional distribution of $Y$ at $v$ given $X=\valueAtRisk{u}{X}$. Consequently, PELCoV at level $v$ is the value $u$ such that $\text{CoVaR}_{v,u}(Y|X) = \valueAtRisk{\lambda}{Y}$.} Another concept is  PELVE$_n$, defined in~\cite{barczy_probability_2023}. It is the Probability Equivalent Level between VaR and higher-order ES. Finally,~\textcite{fiori_generalized_2023} suggest two extensions of PELVE -- distorted PELVE and generalized PELVE. The distorted PELVE is the equivalent level between VaR and parametrized Wang premia. The latter are typically applied to calculate insurance premia. The generalized PELVE substitutes VaR and ES by two families of risk measures, where the second family is obtained by integration with respect to the first. 

Instead, we focus on the regulatory perspective, for which it is indispensable to consider the payoffs of multiple insurers. Accordingly, we propose the following methods to determine an ES level corresponding to a given VaR level: (1) A-PELVE: weighted average of the PELVE values of the agents; (2) BC- and WC-PELVE: smallest (the best-case, BC) and largest (the worst-case, WC) PELVE value over all agents; (3) MSE-PELVE: minimizes  the weighted mean-squared-error (MSE) between ES and VaR; (4) Sys-Pelve: minimum level such that the systemic risk measured by the ES does not exceed the systemic risk measured by the VaR. We summarize these methods under the umbrella of \textit{Multi-PELVE}. 

Our journey then starts by developing auxiliary results for ES and PELVE curves. These results are needed to analyze the Multi-PELVE methods, but they are also of independent interest. Then, we state existence results. Furthermore, in the majority of the cases, a Multi-PELVE method gives a unique solution. However, a main finding is a comprehensive example, showing that the MSE-PELVE is non-unique in general. Furthermore, for elliptically distributed payoffs, we find that most of the Multi-PELVE methods can be represented by a one-dimensional elliptically distributed random variable. To model extreme scenarios, we employ distributions with multivariate regularly varying tails. In this framework, and as the VaR level approaches zero, all methods converge to the PELVE of a one-dimensional Pareto-distributed random variable. 

In a case study, we first assume that the equity capital distributions are from the same distribution family. In this situation, the choice of a Multi-PELVE method (aside from the BC- and the WC-PELVE) is not significant. Also the Multi-PELVE values are close to the PELVE values of each individual insurer and so, no tremendous changes in the capital reserves occur by moving from VaR to ES. This is a convenient situation for the regulator.

However, this behavior can change drastically when the equity capitals arise from different distribution families. To illustrate this, we perturb the model by replacing some light-tailed distributions with heavy-tailed ones. In doing so, the choice of the Multi-PELVE method becomes significant: BC-PELVE and WC-PELVE do not seem to be adequate choices, because they either require or allow some insurers to drastically change their capital reserves. With respect to overall changes in the capital reserves, the Sys-PELVE and the MSE-PELVE provide convenient results. Furthermore, the A-PELVE and a weighted variant of the MSE-PELVE do not prioritize a market leader, as it is the case for other Multi-PELVE versions. This weighted MSE-PELVE and the A-PELVE are the recommended Multi-PELVE methods for the second part of the case study. 

The manuscript is structured as follows: In~\prettyref{sec:var_es_pelve} we recall the definitions of VaR and ES as well as the definition of PELVE. Furthermore, we discuss auxiliary results for ES and PELVE curves and explain the regulator's objective. Then, in~\prettyref{sec:problemRegulator}, we state the definitions of the Multi-PELVE methods, discuss properties like e.g.,~existence and uniqueness, and compare their pros and cons. In~\prettyref{sec:concreteTheoreticalDistributions}, we focus on elliptical and heavy-tailed distributions. Finally, in the case study in~\prettyref{sec:caseStudy} we compare the Multi-PELVE methods in a comprehensive way. All proofs are gathered in the appendix. 

%% file: sections/pelve.tex
\section{Preliminaries and regulator's objective}\label{sec:var_es_pelve}

In this section, we first repeat the definition of PELVE from~\textcite{li_pelve_2023}. We then state new results regarding ES and PELVE curves, which are used in the upcoming sections and which are also of independent interest. Finally, we discuss the regulator's perspective and formally introduce the problem of interest. 

\subsection{VaR, ES and PELVE}

We impose an atomless probability space $(\Omega,\mathcal{F},P)$.\footnote{Atomless ensures the existence of continuous random variables~\parencite[Proposition A.31]{follmer_stochastic_2016} in~\prettyref{sec:caseStudy}.} For simplicity, we work in the space of integrable random variables $L^1(\Omega,\mathcal{F},P)$ or $L^1$ for short. A random variable $X\in L^1$ stands for the payoff of an agent at a fixed future time point. The corresponding distribution function is denoted by $F^{X}$.

The two most common risk measures are Value-at-Risk and Expected Shortfall:
\begin{defi}
	\begin{enumerate}[(i)]
		\item The Value-at-Risk (VaR) at level $\lambda\in(0,1)$ of $X\in L^1$ is 
		\begin{align*}
			\symbolValueAtRisk{\lambda}(X) := -q_{X}^{+}(\lambda) := \inf\{m\in\mathbb{R}\,|\, P(X<-m)\leq \lambda\},  
		\end{align*}
		where $q_{X}^{+}$ is the upper quantile function of $X$. 
		\item The Expected Shortfall (ES) at level $\lambda \in(0,1]$ of $X\in L^1$ is
		\begin{align*}
			\symbolExpectedShortfall{\lambda}(X) := \frac{1}{\lambda}\int\limits_{0}^{\lambda}\symbolValueAtRisk{u}(X)\,\diff u.  
		\end{align*}
	\end{enumerate}
\end{defi}

Our starting point is the map in the next definition, which follows~\textcite{li_pelve_2023}. We always use the convention $\inf\emptySet = \infty$.
\begin{defi}\label{defi:pelve}
	The Probability Equivalent Level of VaR and ES (PELVE) at level $\lambda\in(0,1)$ is the map $\Pi_{\lambda}:L^1\rightarrow [1,\lambda^{-1}]\cup\{\infty\}$ defined by
	\begin{align*}
		\Pi_{\lambda}(X):= \inf\left\{c\in[1,\lambda^{-1}]\,\middle|\,\symbolExpectedShortfall{c\lambda}(X)\leq \symbolValueAtRisk{\lambda}(X)\right\}.
	\end{align*}
\end{defi}

\begin{rem}
	Note that our definitions of VaR and ES are the ones for payoffs (instead of losses), see e.g.,~\textcite{artzner_coherent_1999} and~\textcite{follmer_stochastic_2016}. As consequence, our definition of PELVE differs from the one in~\cite{li_pelve_2023}. Ours is the PELVE version when VaR and ES are introduced in the aforementioned way for financial payoffs.
\end{rem}

\begin{rem}\label{rem:existenceUniquenessPELVE}
	By~\cite[Proposition 1]{li_pelve_2023} we know that $\pelve{\lambda}(X)$ exists, meaning that $\pelve{\lambda}(X)<\infty$, if and only if $\expectation{}{-X}\leq \valueAtRisk{\lambda}{X}$. If in addition, $p\mapsto \valueAtRisk{p}{X}$ is not a constant function on $(0,\lambda]$, guaranteeing that $p\mapsto \expectedShortfall{p}{X}$ is strictly decreasing on $[\lambda,1]$, then, by~\cite[Proposition 2]{li_pelve_2023} there exists a unique $c\in[1,\lambda^{-1}]$ such that $\symbolExpectedShortfall{c\lambda}(X)= \symbolValueAtRisk{\lambda}(X)$.
\end{rem}

The PELVE $c^{\ast} = \Pi_{\lambda}(X)$ refers to the ES level $c^{\ast}\lambda$ for which the corresponding ES is equal to the VaR at level $\lambda$ given the payoff $X$. PELVE is useful for analyzing the impact on the capital reserve of an individual insurer with payoff $X$ when the regulator requires the insurer to apply from now on an ES at level $c\lambda$ instead of the VaR at level $\level$: If $c>c^{\ast}$, then the new reserve  $\expectedShortfall{c\lambda}{X}$ is smaller than the previous one given by $\valueAtRisk{\lambda}{X}$. For  $c<c^{\ast}$ the opposite is true.

Economically, changing to an ES level for which the capital reserve does not change drastically is desirable: On the one hand, if the ES becomes significantly larger than the VaR, the insurer may be unable to raise sufficient additional reserves quickly. On the other hand, if the ES becomes substantially smaller than the VaR, the insurer would be insufficiently capitalized, which endangers its financial stability and is particularly undesirable from a regulatory perspective, because it also increases the systemic risk arising within the entire network of insurers.

Two concepts related to the analysis of PELVE are the ES curve and the PELVE curve. Hence, before we discuss the regulatory perspective in detail, we state specific results for such curves. 

\subsection{Results on ES curves and PELVE curves}\label{sec:esCurves}

We start by collecting results for the ES interpreted as a function of its level, the so-called ES curve. For brevity, two standard properties of these ES curves are shifted to the appendix, see~\prettyref{lem:nonConcavityEsCurve} and~\prettyref{lem:esStrictlyDecreasing}. Here, we characterize when functions from $[0,1]$ to $\mathbb{R}$ can be represented by an ES curve, i.e.,~by a function in $\{[0,1]\rightarrow \mathbb{R},\lambda\mapsto \expectedShortfall{\lambda}{X}\,|\, X\in L^1\}$. It is a consequence of the fundamental theorem of calculus for the Lebesgue integral, see~\cite[Theorem 7.18]{rudin_real_1987}. We use this result also later in~\prettyref{exam:uniquenessMSEpelve}.

\begin{theorem}\label{thm:existenceOfEScurve}
	Let $\tES:(0,1]\rightarrow\mathbb{R}$. The following statements are equivalent:
	\begin{enumerate}[(a)]
		\item There exists $X\in L^1$ such that for all $t\in(0,1]$ it holds that $\tES(t) = t\expectedShortfall{t}{X}$.
		\item It holds that $\lim\limits_{u\downarrow 0}\tES(u) = 0$ and $\tES$ extended to $[0,1]$ is absolutely continuous and there exists a Lebesgue null set $A\subseteq[0,1]$ such that the derivative $-\tES^{\prime}(t)$ exists for all $t\in[0,1]\setminus A$ and it is increasing and right-continuous. 
	\end{enumerate}
\end{theorem}

ES curves are decreasing, which would automatically imply, e.g.,~by~\cite[Theorem 3.23 (b)]{folland_real_2013}, that they are almost everywhere differentiable, without relying on the fundamental theorem of calculus. The decreasing behavior of ES curves allows us to invoke~\prettyref{thm:existenceOfEScurve} to establish a key property of these curves.
\begin{cor}\label{cor:conclusionTheoremEScurves}
	Let $\tES:(0,1]\rightarrow \mathbb{R}$ be a function as in~\prettyref{thm:existenceOfEScurve} and set $f(t)=\tES(t)/t$ for all $t\in(0,1]$.
	Then, $f$ is either strictly decreasing or there exists $a\in(0,1]$ such that $f$ is constant on $(0,a]$ and strictly decreasing on $(a,1]$.
\end{cor}
	
For a twice differentiable function, we give a more concise sufficient condition that it can be represented as ES curve. 

\begin{cor}\label{cor:existenceOfEScurve}
	Let $f:(0,1]\rightarrow\mathbb{R}$ be twice differentiable,  decreasing and concave. Then there exists $X\in L^1$ such that for all $t\in(0,1]$ it holds that $f(t) = \expectedShortfall{t}{X}$.
\end{cor} 

The proof of~\prettyref{cor:existenceOfEScurve} applies~\prettyref{thm:existenceOfEScurve} by setting $\tES(t) = tf(t)$. Concavity of $f$ then ensures that $t\mapsto -\tES^{\prime}(t)=-f(t)-tf^{\prime}(t)$ is increasing. \prettyref{exam:counterxampleCorollaryEScurve} in the appendix states a function $f$, being twice differentiable and decreasing, but not concave, in such a way that  $(0,1]\rightarrow\mathbb{R},t\mapsto -f(t)-tf^{\prime}(t)$ is not increasing, which shows that we cannot omit the assumption of concavity in~\prettyref{cor:existenceOfEScurve}.

Parts of the case study in~\prettyref{sec:caseStudy} rely on PELVE curves. They are used in~\cite{assa_calibrating_2024} to calibrate distribution functions. Here, we summarize theoretical properties of these curves, aiming to derive a discontinuity result for PELVE curves in combination with empirical distribution functions. To do so, we first state an equivalent condition for a PELVE curve to be continuous. 
\begin{prop}\label{prop:continuityPELVE}
	Let $\alpha\in(0,1]$ and $X\in L^1$. The following statements are equivalent:
	\begin{enumerate}[(a)]
		\item The map $f:(0,\alpha)\rightarrow [1,\infty],\lambda \mapsto \pelve{\lambda}(X)$ is continuous and finite.
		\item For all $\lambda\in (0,\alpha)$ it holds that $\expectation{}{-X}\leq \valueAtRisk{\lambda}{X}$ and that $q_X^{+}$ is continuous on $(0,\alpha)$.
	\end{enumerate}
\end{prop}

The PELVE curves illustrated in~\prettyref{sec:caseStudy} rely on empirical distribution functions. The latter are step functions and thus, the next result tells us that the corresponding PELVE curves are discontinuous.
\begin{cor}\label{cor:discontinuityPELVEcurves}
	Let $X\in L^1$ such that $F^{X}$ is a step function. In addition, assume that there exists $x\in\mathbb{R}$ with $F^{X}(x)\in(0,1)$.\footnote{This assumption ensures that $F^{X}$ admits more than one jump.} Then, the PELVE curve $(0,1)\ni\lambda\mapsto\pelve{\lambda}(X)$ is discontinuous.
\end{cor}

\subsection{The regulator's objective}

As previously discussed, PELVE facilitates the analysis of the risk exposure of a single agent. We now aim to extend this concept by proposing methods that allow a regulator to determine an adequate ES level for an entire market, i.e.,~for considering the payoffs of multiple insurers. If the PELVE values of all insurers are equal, then a natural choice would be this PELVE value, i.e.,~an insurer does not need to adjust its actual capital reserve. However, in general, it is not possible to identify an ES level such that the capital reserve of each insurer remains unchanged, which encourages the need for further discussions. Here, we follow closely the idea behind PELVE, i.e.,~we suggest methods with the goal on limiting the changes in the capital reserves of the insurers. This is the most important criterion, because large reductions would increase the risk of future bankruptcies and large increases may not be tolerable by single insurers and furthermore, may threaten the financial stability, as an insurer must cover a large capital requirement within a short period of time. There is, of course, no unique way to achieve this. Our contribution consists in proposing and discussing several possible methods. 

First, we describe the situation of the regulator mathematically. To do so, for $n\in\mathbb{N}$, we denote the linear space of all $n$-dimensional integrable random variables by $L^1_n(\Omega,\mathcal{F},P)$ or $L^1_n$ for short. For brevity, we use the standard terminology of $[n]:=\{1,\dots,n\}$.

Our market consists of $n\in\mathbb{N}$ agents. Agent $i\in[n]$ faces an unknown future payoff $X_i\in L^1$. For brevity, we write $\mathbf{X}=(X_1,\dots, X_n)^{\intercal}$, i.e.,~$\mathbf{X}$ is a random vector. Furthermore, we assume that the prevailing regulatory framework relies on the VaR at level $\lambda\in(0,1)$. The regulator aims to replace this VaR with an ES at level $\alpha\in(0,1)$. For $n=1$, the natural choice is $\alpha = \pelve{\lambda}(X_1)\lambda$, ensuring that the agent does not need to adjust its capital reserve. The primary objective of this manuscript is to explore possible answers to the following question: 

\begin{center}
	\textit{How should the regulator decide on $\alpha$ in the situation of $n>1$ agents?}
\end{center}

The subsequent section presents a range of quantitative methods as candidate solutions and analyses their properties in terms of existence, uniqueness, and continuity. Moreover, the advantages and disadvantages of these methods are examined, with a summary given in Table~\ref{tab:comparisonMethods}.

%% file: sections/optimizationProblemsPelve.tex
\section{Addressing the regulator's problem}\label{sec:problemRegulator}

We now introduce the quantitative methods, motivated by the original concept of PELVE.

\subsection{Representative agent and Average-PELVE}\label{sec:pelve_based_measures}

One possibility is to determine $\alpha$ via a single PELVE, e.g.,~by using a benchmark payoff $Y$ and setting $\alpha = \Pi_{\lambda}(Y)\lambda$. For instance, if agent~$i$ is a representative agent on the market, then the regulator could set $\alpha = \pelve{\lambda}(X_i)\lambda$. While easy to explain, this approach can lead to significant changes in the capital reserve of agent $j\neq i$ when the PELVE values of agents $j$ and $i$ differ considerably. Hence, agent $j$ could feel discriminated. 

It is therefore more appropriate to account for the situation of all agents, i.e.,~to determine $\alpha$ based on multiple PELVE values. The first possibility that we suggest is the intuitive method of a weighted average of these PELVE values. Unequal weights can reflect different market participation among the agents.

\begin{defi}
	The Average-PELVE (A-PELVE) $\apelve{\lambda,\omega}:L^1_n\rightarrow [1,\lambda^{-1}]\cup\{\infty\}$ for $\lambda\in(0,1)$ and $\omega\in[0,1]^{n}$ with $\sum_{i=1}^{n}\omega_i=1$ is defined for all $\mathbf{X}\in L^1_n$ as $\apelve{\lambda,\omega}(\mathbf{X}) \defgl \sum_{i=1}^{n}\omega_i\pelve{\lambda}(X_i)$.
\end{defi}

\begin{rem}\label{rem:properties_APELVE}
	The A-PELVE is finite if and only if all PELVE values are finite, i.e.,~for all $i\in[n]$ it holds that $\pelve{\lambda}(X_i)<\infty$. Equivalent conditions for the latter are stated in~\cite[Proposition 1]{li_pelve_2023}. Next, assume that for all $i\in[n]$, the map $\symbolUpperQuantile{X_i}$ is continuous at $\lambda$ and there exists a sequence $(X_i^m)_m\subseteq L^1$ such that $X_i^m\rightarrow X_i$ with respect to the $L^1$-norm. By setting $\mathbf{X}^m=(X_1^m,\dots,X_n^m)^{\intercal}$, we obtain $\apelve{\lambda,\omega}(\mathbf{X}^m) \rightarrow \apelve{\lambda,\omega}(\mathbf{X})$ as $m\rightarrow\infty$ by applying~\cite[Theorem 2]{li_pelve_2023}.\footnote{\textcite[Theorem 2]{li_pelve_2023} require that $(X_i^m)_m\subseteq L^1$ is uniformly integrable and converges in distribution to $X_i\in L^1$. By Vitali's Convergence Theorem~\parencite[Theorem 21.4]{bauerMeasureIntegrationTheory2001} this holds iff $X_i^m\xrightarrow{L^1}X_i$, where $X_i\in L^1$.}
\end{rem} 

The A-PELVE takes the situation of all agents into account. However, as an arithmetic mean, the A-PELVE is not robust against outliers in the individual PELVE values. Furthermore, the capital reserve of a single agent can change drastically, if its PELVE is far away from the A-PELVE. Motivated by this last argument, we now propose two further methods. 

\subsection{Mean-Squared-Error and Worst-Case PELVE}\label{sec:mse_based_measures}

The next method minimizes the squared distances of changes in the capital reserves. This refers to the (weighted) mean-squared-error (MSE).

\begin{defi}\label{defi:mse_pelve}
	A map $\msepelve{\lambda,\omega}:L^1_n\rightarrow [1,\lambda^{-1}]$ is called an MSE-PELVE for $\lambda\in(0,1)$ and $\omega\in[0,1]^{n}$ with $\sum_{i=1}^{n}\omega_i=1$, if for all $\mathbf{X}\in L^1_n$ the value $\msepelve{\lambda,\omega}(\mathbf{X})$ minimizes
	\begin{align}\label{eq:mse}
		[1,\lambda^{-1}]\rightarrow[0,\infty),c\mapsto  \sqrt{\sum_{i=1}^{n}\omega_i\Big(\symbolExpectedShortfall{c\lambda}(X_i)-\symbolValueAtRisk{\lambda}(X_i)\Big)^2}.
	\end{align}
\end{defi}

The map in~\prettyref{eq:mse} measures the weighted squared distance of ES and VaR values. This aligns perfectly with our motivation to reduce the impact of a regulatory change on individual insurers, on average. It has another important implication, namely, that the usage of the MSE avoids cross-subsidization effects in the calculation of the new ES level. That is, a positive change in the reserve of an insurer cannot be  compensated by negative changes in the reserve of other insurers.

A-PELVE and MSE-PELVE are based on a weighting vector. These weights allow to regulate the impact of an insurer on the calculation of the new ES level. This allows us especially to reduce the impact of insurers who have not managed their risks in an adequate manner, which is reflected by a heavy left tail of the equity distribution. This is pointed out in the case study in~\prettyref{sec:caseStudy}.

The upcoming result shows that there always exists at least one MSE-PELVE.

\begin{prop}\label{prop:existenceMsePELVE}
	Given the situation in~\prettyref{defi:mse_pelve}, the map in~\prettyref{eq:mse} attains a minimum.
\end{prop}
 
Uniqueness of the minimum is not guaranteed, even if~\cite[Assumption 1]{li_pelve_2023} holds for all $X_1,\dots,X_n$, as the next example shows.
\begin{exam}\label{exam:uniquenessMSEpelve}
	We consider $n=2$ in~\prettyref{eq:mse}, i.e.,~two agents. Based on~\prettyref{thm:existenceOfEScurve}, we can choose a suitable function to describe the ES curve of a random variable. To do so, we choose functions $f,g:(0,1]\rightarrow \mathbb{R}$, weights $\omega_1,\omega_2\in(0,1)$ with $\omega_1+\omega_2 = 1$ and a VaR level $\lambda\in(0,1)$. To be precise, we choose $g(t) = 1-t^2$. With $t_g^{\ast}\in(0,1)$, we denote the value satisfying $g(t_g^{\ast}) = g(\lambda)+\lambda g^{\prime}(t)$. Furthermore, given a constant $c\in(0,\infty)$, let $u\in(t_g^{\ast},1]$ be a constant such that 
	\begin{align*}
		c - \frac{\omega_2}{\omega_1}\big(g(u)-g(\lambda)-\lambda g^{\prime}(\lambda)\big)^2>0.
	\end{align*}
	
	The map $g$ is twice differentiable, decreasing and concave. Hence, by~\prettyref{cor:existenceOfEScurve} there exists an integrable random variable such that its ES curve is equal to $g$. To define the function $f$, we use the following auxiliary function, defined for a parameter $d\in(0,1)$ and for all $t\in[t_g^{\ast},u]$:
	\begin{align*}
		\tilde{f}(t) = d + \sqrt{c - \frac{\omega_2}{\omega_1}(g(t)-g(\lambda)-\lambda g^{\prime}(\lambda))^2}.
	\end{align*}
	Finally, we choose a value $s\in(t_g^{\ast},u)$, a weight $w\in(0,1)$ and $\epsilon\in(0,1)$ such that $\lambda>\epsilon$. Then, we set $t_\lambda^{\text{mixed}}= w(\lambda-\epsilon) + (1-w)s$ and $a = t_\lambda^{\text{mixed}}(\tilde{f}^{\prime}(s)t_\lambda^{\text{mixed}} + \tilde{f}(s)-\tilde{f}^{\prime}(s)s-d)$.
	
	Then, for all $t\in(0,1]$ we set
	\begin{align*}
		f(t)=\begin{cases}
			\frac{a}{\lambda - \epsilon} + d, &t\in(0,\lambda-\epsilon),\\
			\frac{a}{t} + d, &t\in[\lambda-\epsilon, t_\lambda^{\text{mixed}}),\\
			\tilde{f}^{\prime}(s)t+\tilde{f}(s)-\tilde{f}^{\prime}(s)s, &t\in[t_\lambda^{\text{mixed}}, s),\\
			\tilde{f}(t), &t\in[s, u),\\
			\tilde{f}^{\prime}(u)t+\tilde{f}(u)-\tilde{f}^{\prime}(u)u, &t\in[u,1].
		\end{cases}
	\end{align*} 
	
	We use the following parameter values:
	\begin{align*}
		\lambda = \frac{1}{3},\ c = 0.1,\ d = 0.8,\ \omega_1 = 1-\omega_2 = 0.4,\ \epsilon = 0.05,\ s = 0.66,\ u = 0.74 ,\ w = 0.25.
	\end{align*}
	
	We illustrate the functions $f$ and $g$ by the solid lines in the left plot in~\prettyref{fig:nonUniquenessMSEpelve}. The dashed lines represent the almost everywhere defined functions $t\mapsto -f(t)-tf^{\prime}(t)$ and $t\mapsto -g(t)-tg^{\prime}(t)$, which are obviously non-decreasing. By definition of $f$, the corresponding a.e.~defined derivative of $-tf(t)$ is right-continuous. Furthermore, $f(0):=\lim\limits_{u\downarrow 0} f(u) = \frac{a}{\lambda-\epsilon}+d$ and in turn $\lim\limits_{t\downarrow 0} tf(t) = 0$. Also $f^{\prime}$ exists a.e.~and it is Lebesgue integrable such that $f(t) = f(0)+\int_{0}^{t}f^{\prime}(v)\diff v$, which is equivalent to $f$ being absolutely continuous. This also shows that $t\mapsto tf(t)$ is absolutely continuous. Hence, all conditions in~\prettyref{thm:existenceOfEScurve} are satisfied and hence, there exists an integrable random variable such that its ES curve is exactly given by $f$.\footnote{From~\prettyref{fig:nonUniquenessMSEpelve} it is also not difficult to deduce that Assumption 1 in~\textcite{li_pelve_2023} is satisfied.} 
	
	In the right plot in~\prettyref{fig:nonUniquenessMSEpelve}, we illustrate the following function: 
	\begin{align*}
		h:(0,1]\rightarrow \mathbb{R},t\mapsto \omega_1(f(t)-f(\lambda)-\lambda f^{\prime}(\lambda))^2+ \omega_2(g(t)-g(\lambda)-\lambda g^{\prime}(\lambda))^2.
	\end{align*}
	Note, the values $f(\lambda)+\lambda f^{\prime}(\lambda)$ and $g(\lambda)+\lambda g^{\prime}(\lambda)$ are the VaR values with respect to the ES curves $f$ and $g$. We see that function $h$ has no unique minimum. It attains its minimum value all over the interval $[s,u]$, represented by the vertical dashed lines.
	
	\begin{figure}
		\begin{subfigure}[b]{0.48\textwidth}
			\includegraphics[width=\linewidth]{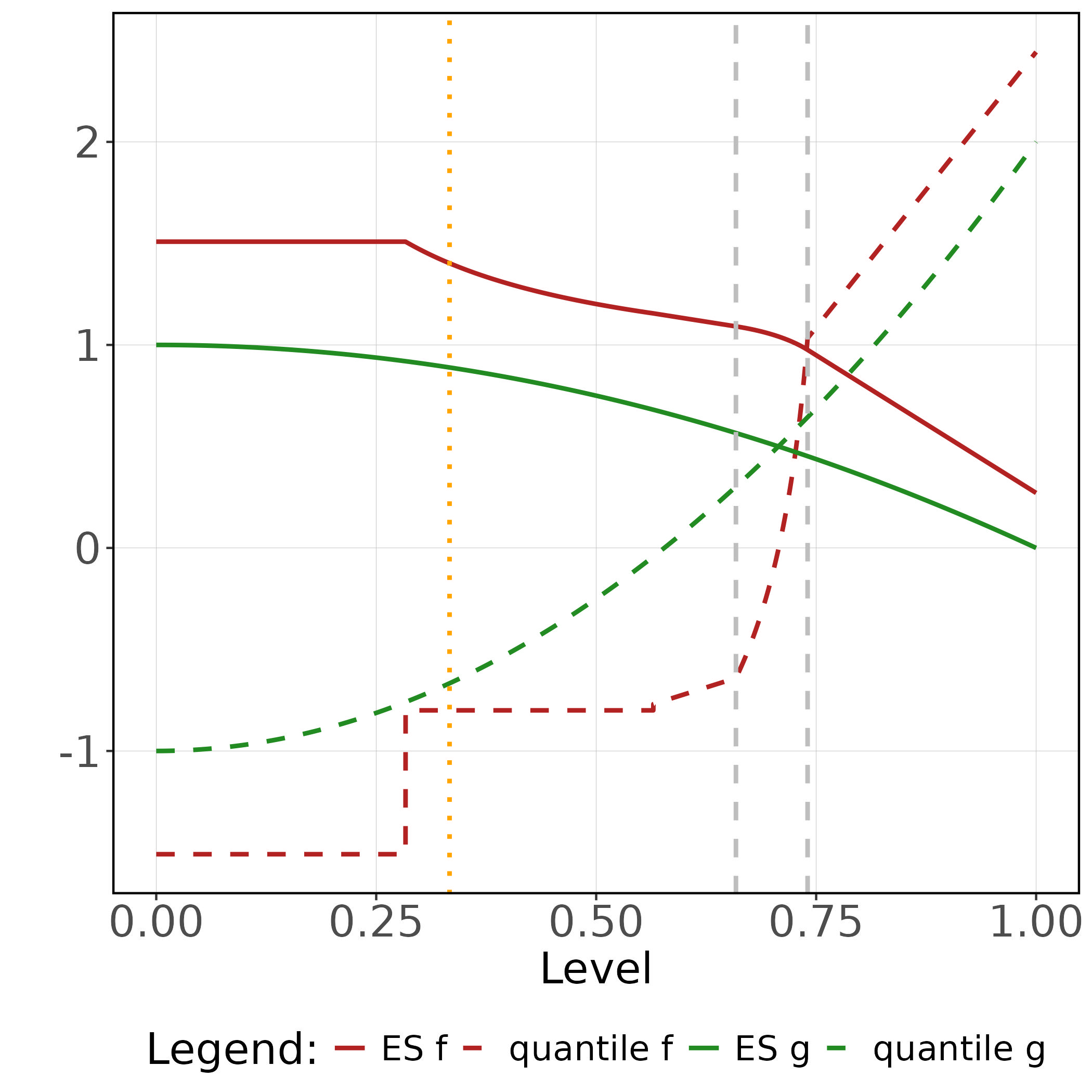}
		\end{subfigure}
		\hfill
		\begin{subfigure}[b]{0.48\textwidth}
			\includegraphics[width=\linewidth]{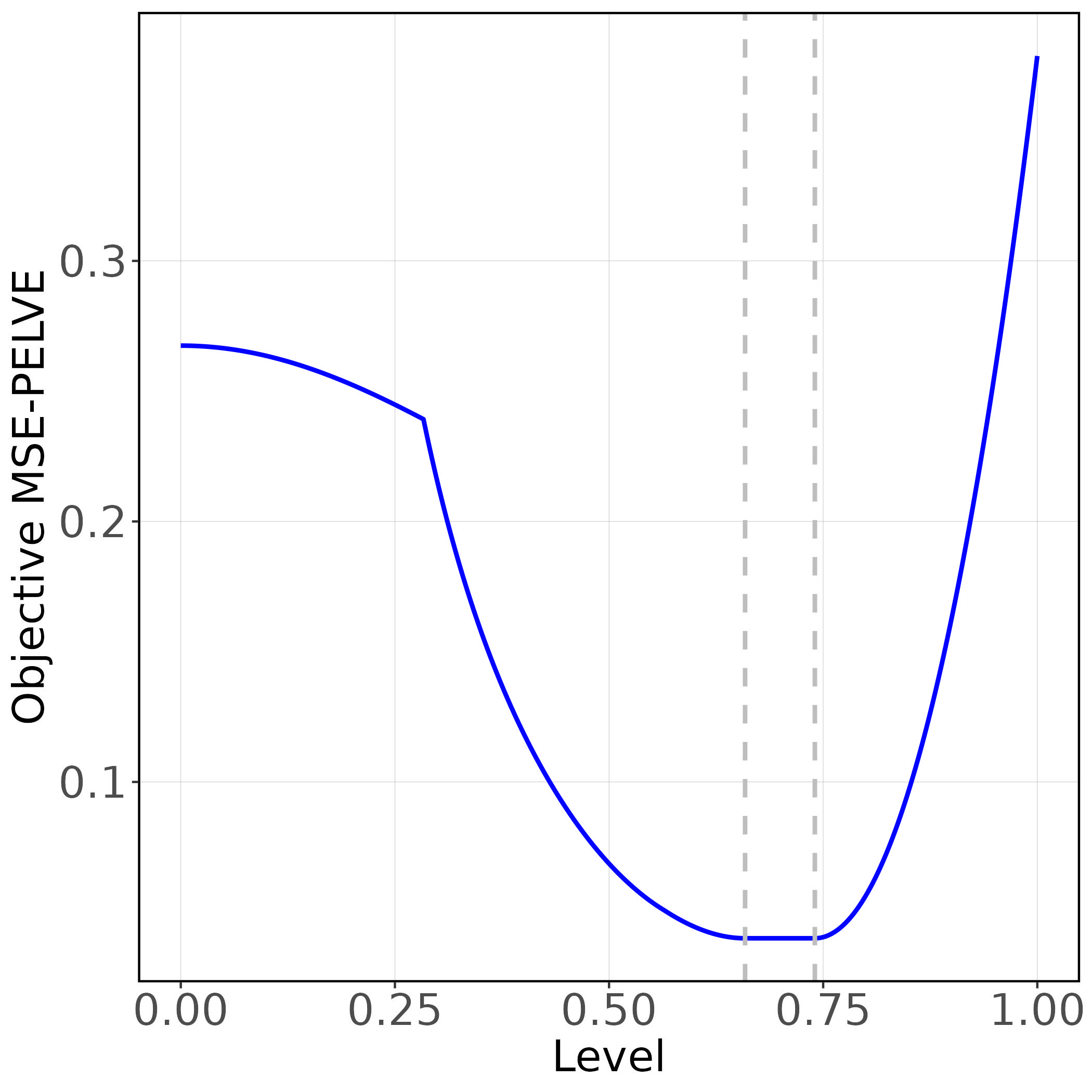}
		\end{subfigure}
		\captionsetup{font=footnotesize}
		\caption{Solid lines in the left plot are the functions $f$ and $g$. The red and green dashed lines are the a.e.~defined maps $t\mapsto -f(t)-tf^{\prime}(t)$ and $t\mapsto -g(t)-tg^{\prime}(t)$. The orange dotted vertical line is the chosen level $\lambda=\frac{1}{3}$. The gray vertical lines indicate the interval over which the objective of the MSE-PELVE remains constant; compare with the right-hand side.}
		\label{fig:nonUniquenessMSEpelve}
	\end{figure}
\end{exam}

Next, we develop a continuity result. To do so, we adopt the same assumptions as used in~\cite[Theorem 2]{li_pelve_2023}. In addition, we assume that in the limit, the MSE-PELVE is unique. For its proof, we denote by $|.|_d$ the Euclidean norm on $\mathbb{R}^d$.
\begin{theorem}\label{thm:convergenceMSE}
	In the situation of~\prettyref{defi:mse_pelve} fix $\mathbf{X}\in L^1_n$, a sequence $(\mathbf{X}^m)_m\subseteq L^1_n$, and assume the following:
	\begin{enumerate}[(a)]
		\item For all $i\in[n]$ let $X_i^m\xrightarrow{L^1}X_i$ as $m\rightarrow \infty$ and let the map $\symbolUpperQuantile{X_i}$ be continuous at $\lambda$;
		\item The minimum of~\prettyref{eq:mse} regarding $\mathbf{X}$ is unique.
	\end{enumerate}
	Then, it holds that
	\begin{align*}
		\lim\limits_{m\rightarrow\infty}\msepelve{\lambda,\omega}(\mathbf{X}^m) = \msepelve{\lambda,\omega}(\mathbf{X}).
	\end{align*}
\end{theorem}

As for the A-PELVE, we obtain that the MSE-PELVE takes the situation of all agents into account. However, in contrast to the A-PELVE, it leads to smaller overall changes in the capital reserves. Individually, i.e.,~per agent, MSE- or A-PELVE can lead to the smallest change in an individual capital reserve. For statements in this direction, we refer to the case study in~\prettyref{sec:caseStudy}. As a disadvantage, we can say that the objective in~\prettyref{eq:mse} is not robust against outliers. Also, it does not differ between an increase or decrease in the capital reserve of an agent. 

To conclude this subsection, we introduce two methods, which act as lower and upper bounds for the other methods. The first method leads for each agent to a decrease in the capital reserve. This can be interpreted as the best case for the insurers, since none of them needs to increase its reserve, but as the worst case for the regulator, as every insurer can reduce its reserve, making a future bankruptcy more likely. The second method vice versa leads to an increase in the capital reserve of each agent, for which we speak of the best-case for the regulator.

\begin{defi}\label{defi:wc_pelve}
	The Worst-Case PELVE (WC-PELVE) $\wcpelve{\lambda}:L^1_n\rightarrow [1,\lambda^{-1}]\cup\{\infty\}$ for $\lambda\in(0,1)$ is defined for all $\mathbf{X}\in L^1_n$ as
	\begin{align}\label{eq:wcPelve}
		\wcpelve{\lambda}(\mathbf{X}) = \max\{\pelve{\lambda}(X_1),\dots,\pelve{\lambda}(X_n)\}.
	\end{align}	
	Analogously, the Best-Case PELVE (BC-PELVE) $\bcpelve{\lambda}:L^1_n\rightarrow [1,\lambda^{-1}]\cup\{\infty\}$ for $\lambda\in(0,1)$ is defined for all $\mathbf{X}\in L^1_n$ as $\bcpelve{\lambda}(\mathbf{X}) = \min\{\pelve{\lambda}(X_1),\dots,\pelve{\lambda}(X_n)\}$.
\end{defi} 

\begin{rem}
	The WC-PELVE (resp.~BC-PELVE) exists iff every (resp.~at least one) PELVE exists. In the same situation as in~\prettyref{rem:properties_APELVE}, we obtain that $\wcpelve{\lambda}(\mathbf{X}^{m}) \rightarrow \wcpelve{\lambda}(\mathbf{X})$ (resp.~$\bcpelve{\lambda}(\mathbf{X}^{m}) \rightarrow \bcpelve{\lambda}(\mathbf{X})$) as $m\rightarrow\infty$.
\end{rem}

The WC-PELVE considers the situations of all agents and distinguishes between an increase or decrease in the corresponding capital reserves in the sense that it results in smaller capital requirements for all agents, because it holds that
$$\wcpelve{\lambda}(\mathbf{X}) = \inf\{c\in[1,\lambda^{-1}]\,|\,\forall i\in[n]:\symbolExpectedShortfall{c\lambda}(X_i)\leq \symbolValueAtRisk{\lambda}(X_i)\}.$$
Such a decrease can also have disadvantages, namely, for some agents we expect significant smaller capital reserves. In contrast, the BC-PELVE ensures an increase in the reserves of all insurers, which has the drawback that a substantial increase in the reserve is hard to put into practice. From this discussion, we should understand WC- and BC-PELVE only as lower and upper bounds for the other Multi-PELVE methods.

\subsection{Systemic PELVE}\label{sec:def_systemic_pelve}

In this section, we present a method based on systemic risk before and after a transition from VaR to ES. To do so, let $g$ be either given as the identity $g(x)=x$ or as $g(x)=\max\{0,x\}$. Then, for $\lambda\in(0,1)$ and $\mathbf{X}\in L^1_n$ we define
\begin{align}\label{eq:systemic_pelve}
	\syspelve{\lambda,g}(\mathbf{X})\defgl \inf\{c\in[1,\lambda^{-1}]\,|\,\rho^{\symbolExpectedShortfall{c\lambda},g}(\mathbf{X})\leq\rho^{\symbolValueAtRisk{\lambda},g}(\mathbf{X})\},
\end{align}
where $\rho^{\symbolExpectedShortfall{\beta},g}$, respectively  $\rho^{\symbolValueAtRisk{\beta},g}$, is a systemic risk measure based on ES, respectively VaR. There are different ways to specify these systemic risk measures. We use
\begin{align}\label{eq:systemic_risk_biagini}
	\rho^{\nu,g}(\mathbf{X}) &\defgl \sum_{i=1}^{n}g(\nu(X_i)),
\end{align}
where $\nu\in\{\symbolOperatorExpectedShortfall_{\beta},\symbolOperatorValueAtRisk_{\beta}\}$. Equation~\prettyref{eq:systemic_risk_biagini} is a systemic risk measure in the form of~\cite[Equation (2.13)]{biagini_unified_2019}.\footnote{This can be seen by noting that $\rho^{\nu,g}(\mathbf{X})=\inf\left\{\sum_{i=1}^{n}g(m_i)\,\middle|\,\forall i\in[n],m_i\in\mathbb{R}:\nu(X_i+m_i)\leq 0\right\}$.} We call the map in~\prettyref{eq:systemic_pelve} a Systemic PELVE (Sys-PELVE). Furthermore, note that choosing $g(x) = \max\{0,x\}$ avoids cross-subsidization among agents, i.e.,~positive values cannot be compensated by negative ones.

\begin{rem}
	By the same reasoning as in~\cite[Proposition 1]{li_pelve_2023}, we have $\syspelve{\lambda,g}(\mathbf{X})<\infty$ if and only if $\sum_{i=1}^{n}g(\expectation{}{-X_i}) \leq \sum_{i=1}^{n}g(\valueAtRisk{\lambda}{X_i})$. For $g(x) = x$, if it holds that $\syspelve{\lambda,g}(\mathbf{X})<\infty$ and if for each $i\in[n]$, $p\mapsto \valueAtRisk{p}{X_i}$ is not being constant on $(0,\lambda]$, then there exists a unique $c\in[1,\lambda^{-1}]$ such that $\sum_{i=1}^{n}\expectedShortfall{c\lambda}{X_i} \leq \sum_{i=1}^{n}\valueAtRisk{\lambda}{X_i}$.
\end{rem}

Similar to~\prettyref{thm:convergenceMSE}, we state a continuity property of Sys-PELVE. 
\begin{prop}\label{prop:convergenceSys}
	Let $\lambda\in(0,1)$, $\mathbf{X}\in L^1_n$, $(\mathbf{X}^m)_m\subseteq L^1_n$ and $g$ either be given as $g(x)=x$ or as $g(x) = \max\{x,0\}$. Assume the following:
	\begin{enumerate}[(a)]
		\item For all $i\in[n]$ let $X_i^m\xrightarrow{L^1}X_i$ as $m\rightarrow \infty$ and let the map $\symbolUpperQuantile{X_i}$ be continuous at $\lambda$;
		\item For all $i\in[n]$, the map $p\mapsto \valueAtRisk{p}{X_i}$ is not constant on $(0,\lambda]$;
		\item It holds that $\sum_{i=1}^{n}g(\expectation{}{X_i})<\sum_{i=1}^{n}g(\valueAtRisk{\lambda}{X_i})$. 
	\end{enumerate}
	Then it holds that $$	\lim\limits_{m\rightarrow\infty}\syspelve{\lambda,g}(\mathbf{X}^m) = \syspelve{\lambda,g}(\mathbf{X}).$$
\end{prop}

\begin{rem}
	Assumption (a) is the same as the one  in~\prettyref{thm:convergenceMSE}. Assumptions (b) and (c) adapt~\textcite[Assumption 1]{li_pelve_2023} to the case of the Sys-PELVE. They differ from Assumption (b) in~\prettyref{thm:convergenceMSE} (\textit{uniqueness} of minimizer of the objective function in case of the MSE-PELVE). However, they fulfill the same purpose, that is, providing a \textit{unique} value meeting the constraint in case of the Sys-PELVE. Furthermore,~\prettyref{exam:uniquenessMSEpelve} indicates that Assumptions (b) and (c) in~\prettyref{prop:convergenceSys} adapted to the MSE-PELVE case are not stronger than Assumption (b) in~\prettyref{thm:convergenceMSE}.
\end{rem}

The systemic risk measures used in the Sys-PELVE are of relevance for the regulator as they describe the overall situation of the financial system. The Sys-PELVE then guarantees that the summation of capital reserves in the network of agents is reduced after the transition from VaR to ES. Nonetheless, since the systemic risk measures aggregate the VaR or ES values of all agents, it is possible that the capital reserve of an individual agent increases after moving from VaR to ES.

All discussed pros and cons of the presented methods are summarized in~\prettyref{tab:comparisonMethods}.

\begin{table}
	\centering
	\setlength{\leftmargini}{0.4cm}
	\begin{tabular}{|l|m{6cm}|m{6cm}|}
		\hline
		& \textbf{Favorable features} & \textbf{Unfavorable features}\\
		\hline\hline
		PELVE $\pelve{\lambda}(X_i)$ &\begin{itemize}
				\item readily understandable 
				\item easy to calculate
				\item does not change the solvency probability of agent $i$
				\item suitable for a homogeneous market of agents
			\end{itemize} & \begin{itemize}
			\item difficult to choose benchmark agent $i$ in heterogeneous markets
			\item huge differences in capital reserve for agent $j\neq i$ possible
			\item agent $j\neq i$ might feel disregarded
			\item does not differ between increase or decrease in capital reserve 
			\end{itemize} \\
		\hline
		A-PELVE $\apelve{\lambda,\omega}(\mathbf{X})$ &\begin{itemize}
			\item readily understandable
			\item easy to calculate
			\item considers all agents’ situations, including agents’ market shares $\omega$
			\item easy calculation
		\end{itemize} & \begin{itemize}
			\item huge differences in capital reserve, if agent $i$'s PELVE differs significantly from the mean
			\item does not differ between increase or decrease in capital reserve
			\item not robust against outliers
		\end{itemize}\\
		\hline
		MSE-PELVE $\msepelve{\lambda,\omega}(\mathbf{X})$ & \begin{itemize}
			\item considers all agents’ situations, including  agents’ market shares $\omega$
			\item leads to smaller changes in the capital reserves over all agents than previous methods
		\end{itemize}
		& \begin{itemize}
			\item does not differ between increase or decrease in capital reserve
			\item objective function is not robust against outliers
		\end{itemize}\\
		\hline
		WC-PELVE $\wcpelve{\lambda}(\mathbf{X})$ & \begin{itemize}
			\item considers all agents’ situations
			\item differs between increase or decrease in capital reserve
			\item positive from an agent's perspective: results in a smaller capital reserve after the transition from VaR to ES
		\end{itemize} & 
		\begin{itemize}
			\item can lead to a significant lower capital reserve for an agent, which is  undesirable from the regulator's perspective
			\item increases the probability that at least one agent will become insolvent in the future
		\end{itemize} \\
		\hline
		Sys-PELVE $\syspelve{\lambda,g}(\mathbf{X})$ & \begin{itemize}
				\item considers all agents’ situations
				\item based on systemic risk measures, relevant for a regulator
				\item unlike MSE, guaranteeing that needed reserve in the network (sum of ES or VaR values) is reduced by moving from VaR to ES
				\item $g(x)=\max\{x,0\}$ prevents cross-subsidization between agents
			\end{itemize} & 
			\begin{itemize}
				\item the constraint only works on the sum of  capital reserves of all agents. Hence, it is possible that the capital reserve of a single agent increases after the transition
				\item $g(x)=x$ allows for cross-subsidization between agents
			\end{itemize}\\
			\hline
	\end{tabular}
	\captionsetup{font=footnotesize}
	\caption{Features of the suggested methods.}
	\label{tab:comparisonMethods}
\end{table}

%% file: sections/exampleNormalDistributions.tex
\section{Common distributional assumptions}\label{sec:concreteTheoreticalDistributions}

This section considers payoff vectors following specific multivariate distributions.

\subsection{Multivariate elliptical distributions}\label{sec:examNormal}

We start by calculating the Multi-PELVE methods for a typical class of models often used in finance and insurance, namely elliptical distributions. They contain multivariate normal and multivariate $t$-distributions as special cases. 

\begin{defi}
	Assume an $n$-dimensional random vector $\mathbf{X}$, a vector $\mu\in\mathbb{R}^n$, a positive semi-definite symmetric matrix $\Sigma\in\mathbb{R}^{n\times n}$ and a map $\phi:[0,\infty)\rightarrow \mathbb{R}$. If the characteristic function $\varphi_{\mathbf{X}-\mu}$ of $\mathbf{X}-\mu$ is of the form $\varphi_{\mathbf{X}-\mu}(t) = \phi(t^{\intercal} \Sigma t)$, then $\mathbf{X}$ is elliptically distributed, denoted as $\mathbf{X}\sim E_{n}(\mu,\Sigma,\phi)$.  
\end{defi}

We write $I_n$ for the $n$-dimensional identity matrix and denote its $i$-th column by $\unitVec_i$. 

\begin{rem}\label{rem:cholesky_elliptical}
	Let $A\in\realNumbers^{n\times n}$ with $\Sigma=A A^{\intercal}$. If $\Sigma$ is positive definite, then $A$ can be chosen as the Cholesky decomposition of $\Sigma$. By~\cite[Lemma 3.1]{hult_multivariate_2002}, for a random variable $\mathbf{Y}\sim E_{n}(0,I_n,\phi)$ it holds that $\mathbf{X}\stackrel{d}{=}\mu + A \mathbf{Y}$. 
\end{rem}

The next result states that for elliptical distributions, in most of the cases, the Multi-PELVE methods simplify to a PELVE for a single elliptically distributed random variable. To do so, we denote by $\var(X)$ the variance of a random variable $X$.

\begin{prop}\label{prop:pelveForElliptical}
	Let $\lambda\in(0,1)$ and $\omega\in[0,1]^{n}$ with $\sum_{i=1}^{n}\omega_i = 1$. Further, let $\mathbf{X}\sim E_{n}(\mu,\Sigma,\phi)$ and $Z\sim E_{1}(0,1,\phi)$. The following holds:
	\begin{enumerate}[(i)]
		\item For all $i\in[n]$ we have $\pelve{\lambda}(X_i) = \pelve{\lambda}(Z)$. 
		\item We have $\apelve{\lambda,\omega}(\mathbf{X}) =\wcpelve{\lambda}(\mathbf{X}) = \syspelve{\lambda,g}(\mathbf{X}) = \pelve{\lambda}(Z)$, where $g(x) = x$.
		\item The following map is an MSE-PELVE:
		\begin{align*}
			\msepelve{\lambda,\omega}(\mathbf{X})=\begin{cases}\pelve{\lambda}(Z), &\expectation{}{-Z}\leq \valueAtRisk{\lambda}{Z},\\
				\lambda^{-1}, &\text{otherwise}.
				\end{cases} 
		\end{align*}
		\item For $g(x)=\max\{0,x\}$ we have $\syspelve{\lambda,g}(\mathbf{X}) \leq  \pelve{\lambda}(Z)$. Further, if $\expectation{}{-Z}\leq\valueAtRisk{\lambda}{Z}$ and the map $(0,1)\rightarrow\mathbb{R},\alpha\mapsto \expectedShortfall{\alpha}{Z}$ is strictly decreasing, then $\syspelve{\lambda,g}(\mathbf{X}) = \pelve{\lambda}(Z)$ if and only if there exists $i\in[n]$ such that $\sigma_i^2 > 0$ and $\frac{\mu_i}{\sigma_i} \leq \valueAtRisk{\lambda}{Z}$, where $\sigma_i^2 := \var(X_i)$.\footnote{Note that $\var(X_i) = \unitVec_i^{\intercal} \Sigma \unitVec_i$. Furthermore, $\var(X_i)>0$ is equivalent to requiring that $X_i$ is not a.s.~constant.}
	\end{enumerate}
\end{prop}

For a multivariate normal distribution, it holds that $\phi(x) = e^{-\frac{1}{2}x}$. Hence, $\pelve{\lambda}(Z)$ is independent of the mean vector and the covariance matrix. 
For the situation of a multivariate t-distribution note the following: \textcite[Theorem 3.1]{hult_multivariate_2002} state that $\mathbf{X} \stackrel{d}{=} \mu + RA \mathbf{U}$, where $R$ is a non-negative random variable independent of $U$. The latter is a $\text{rank}(\Sigma)$-dimensional random vector, which is uniformly distributed on the corresponding unit-sphere and $A$ is an $n\times \text{rank}(\Sigma)$-matrix with $AA^{\intercal} = \Sigma$. Now, $\mathbf{X}$ is multivariate t-distributed iff $R^{2} / n\sim F(n, \nu)$, where $F(n,\nu)$ is an $F$-distribution with $n$ and $\nu$ degrees of freedom. So, $\pelve{\lambda}(Z)$ only depends on the degrees of freedom of the multivariate t-distribution, hidden in the function $\phi$.

As a conclusion of~\prettyref{prop:pelveForElliptical}, most of the suggested methods lead to the PELVE of a one-dimensional elliptically distributed random variable. Only the Sys-PELVE can be strictly less than $\pelve{\lambda}(Z)$. The following example illustrates this case. Here, a normal distribution with mean $\mu$ and variance $\sigma^2$ is denoted by $\text{N}(\mu,\sigma^2)$; the density and cumulative distribution function (CDF) of a standard normal random variable are denoted by $\varphi$ and $\Phi$, respectively.

\begin{exam}	
	Let $n=1$ and $X\sim \text{N}(\mu,\sigma^{2})$. In particular, $Z$ is assumed to be standard normally distributed. We choose the values $\mu = 0.75$ and $\sigma = 0.4$, as well as the level $\lambda = 0.05$. For this specification it holds that $\valueAtRisk{\lambda}{X}<0$. Thus, to calculate $\syspelve{\lambda,g}(X)$ for $g(x)=\max\{x,0\}$ we have to find $c^{\star}$ such that $\expectedShortfall{c^{\star}\lambda}{Z}= \frac{\mu}{\sigma}$. The value $\expectedShortfall{c\lambda}{Z}=\frac{\varphi(\Phi^{-1}(c\lambda))}{c\lambda}$ as a function of $c$ is illustrated in~\prettyref{fig:example_sysPELVE}. Moreover, the Sys-PELVE is indicated by the solid blue line, while the PELVE $\pelve{\lambda}(X) = \pelve{\lambda}(Z)$ is indicated by the solid red line. Hence, we get $\syspelve{\lambda,g}(X)<\pelve{\lambda}(Z)$.
\end{exam}

\begin{figure}[htb]
	\includegraphics[width=0.4\linewidth]{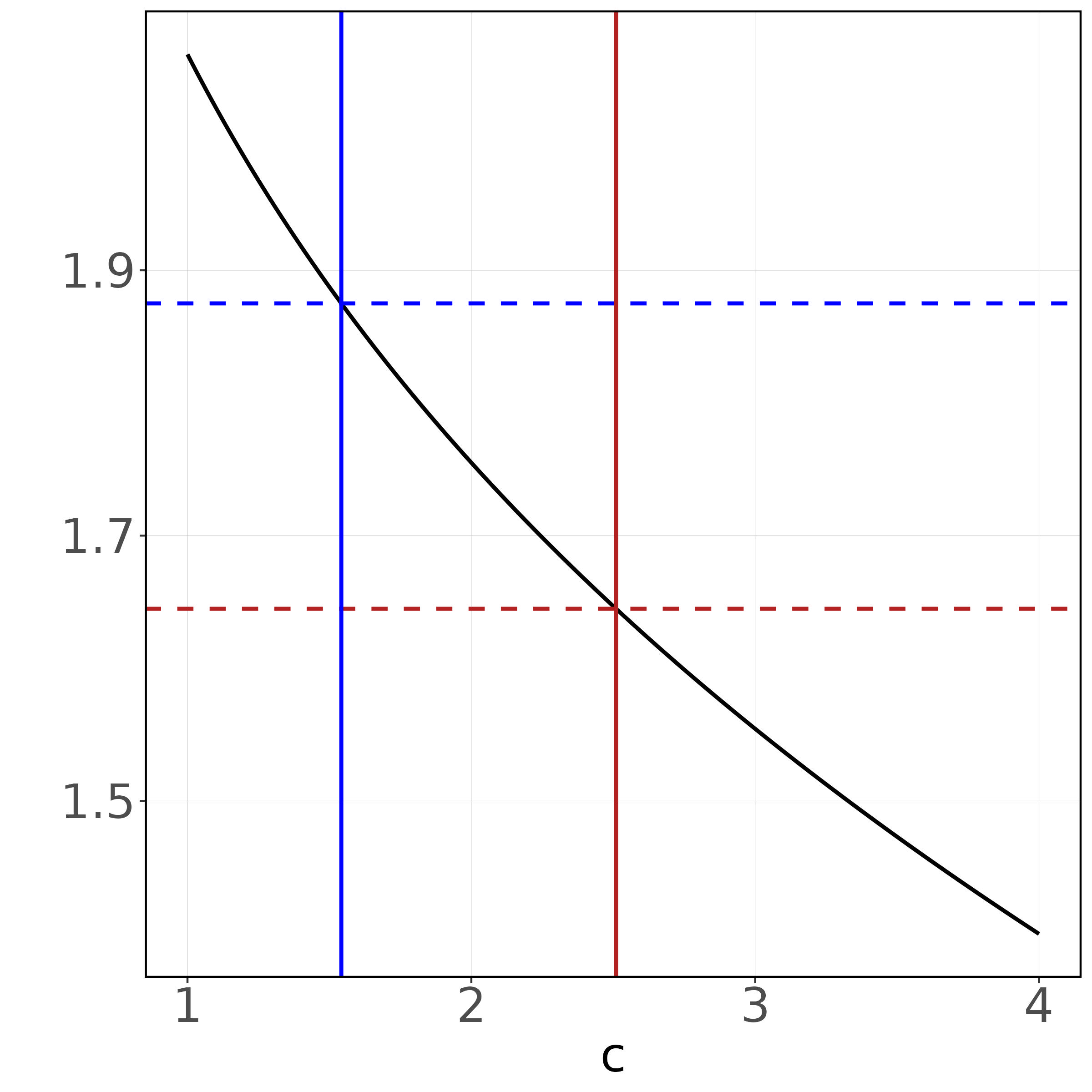}
	\captionsetup{font=footnotesize}
	\caption{Graph of the function $c\mapsto \frac{\varphi(\Phi^{-1}(c\lambda))}{c\lambda}$, where $\lambda=0.05$. The solid blue line is $\syspelve{\lambda,g}(X)$ for $X\sim \text{N}(\mu,\sigma)$ and $g(x)=\max\{0,x\}$, where $\mu = 0.75$ and $\sigma = 0.4$. The blue dashed line refers to the value $\frac{\mu}{\sigma}$. The red solid line is $\pelve{\lambda}(X)$, while the red dashed line refers to the value $-\Phi^{-1}(\lambda)$.}
	\label{fig:example_sysPELVE}
\end{figure} 

\subsection{Multivariate regularly varying distributions}\label{sec:heavyTails}

Next, we consider the heavy tailed case via multivariate regularly varying (MVR) distributions. To do so, we first recall the definition of a regularly varying random variable.

\begin{defi}
	A one-dimensional random variable $X$ has a regularly varying right-tail with tail index $\gamma>0$, if for all $t>0$ it holds that 
	\begin{align*}
		\lim\limits_{x\rightarrow\infty}\frac{P(X>tx)}{P(X>x)} = t^{-\gamma}.
	\end{align*}
	We write $X\in\text{RV}_{\gamma}$.
\end{defi}

For the definition of MVR models, we denote the standard $L^1_n$-norm by $\lVert.\rVert_{L^1_n}$. The boundary of a set $A\subseteq \mathbb{R}^n$ with respect to the Euclidean norm is written as $\partial A$.

\begin{defi}
	An $n$-dimensional random vector $\mathbf{X}$ has a multivariate regularly varying right-tail, if there exists $\gamma>0$ and a Borel probability measure $\Psi$ on the unit sphere $\mathbb{S}^n\defgl \{s\in\mathbb{R}^n\,|\,|s|_n = 1 \}$ such that for any $t>0$ and any Borel set $S\subseteq \mathbb{S}^n$ with $\Psi(\partial S) = 0$ it holds that
	\begin{align*}
		\lim\limits_{x\rightarrow\infty}\frac{P\left(\lVert \mathbf{X}\rVert_{L^1_d}>tx, \frac{\mathbf{X}}{\lVert \mathbf{X}\rVert_{L^1_n}}\in S\right)}{P(\lVert \mathbf{X}\rVert_{L^1_n}>x)} = t^{-\gamma}\Psi(S).
	\end{align*}
	In this situation, $\gamma$ is called the tail index of $\mathbf{X}$ and $\Psi$ the spectral measure of $\mathbf{X}$. The set of all $n$-dimensional random vectors with tail index $\gamma$ and spectral measure $\Psi$ is denoted by $\text{MVR}^n_{\gamma}(\Psi)$.
\end{defi}

We study the asymptotic behavior under an MVR model as the VaR level tends to zero. This aligns with the considerations for the classical PELVE in~\textcite{li_pelve_2023}, and can be seen as the counterpart of our~\prettyref{prop:pelveForElliptical} in the presence of heavy tails.

\begin{theorem}\label{thm:pelveForMVR}
	Let $\lambda\in(0,1)$ and $\omega\in[0,1]^{n}$ with $\sum_{i=1}^{n}\omega_i = 1$. Let $-\mathbf{X}\in \text{MVR}^n_{\gamma}(\Psi)$ with $\gamma>1$ and $\Psi$ satisfies $\Psi(\mathbb{S}^n\cap(0,\infty)^n) > 0$ (non-degeneracy condition). The following holds:
	\begin{enumerate}[(i)]
		\item For all $i\in[n]$ we have $\lim\limits_{\lambda\downarrow 0}\pelve{\lambda}(X_i) =\left(\frac{\gamma}{\gamma - 1}\right)^{\gamma}$. 
		\item We have $$\lim\limits_{\lambda\downarrow 0}\apelve{\lambda,\omega}(\mathbf{X}) =\lim\limits_{\lambda\downarrow 0}\wcpelve{\lambda}(\mathbf{X}) = \lim\limits_{\lambda\downarrow 0}\syspelve{\lambda,g}(\mathbf{X}) = \left(\frac{\gamma}{\gamma - 1}\right)^{\gamma},$$ where $g(x) = x$ or $g(x) = \max\{x,0\}$.
		\item For each MSE-PELVE we have $\lim\limits_{\lambda\downarrow 0}\msepelve{\lambda,\omega}(\mathbf{X})= \left(\frac{\gamma}{\gamma - 1}\right)^{\gamma}.$
	\end{enumerate}
\end{theorem}

The assumption $\gamma>1$ ensures the application of~\cite[Theorem A3.6 (b)]{embrechts_modelling_1997}. For $n=1$, this is equivalent to $-X$ being in the maximum domain of attraction of the Fr\'{e}chet distribution with parameter $\gamma>1$, see~\cite[Theorem 3.3.7]{embrechts_modelling_1997}. For instance, a Pareto distribution with shape parameter $\gamma >1$, i.e.,~finite mean, lies in this class.

The limits in~\prettyref{thm:pelveForMVR}, given by $\left(\frac{\gamma}{\gamma - 1}\right)^{\gamma}$, are all equal. This implies that, in the limiting case, the choice of the Multi-PELVE method is insignificant. Moreover, the value $\left(\frac{\gamma}{\gamma - 1}\right)^{\gamma}$ is equal to the PELVE of a Pareto distributed random variable with parameter $\gamma>1$. As $\gamma\rightarrow \infty$ this value converges to Euler's number $e\approx 2.718$. For $\gamma>1$, $\left(\frac{\gamma}{\gamma - 1}\right)^{\gamma}$ exceeds $e$. Hence,~\prettyref{thm:pelveForMVR} strongly indicates that the observation in~\cite{li_pelve_2023} -- that $e$ separates heavy- from light-tailed phenomena -- remains valid for multivariate extensions of PELVE.\newline

\noindent\fbox{\parbox{\textwidth}{\textit{Important consequence of~\prettyref{sec:concreteTheoreticalDistributions}:} The situations of multivariate normal and multivariate regularly distributions refer to situations of homogeneous market structures, referring to the fact that marginal distributions are out of the same family (at least in the tail). \prettyref{prop:pelveForElliptical} and~\prettyref{thm:pelveForMVR} indicate that in such a situation (the one of a homogeneous market), it is enough to use a PELVE with respect to a suitable random variable.}} \newline

This finding is also confirmed by the case study in the next section. However, we also analyze the situation of an inhomogeneous market, in which it will turn out that the application of Multi-PELVE methods cannot be reduced to the application of the PELVE from~\prettyref{defi:pelve}.

%% file: sections/caseStudy.tex
\section{Numerical case study}\label{sec:caseStudy}

We consider six insurers as agents, where the random vector $\mathbf{X}$ represents their future equity capitals. Future balance sheets are modeled following~\textcite{laudage_combining_2022}. Unlike their framework, assets consist of the sum of a constant and a payoff from a Black-Scholes model based on a constant portfolio process instead of a buy-and-hold strategy. Liabilities follow a gamma distribution, as in~\cite{laudage_combining_2022}, or heavier-tailed lognormal or generalized Pareto distributions (GPDs).

\subsection{Chosen distributions}

We model the future equity capital of insurer $i$ as $X_i = Y_i - Z_i$, where $Y_i$ and $Z_i$ denote the insurer's assets and liabilities, respectively. Liquid assets are modeled as payoffs stemming from a Black-Scholes market, while the non-liquid assets are assumed to be constant (e.g.,~real estates). Liabilities follow distributions commonly applied in non-life insurance -- gamma, lognormal, and GPD -- for total claim amounts.

For the liquid assets, each insurer invests in a common stock and an idiosyncratic stock. These stocks are modeled as geometric Brownian motions. We assume that an insurer invests a predefined fraction of the total wealth into the stocks, i.e.,~we restrict attention to constant portfolio processes. Note that the optimal solution of Merton's portfolio problem~\parencite{merton_optimum_1971} is a constant portfolio process. 

A constant portfolio process is a two-dimensional vector in $\realNumbers^{2}$, i.e.,~$\pi^i = (\pi^i_1,\pi^i_2)^{\intercal}\in\mathbb{R}^2$, where $\pi^i_j$ is the fraction of initial wealth invested in stock $j$. Let the initial liquid asset value be $x_0^i>0$ and assume that the investment horizon is one year. For each insurer, we need a two-dimensional standard Brownian motion $\mathbf{W}_t^i = (B_t,W_t^i)^{\intercal}$, where $B_t$ and $W_t^i$ are the Brownian motions for the common stock and the idiosyncratic stock of insurer $i$, respectively. The common interest rate is $r\in\mathbb{R}$, the drift vector is $b^i=(b,\mu^i)^{\intercal}$ and the covariance matrix is $\Sigma^{i} = \text{diag}(\sigma, \sigma^{i})$. Then, the wealth process of insurer $i$ at time $t\in[0,1]$ is
\begin{align*}
	X^{x_0^i,\pi^i}_t = x_0^i\exp\left(\left((\pi^i)^{\intercal}(b^i-r\mathbf{1}) + r -\tfrac{\left\lVert(\pi^i)^{\intercal}\Sigma^i\right\rVert^2}{2}\right)t + (\pi^i)^{\intercal}\Sigma^i \mathbf{W}_t^i\right),
\end{align*}
where $\mathbf{1} = (1 \dots 1)^{\intercal}$. We choose the same parameter values for all insurers:
\begin{align*}
	&r = 0.01,\quad b^i = (0.04,0.06)^{\intercal},\quad \Sigma^i = \begin{pmatrix}
		0.2 & 0\\
		0 & 0.4
	\end{pmatrix}.
\end{align*}

The portfolio processes are based on the proportions of stocks to liquid assets in the balance sheets of six life insurance companies, see~\prettyref{tab:equity_insurers}.\footnote{These companies are: Allianz Lebensversicherung AG, Debeka Lebensversicherung AG, ERGO Lebensversicherung AG, R$+$V Lebensversicherung AG, AXA Lebensversicherung AG, Gothaer Lebensversicherung AG.} We emphasize that the upcoming modeling assumptions are hypothetical; hence, our results cannot be used to draw conclusions about specific companies. We used their balance sheet values only to ensure realistic orders of magnitude.
\begin{table}
	\begin{tabular}{rrrrrrrrr}
		\hline
		Insurer i & $EC_i$ & $A_i$ & $L_i$ & $S_i$ & $x_0^i$ & $S_i/x_0^i$ & $\pi^i_1$ & $\pi^i_2$\\
		\hline
		1 & 2567 & 290685 & 288118 & 136096 & 191434 & 0.7109 & 0.6043 & 0.1066\\
		2 & 919 & 57851 & 56932 & 10646 & 51572 & 0.2064 & 0.1755 & 0.0310\\
		3 & 743 & 41133 & 40390 & 11870 & 36068 & 0.3291 & 0.2797 & 0.0494\\
		4 & 1207 & 83914 & 82707 & 25033 & 62545 & 0.4002 & 0.3402 & 0.0600\\
		5 & 383 & 26179 & 25796 & 9747 & 17655 &  0.5521 & 0.4693 & 0.0828\\
		6 & 516 & 18603 & 18087 & 12935 & 15695 & 0.8241 & 0.7005 & 0.1236\\
		\hline
	\end{tabular}
	\captionsetup{font=footnotesize}
	\caption{Equity capital ($EC_i$), assets ($A_i$), liquid assets ($x_0^i$) and stocks ($S_i$) of six insurers from 2023 and the resulting portfolio processes calculated as $\pi^i_1 = 0.85 \cdot S_i/x_0^i$ and $\pi^i_2 = 0.15\cdot S_i/x_0^i$, as well as, the liabilities ($L_i$) given by $L_i=A_i-EC_i$.}
	\label{tab:equity_insurers}
\end{table}

We model the future liabilities $Z_i$ of insurer $i$ in two ways. First, motivated by the approximation of total claim amounts for non-life portfolios with gamma distributions~\parencite[Chapter~4]{wuthrich_non-life_2024}, we set $Z_i\sim\Gamma(k_i,s_i)$, where $\Gamma(k,s)$ denotes the gamma distribution with mean $ks$ and variance $ks^2$, i.e.,~$k$ is the shape and $s$ the scale parameter. We refer to this case as \textit{model 1}.

Note that the gamma distribution is light-tailed~\parencite[Chapter 3]{embrechts_modelling_1997}. In \textit{model~2}, we test the impact of more harmful (heavy-tailed) distributions by replacing the liability distributions for some insurers as follows: for insurers $2$ and $4$, we set $Z_i\sim \text{LN}(\mu_i,\sigma_i)$, where $\text{LN}(\mu,\sigma)$ denotes the lognormal distribution with log-mean $\mu$ and log-standard deviation $\sigma$. For insurers $1$ and $6$, we set $Z_i\sim \text{GPD}(\xi_i,\nu_i,\beta_i)$, where $\text{GPD}(\xi,\nu,\beta)$ denotes the GPD with shape $\xi$, location $\nu$ and scale $\beta$. For insurers $3$ and $5$ we apply the same gamma distributions as before.

The parameters of the liability distributions are estimated using the method of moments applied to the annual liabilities from 2019 to 2023 for the first four insurers. For the fifth insurer, we use liabilities from 2018 to 2022 due to an excessive drop in the liabilities from 2022 to 2023, which would distort the results of our analysis. For the sixth insurer, we estimate the variance based on the liabilities from 2018 to 2023 to ensure that the estimated shape parameter $\xi_6$ of the GPD is positive, i.e.,~we obtain a heavy-tailed behavior. For the method of moments, we use the 2023 liability value as mean, see again~\prettyref{tab:equity_insurers}. The estimated parameters are reported in~\prettyref{tab:liability_parameters_insurers}.

\begin{table}
	\begin{tabular}{r|rr|rr|rrr}
		\hline
		Insurer i & $k_i$ & $s_i$ & $\mu_i$ & $\sigma_i$ & $\xi_i$ & $\nu_i$ & $\beta_i$\\
		\hline
		1 & 282.93 & 1018.34 & - & - & \textbf{0.4185} & \textbf{281203} & \textbf{4021}\\
		2 & 729.01 & 78.09 & \textbf{10.94} & \textbf{0.1473} & - & - & - \\
		3 & \textbf{6996.42} & \textbf{5.77} & - & - & - & - & - \\
		4 & 64.04 & 1291.55 & \textbf{11.21} & \textbf{0.4722} & - & - & - \\
		5 & \textbf{1083.58} & \textbf{23.81} & - & - & - & - & - \\
		6 & 1076.11 & 16.81 & - & - & \textbf{0.1901} & \textbf{17653} & \textbf{352}\\
		\hline
	\end{tabular}
	\captionsetup{font=footnotesize}
	\caption{Parameters for the liability distributions in model 1 ($k_i$ and $s_i$) and model 2 (values highlighted in bold). Here, $\nu_i$ is given as $97.6\%$ of the liability value of insurer $i$ from~\prettyref{tab:equity_insurers}.}
	\label{tab:liability_parameters_insurers}
\end{table}

Even if the equity capital of insurer $i$ is given by $X_i = Y_i+(-Z_i)$, that is, the sum of two independent random variables, we cannot relate the (Multi-)PELVE values of the equity capitals and the (Multi-)PELVE values of the assets and the liabilities in general.\footnote{For instance, if the random variables are comonotone, then it is possible to apply~\cite[Theorem 1]{li_pelve_2023}, compare also with their Examples 1 and 2 which discusses the sum of a long position in a stock and a corresponding call option.} Hence, we estimate the distribution of $X_i$ via $1\,000\,000$ Monte-Carlo simulations. \prettyref{fig:boxplot} shows the corresponding boxplots for each insurer. In model $2$, insurer $1$ exhibits large negative outliers due to a GPD with $\xi_1=0.4185$, whereas insurer $6$'s GPD has a much smaller shape parameter $\xi_6=0.1901$ (lighter tail). For insurers $2$ and $4$, perturbed with lognormal distributions, model $2$ gives larger negative outliers and wider boxes. Especially by comparing insurers $4$ and $6$ we see that the distribution of insurer $4$ leads in model $2$ to much larger outliers. \prettyref{fig:hist_firms} illustrates the equity capital distributions for insurers $1$, $2$, $4$ and $6$ -- the ones for which the liability distributions differ between model $1$ and $2$. These histograms show that the equity capital distributions for insurers $1$, $2$ and $4$ admit a heavier left tail in model $2$ than in model $1$, while the left tail for insurer $6$ is only slightly heavier in model~$2$.

\begin{figure}
	\begin{subfigure}[b]{0.48\textwidth}
		\includegraphics[width=\linewidth]{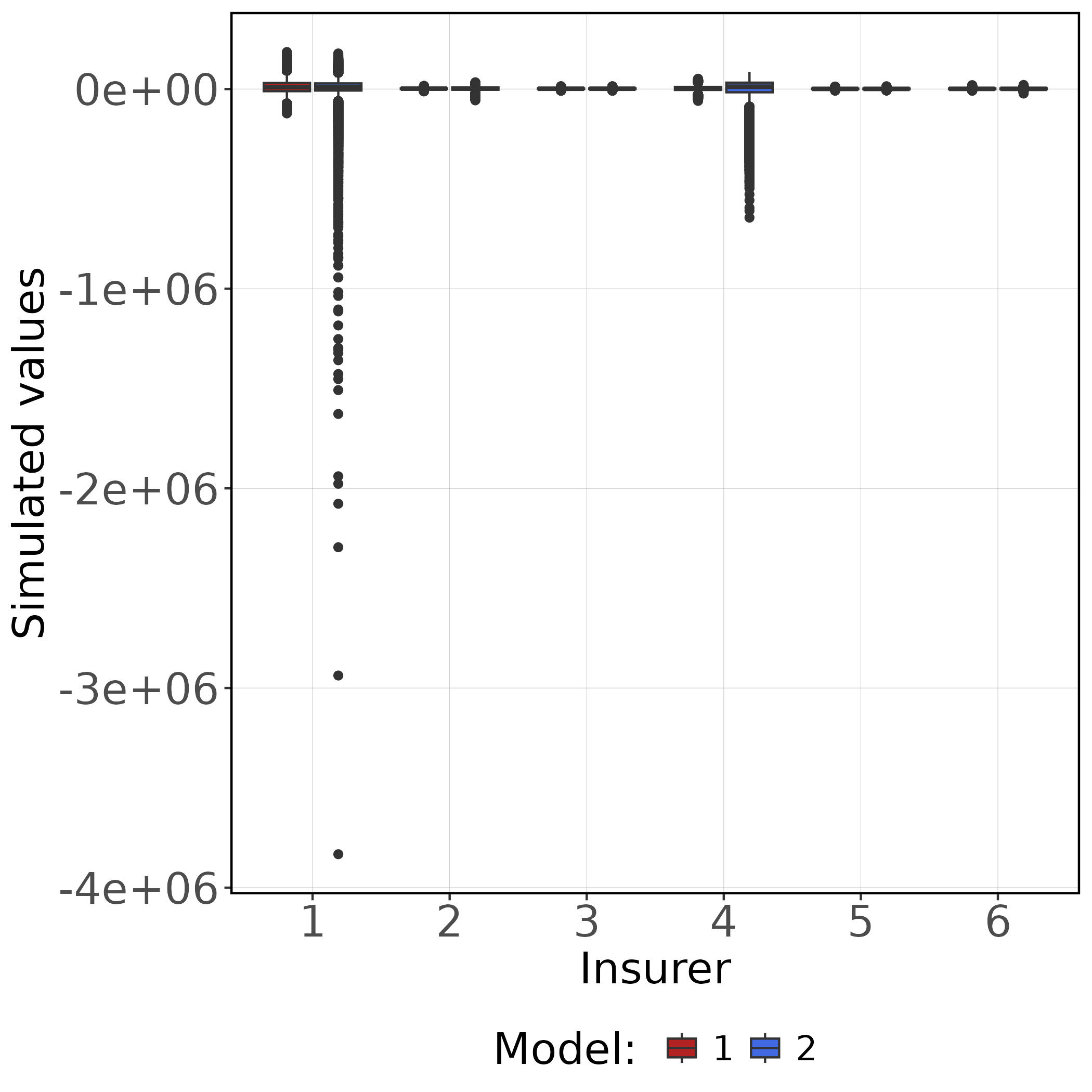}
	\end{subfigure}
	\hfill
	\begin{subfigure}[b]{0.48\textwidth}
		\includegraphics[width=\linewidth]{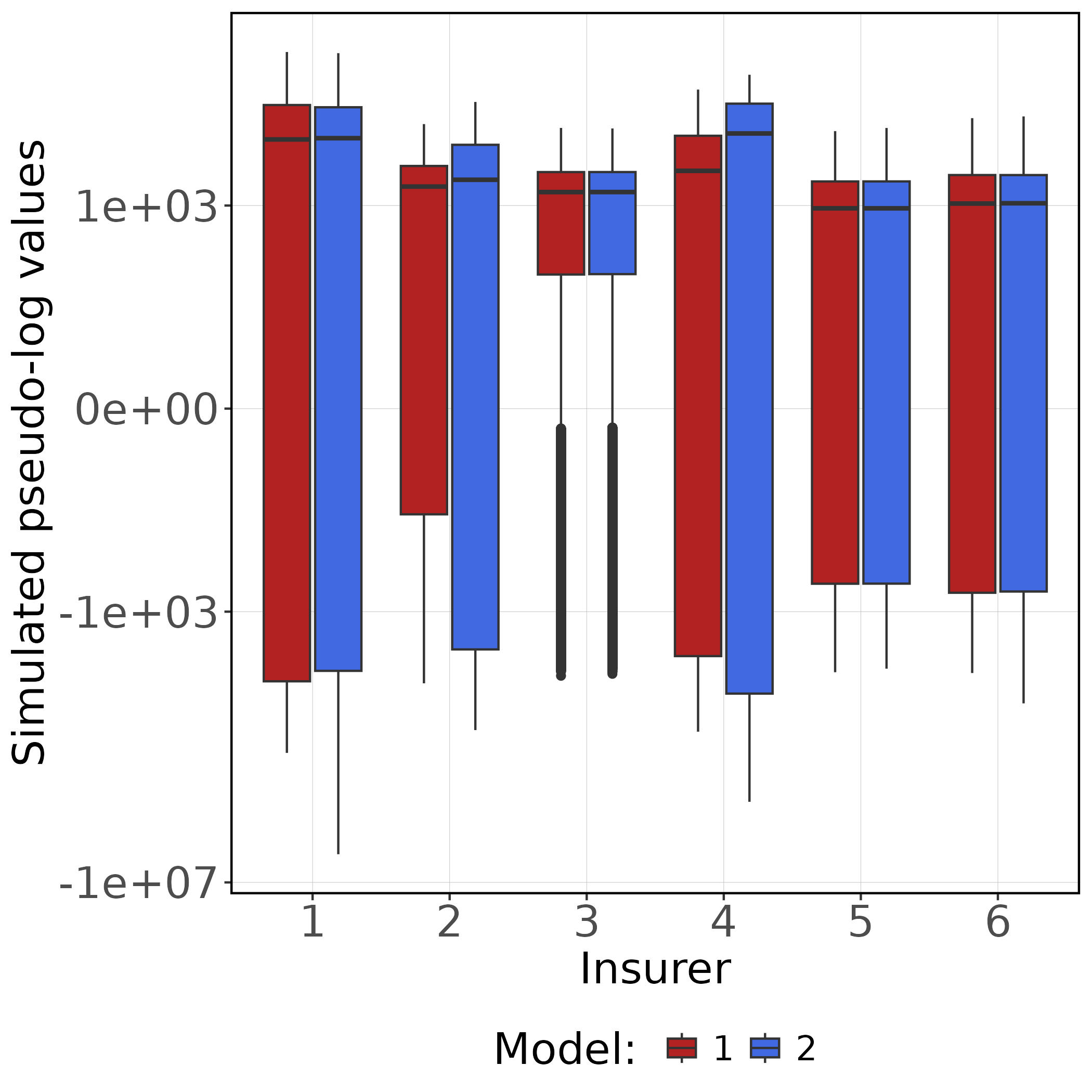}
	\end{subfigure}
	\captionsetup{font=footnotesize}
	\caption{Boxplot of simulated equity capitals of six insurers. The left-hand side are the simulated values and the right-hand side are the pseudo-log transformed ($f(x) = \text{sign}(x)\log_{10}(1+|x|)$) values.}
	\label{fig:boxplot}
\end{figure}

\begin{figure}
	\begin{subfigure}[b]{0.48\textwidth}
		\includegraphics[width=\linewidth]{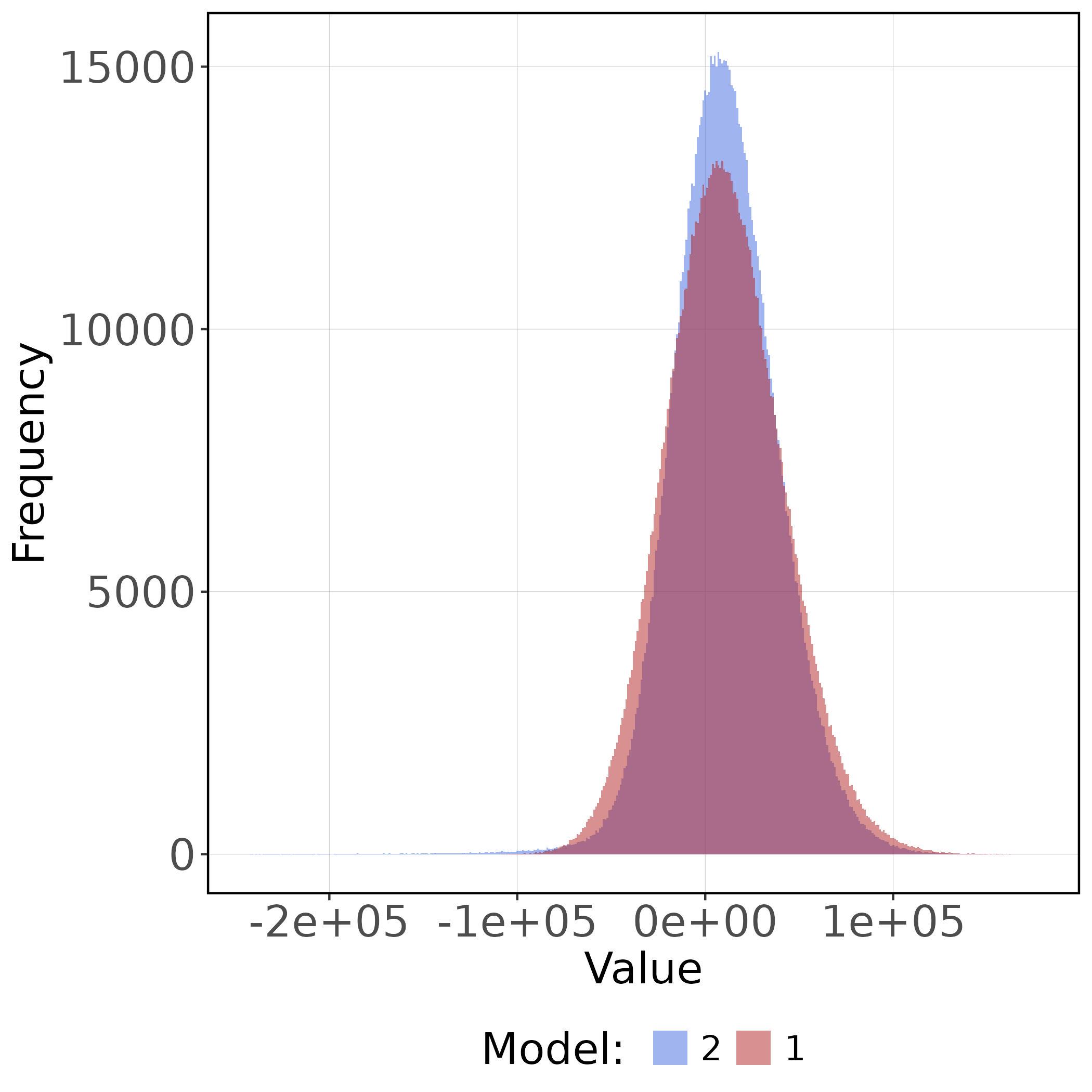}
		\caption{Insurer 1}
	\end{subfigure}
	\hfill
	\begin{subfigure}[b]{0.48\textwidth}
		\includegraphics[width=\linewidth]{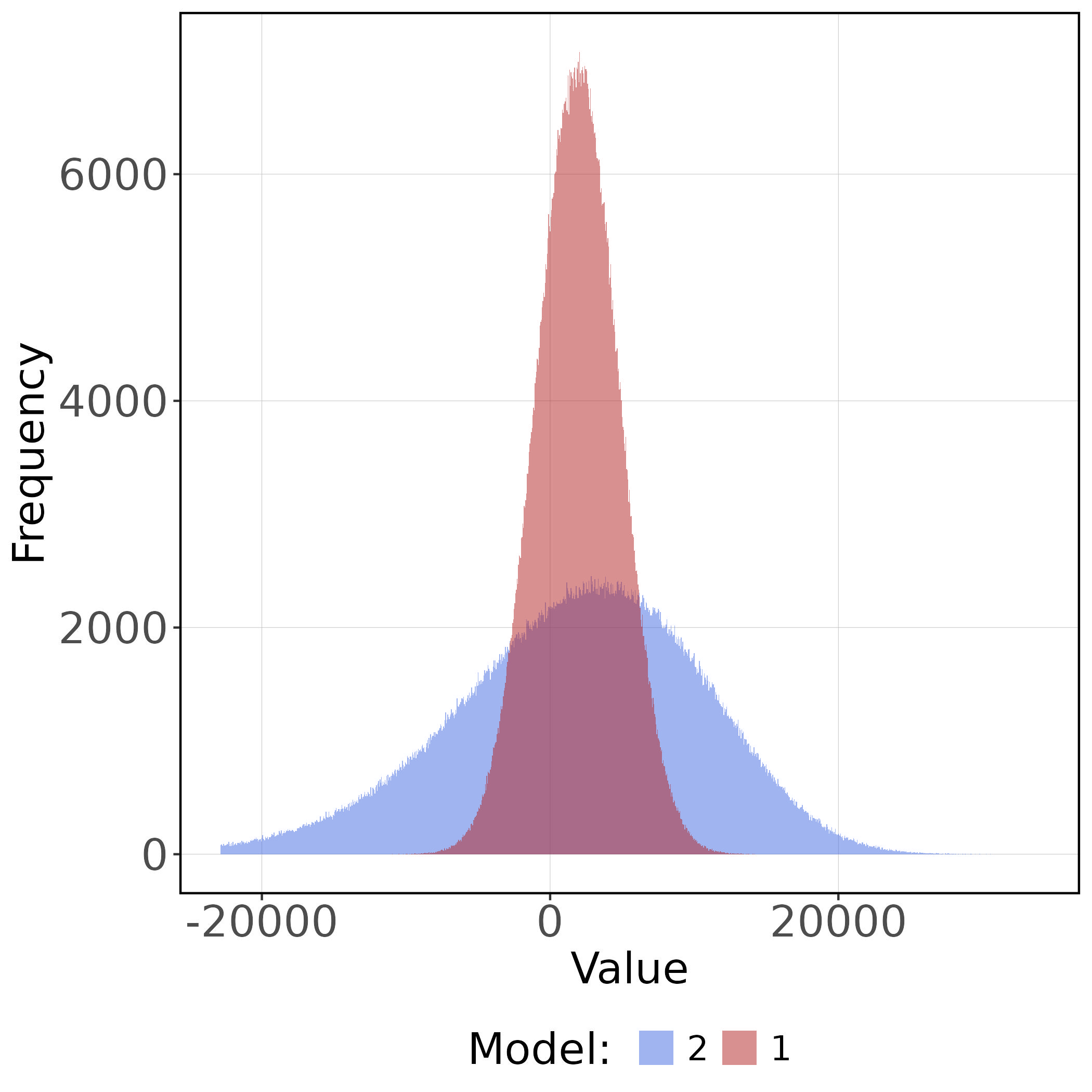}
		\caption{Insurer 2}
	\end{subfigure}
	\begin{subfigure}[b]{0.48\textwidth}
		\includegraphics[width=\linewidth]{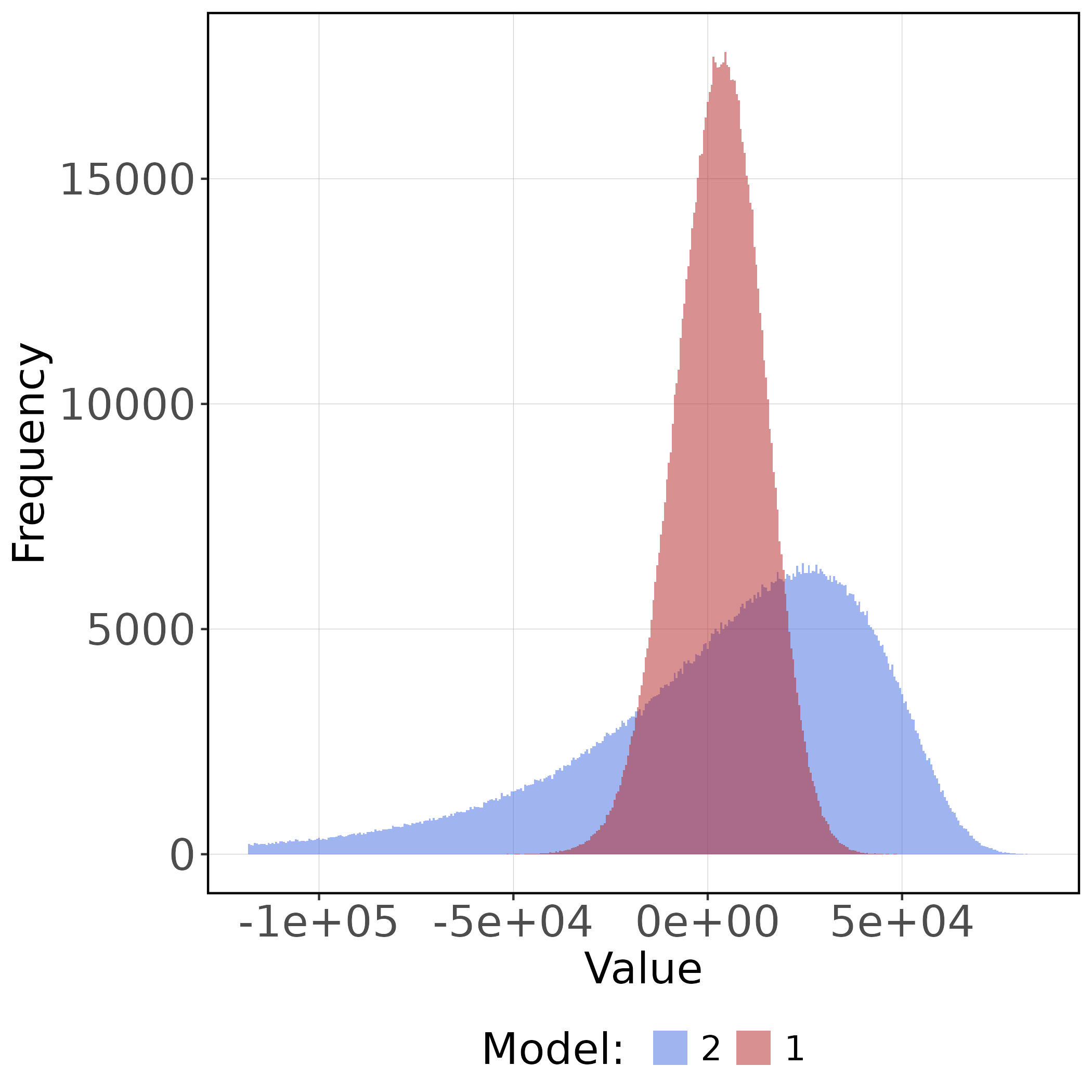}
		\caption{Insurer 4}
	\end{subfigure}
	\hfill
	\begin{subfigure}[b]{0.48\textwidth}
		\includegraphics[width=\linewidth]{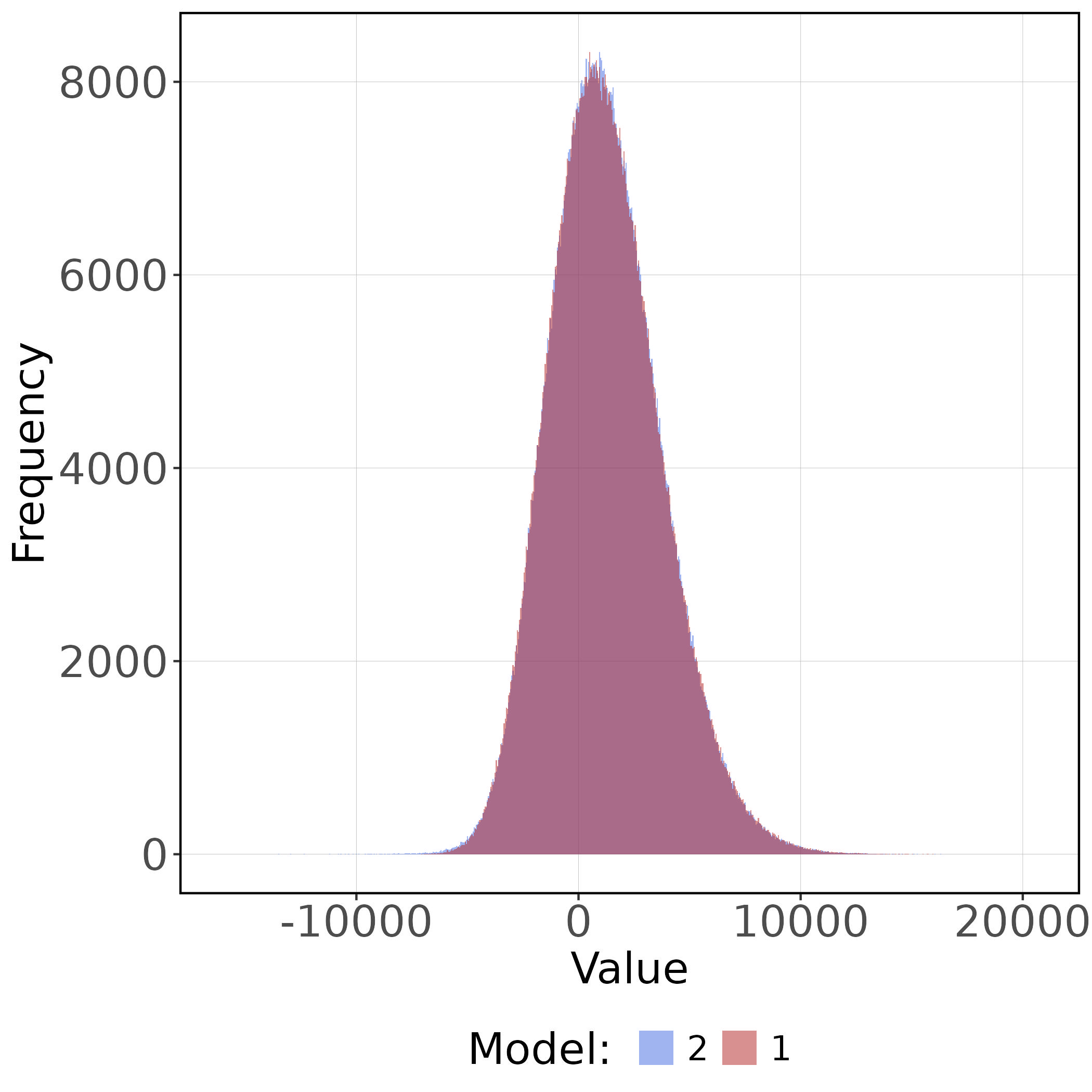}
		\caption{Insurer 6}
	\end{subfigure}
	\captionsetup{font=footnotesize}
	\caption{Histograms of the simulated equity capitals of insurers $1$, $2$, $4$ and $6$.}
	\label{fig:hist_firms}
\end{figure}

In addition, note that negative equity capital values would refer to bankruptcy of an insurer at time $1$. This is a usual way of modeling. For instance, the standard formula of Solvency II is motivated by a multivariate normal distribution, which means that negative equity capital values can occur with a non-negligible probability.

To analyse the situation of insurer $i$, we use its PELVE curve, defined as: $(0,1)\ni\lambda\mapsto \pelve{\lambda}(X_i)$. Note, a PELVE curve does not need to be monotone, see~\cite[Example 4]{assa_calibrating_2024}. Moreover, we present corresponding PELVE curves for our Multi-PELVE methods. By the boxplots in Figure~\ref{fig:boxplot} we see that insurers $1$ and $4$ do not manage their risks adequately in model $2$. On the other side, these insurers admit the largest equity capitals and asset volumes in Table~\ref{tab:liability_parameters_insurers} and therefore they have a strong influence on the market. To deal with these aspects, we test different possible weighting vectors $\omega$ for the A-PELVE and the MSE-PELVE. Of course, we test equal weights, i.e.,~$\omega_i=\frac{1}{6}$ for all $i$. A second variant is then to use weights based on each insurer's assets relative to the total assets of all insurers, i.e.,~$\omega_i = A_i / \left(\sum_{i=1}^{6}A_i\right)$. Also, weights are set to be inversely proportional to the asset volumes of the insurers. Indeed, with the help of $\tilde{\omega}_i = A_i / \left(\sum_{i=1}^{6}A_i\right)$ we set $\omega_i = (1/\tilde{\omega_i}) / (\sum_{i=1}^{6} 1 / \tilde{\omega}_i)$.  

In model 2, the latter choice prevents smaller insurers from being disadvantaged in favor of a single large insurer. For instance, in case of the MSE-PELVE the same change in the capital reserve due to a shift from VaR to ES can be more easily compensated by a large insurer than by a small one. The exact values of $\omega_i$ are given in~\prettyref{tab:optim_weights}. For simplicity, we also speak in case of unequal weights of weighted A-PELVE and weighted MSE-PELVE. For the Sys-PELVE, we set $g(x)=\max\{x,0\}$. However, this does not affect the results, as both the VaR and the optimal ES values are positive.

\begin{table}
	\begin{tabular}{c|rrrrrr}
		\hline
		Version  & 1 & 2 & 3 & 4 & 5 & 6 \\
		\hline
	V1 & 0.1667	& 0.1667 & 0.1667 & 0.1667 & 0.1667 & 0.1667\\
	V2 & 0.5608 & 0.1116 & 0.0794 & 0.1619 & 0.0505 & 0.0359\\
	V3 & 0.0231 & 0.1161 & 0.1633 & 0.0800 & 0.2565 & 0.3610 \\
		\hline
	\end{tabular}
	\captionsetup{font=footnotesize}
	\caption{Three versions (V1, V2, V3) of weighting vectors: V1 refers to equal weights (V1); V2 refers to weights $\omega_i = A_i/\left(\sum_{i=1}^{6} A_i\right)$; V3 refers to weights $\omega_i = (1/\tilde{\omega_i}) / (\sum_{i=1}^{6} 1 / \tilde{\omega}_i)$ with $\tilde{\omega}_i = A_i / \left(\sum_{i=1}^{6}A_i\right)$.}
	\label{tab:optim_weights}
\end{table}

In Figure~\ref{fig:single_pelve}, we present the PELVE curves of the individual insurers in models $1$ and $2$ with interpolation between grid points. Note that since the underlying distributions come from Monte-Carlo simulations,~\prettyref{cor:discontinuityPELVEcurves} implies that the curves are discontinuous. We see that in model $1$ the curves are close to each other, which indicates that also the Multi-PELVE methods should not differ significantly. Indeed, we justify this in Section~\ref{sec:model1} below. In contrast, the PELVE curves in model 2 differ significantly, which will give us significant differences in the Multi-PELVE versions. We analyse this meaningful situation in more detail in~\prettyref{sec:model2} below.

\begin{figure}
	\begin{subfigure}[b]{0.48\textwidth}
		\includegraphics[width=\linewidth]{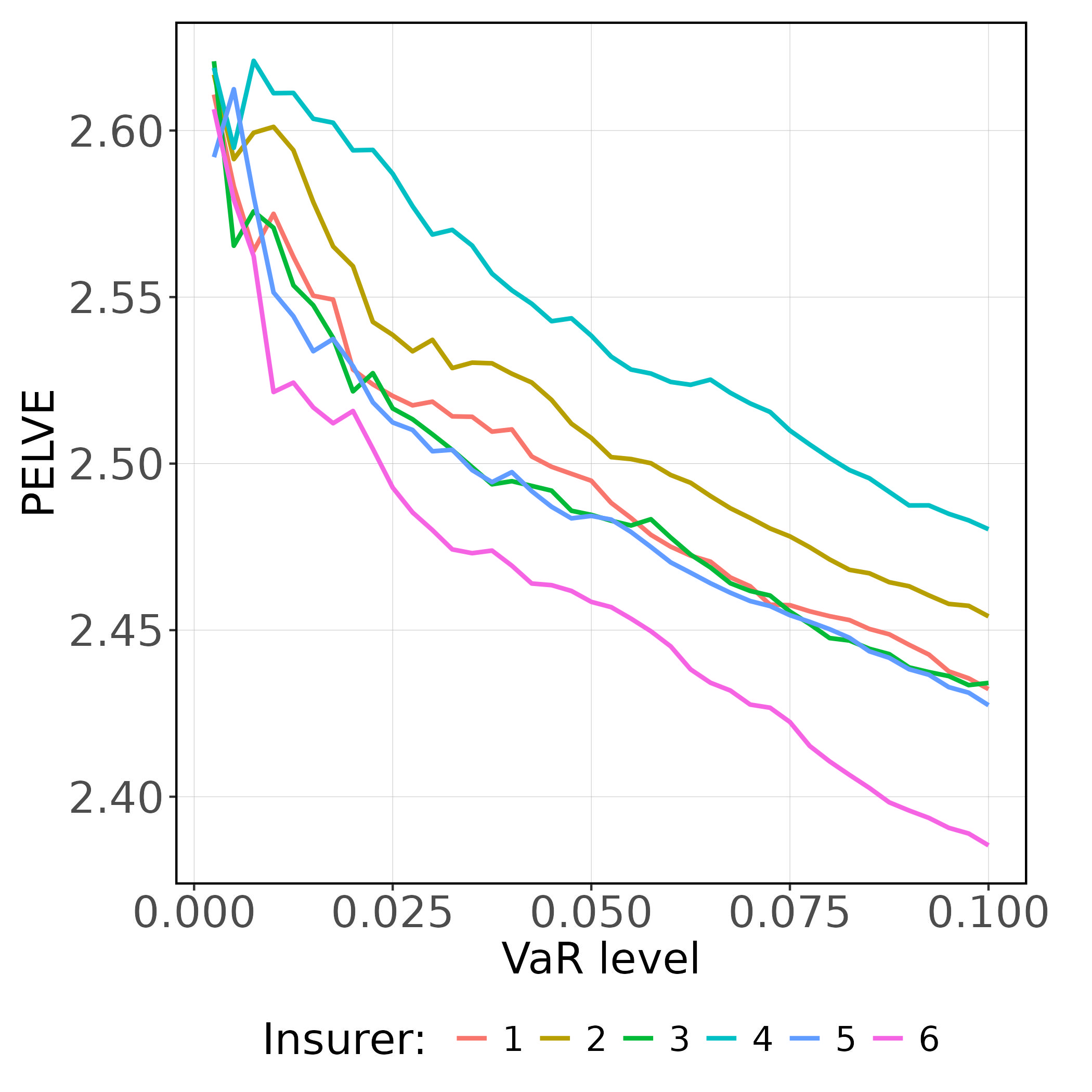}
		\caption{Model $1$}
	\end{subfigure}
	\hfill
	\begin{subfigure}[b]{0.48\textwidth}
		\includegraphics[width=\linewidth]{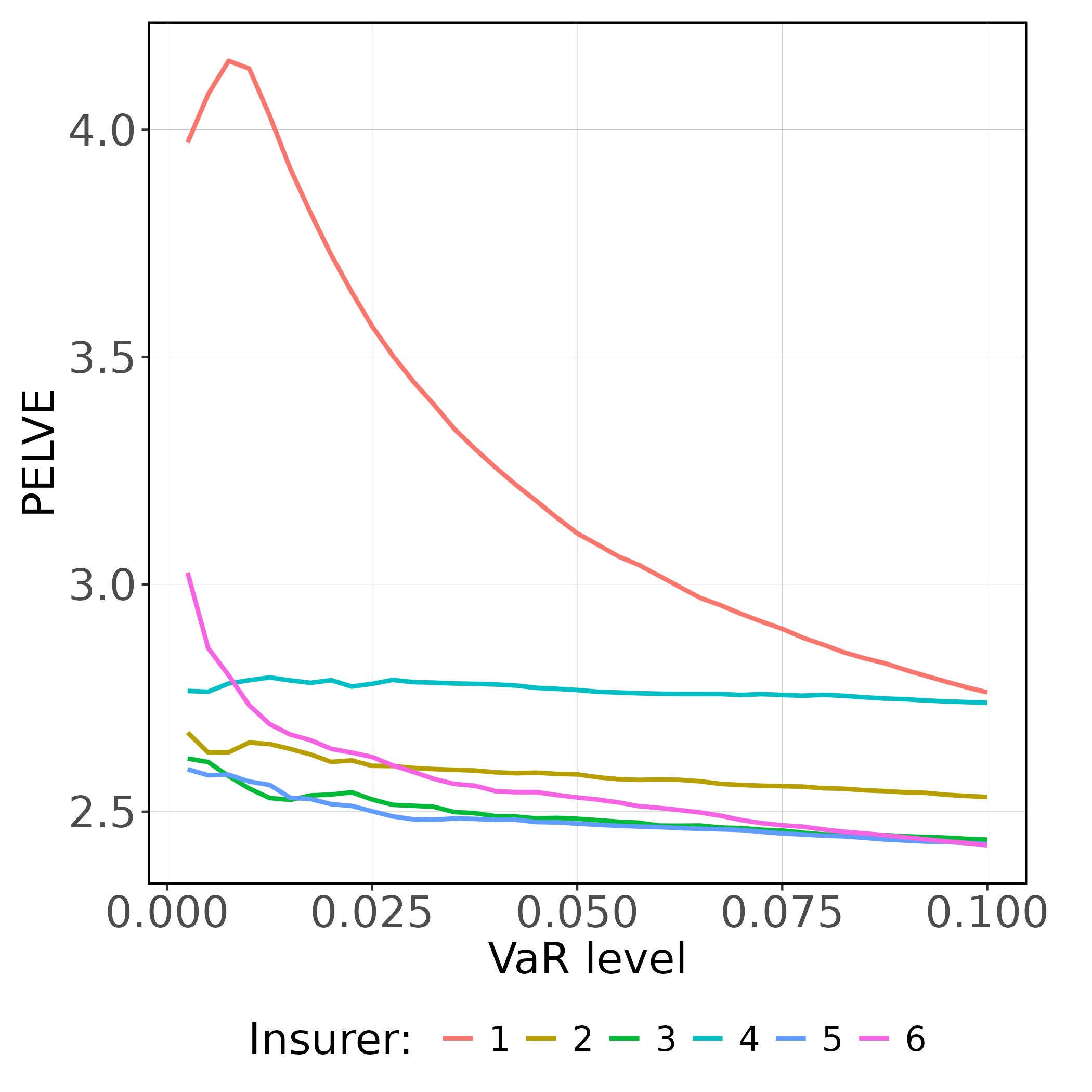}
		\caption{Model $2$}
	\end{subfigure}
	\captionsetup{font=footnotesize}
	\caption{PELVE curves of individual insurers for different VaR levels.}
	\label{fig:single_pelve}
\end{figure}

\subsection{Results: Model 1}\label{sec:model1}

All curves on the left-hand side (LHS) in Figure~\ref{fig:single_pelve} lie below $e\approx 2.718$, reflecting the light-tailed behavior of the underlying distributions, recall the discussion in~\prettyref{sec:heavyTails}. This is expected, as liabilities are modeled as gamma distributions, which are light-tailed. Furthermore, all PELVE curves admit the same shape, which looks roughly like a linear decrease, and they are close to each other. The largest difference, about $0.1$, occurs between insurers $4$ and~$6$. 

\begin{figure}
	\begin{subfigure}[b]{0.48\textwidth}
		\includegraphics[width=\linewidth]{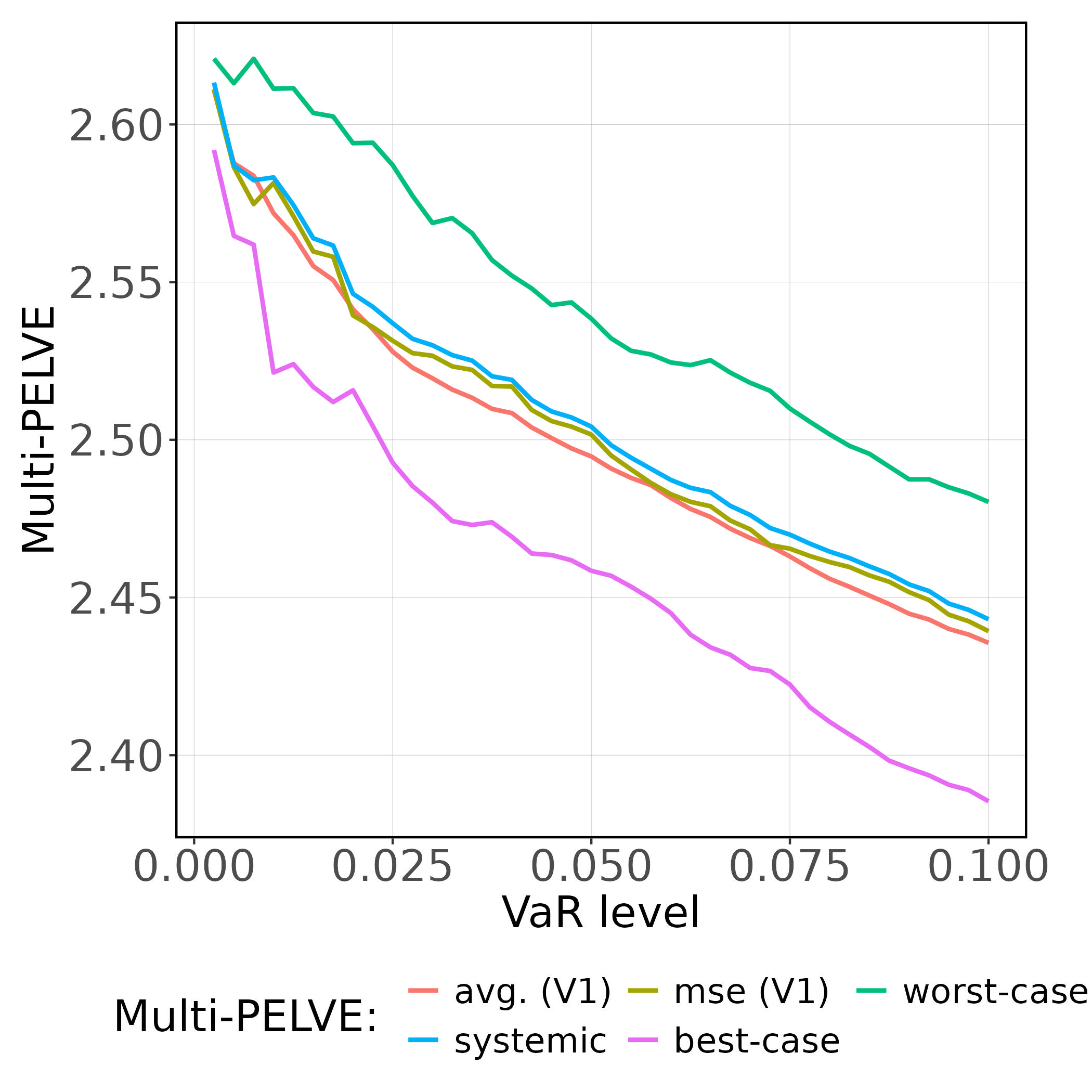}
	\end{subfigure}
	\hfill
	\begin{subfigure}[b]{0.48\textwidth}
		\includegraphics[width=\linewidth]{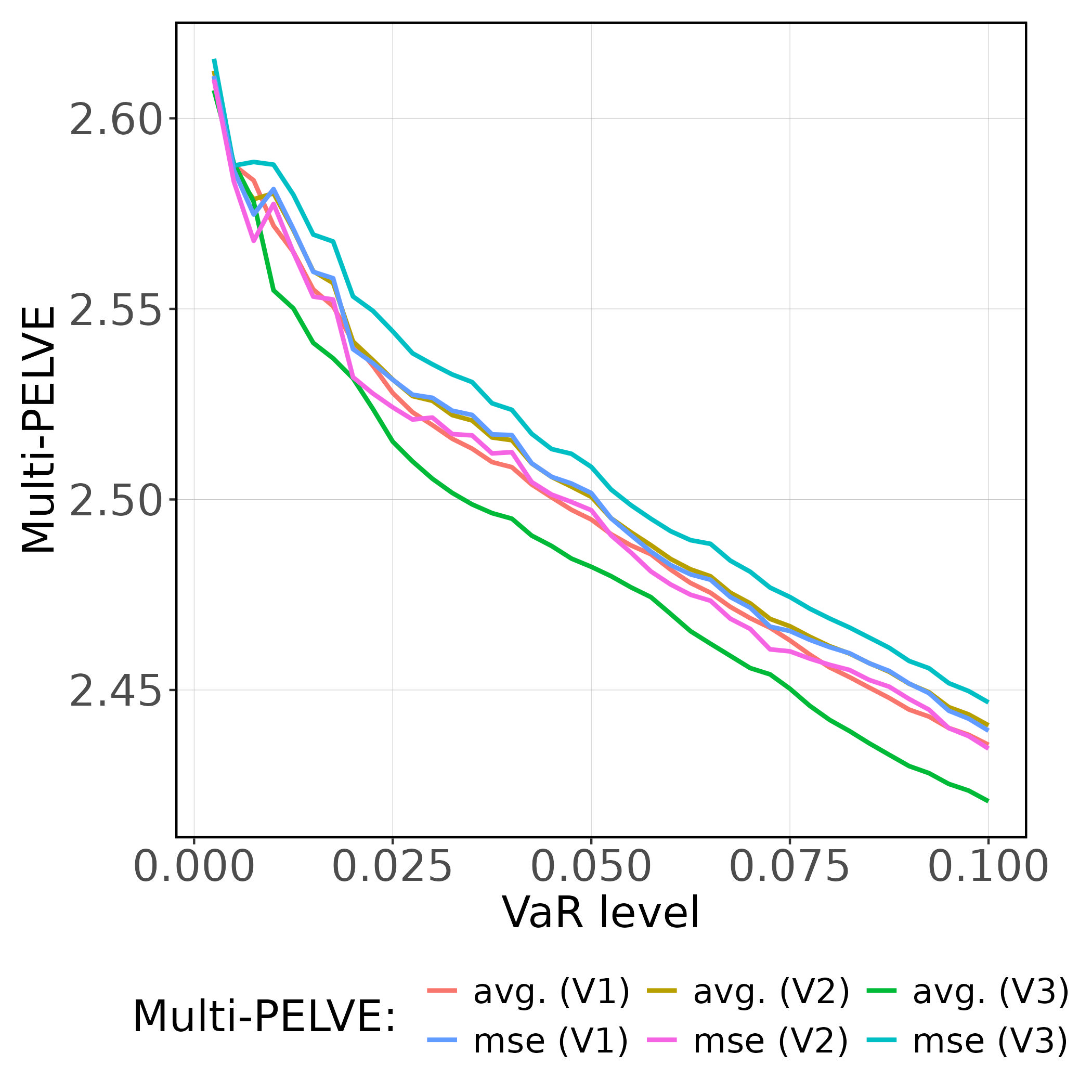}
	\end{subfigure}
	\captionsetup{font=footnotesize}
	\caption{Multi-PELVE methods for different VaR levels in model 1.}
	\label{fig:pelve_lognormalGamma}
\end{figure}

Also the Multi-PELVE curves are close to each other, see~\prettyref{fig:pelve_lognormalGamma}. Besides the two extremes of WC- and BC-PELVE, the differences between the other Multi-PELVE methods are difficult to discern. Next, in~\prettyref{fig:comparing_diffs_in_sum_es_var}, we plot, for each $\Pi^{V}_\lambda\in\{\apelve{\lambda,\omega},\wcpelve{\lambda},\msepelve{\lambda,\omega},\syspelve{\lambda,g}\}$,
\begin{align*}
	\lambda\mapsto\sum_{i=1}^{n}\Big(\expectedShortfall{\Pi^{V}_\lambda(\mathbf{X})\lambda}{X_i}-\valueAtRisk{\lambda}{X_i}\Big),\quad \lambda\mapsto\sum_{i=1}^{n}\Big|\expectedShortfall{\Pi^{V}_\lambda(\mathbf{X})\lambda}{X_i}-\valueAtRisk{\lambda}{X_i}\Big|,
\end{align*}
where these maps are displayed on the LHS and the right-hand side (RHS) of~\prettyref{fig:comparing_diffs_in_sum_es_var}, respectively. They are statistics of the overall change in the capital reserves of the six insurers when VaR is replaced by ES. On the LHS, it is not a surprise that the Sys-PELVE yields a constant zero line, as it satisfies $$\sum_{i=1}^{n}\left(\max\left\{\expectedShortfall{\syspelve{\lambda,g}(\mathbf{X})\lambda}{X_i},0\right\}-\max\{\valueAtRisk{\lambda}{X_i},0\}\right) = 0,$$ 
and thus, if all VaR and ES values are positive, the plotted key figure is zero. 
In contrast, on the RHS, the MSE-PELVE and the (weighted) A-PELVE outperform the Sys-PELVE, indicating that, on the LHS, positive and negative changes in individual insurers' capital reserves offset each other. 

\begin{figure}
	\begin{subfigure}[b]{0.48\textwidth}
		\includegraphics[width=\linewidth]{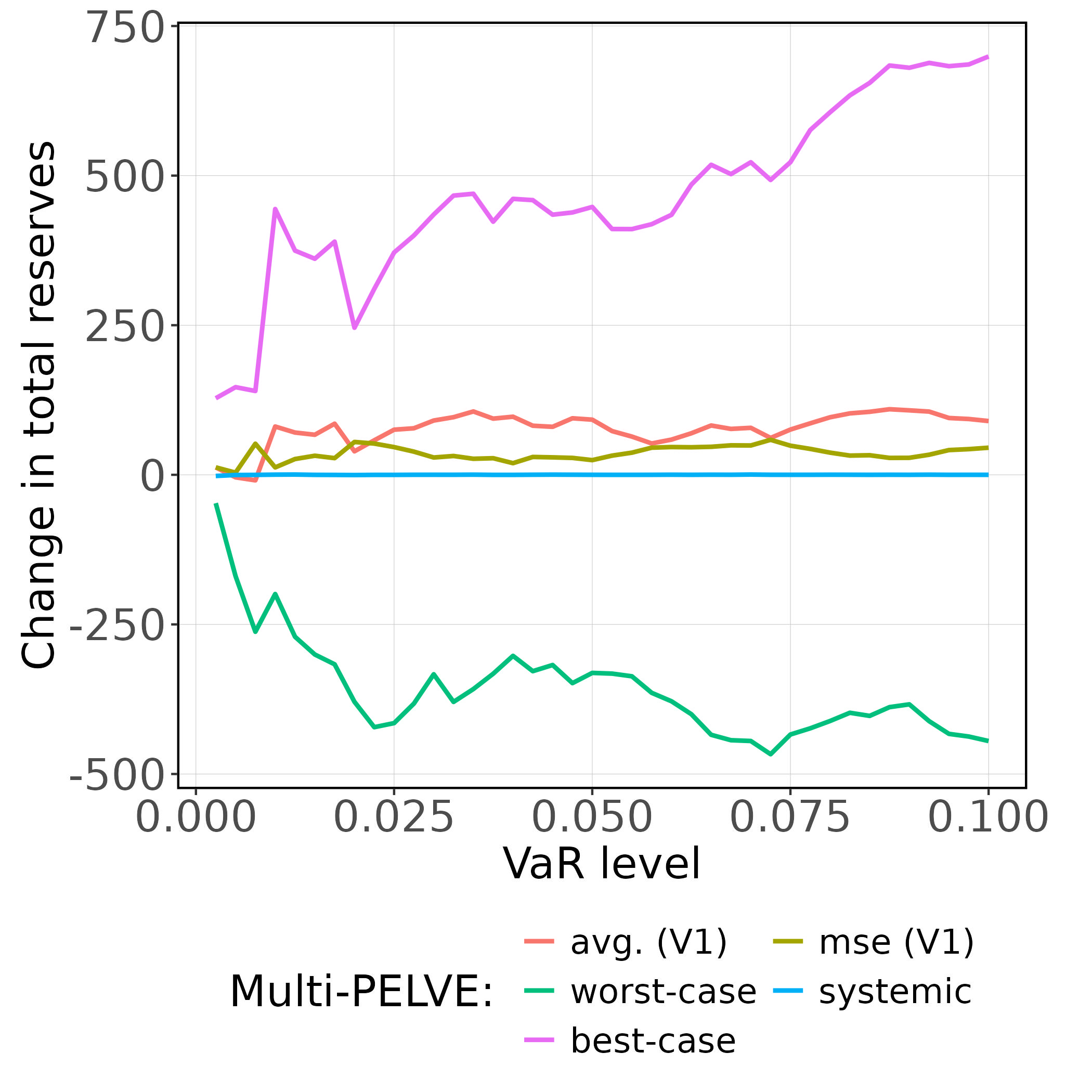}
		\caption{Total change in sum of reserves}
	\end{subfigure}
	\hfill
	\begin{subfigure}[b]{0.48\textwidth}
		\includegraphics[width=\linewidth]{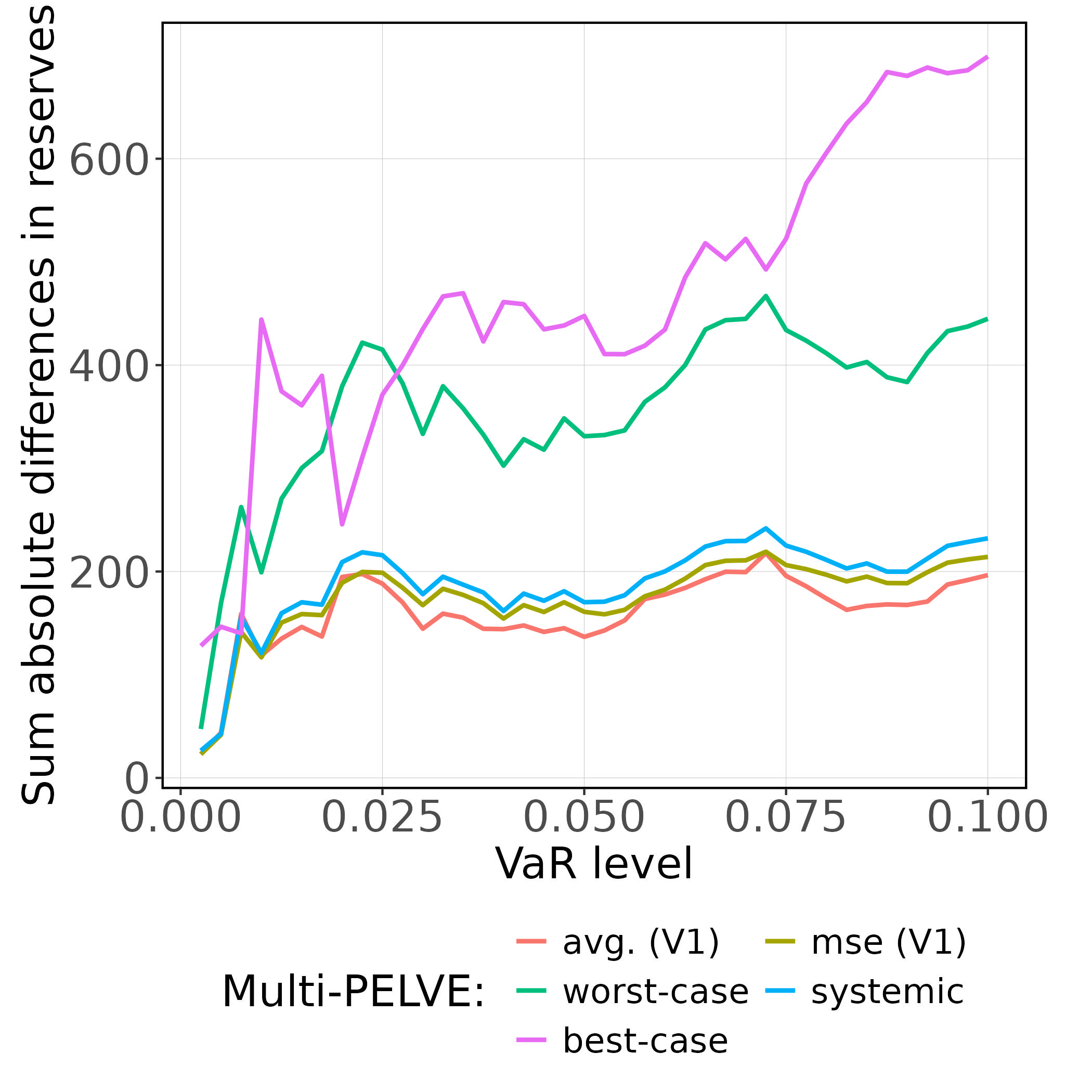}
		\caption{Sum of absolute differences in reserves}
	\end{subfigure}
	\captionsetup{font=footnotesize}
	\caption{Total change in sum of capital reserves and sum of absolute differences in the capital reserves of all insurers under model 1.}
	\label{fig:comparing_diffs_in_sum_es_var}
\end{figure}

Also note that differences among the A-PELVE, MSE-PELVE, and Sys-PELVE are negligible relative to the total equity of $6335$ (cf.~\prettyref{tab:equity_insurers}); for instance, the changes for the A-PELVE and the MSE-PELVE are always below $150$ on the LHS ($2.4\%$ relative to $6335$), and the maximum difference on the RHS is about $50$ ($0.8\%$ relative to $6335$). Furthermore, the RHS demonstrates that the difference between ES and VaR for each insurer goes to zero when the VaR level goes to zero. This fact is interesting, because it does not hold in model 2, cf.~Figure~\ref{fig:comparing_diffs_in_sum_es_var_perturbed} (B) below.  For brevity, detailed explanations are shifted to Appendix~\ref{sec:gammaDistributions}.

\prettyref{fig:capitalChanges_gamma} illustrates the relative change in each insurer's capital reserves using the WC-PELVE (LHS) or MSE-PELVE (RHS), i.e.,~for $\Pi_\lambda^{V}\in\{\wcpelve{\lambda},\msepelve{\lambda,\omega}\}$ and insurer $i\in[6]$ it is the map 
\begin{align*}
	\lambda\mapsto \frac{\expectedShortfall{\Pi_\lambda^{V}(\mathbf{X})\lambda}{X_i}-\valueAtRisk{\lambda}{X_i}}{\valueAtRisk{\lambda}{X_i}}. 
\end{align*}

For the WC-PELVE, the new ES level is given by the PELVE curve of insurer $4$, since it is the maximum among all, cf.~\prettyref{fig:pelve_lognormalGamma}. This results in relative reductions of the capital reserves for the other insurers; for instance, for insurer $6$ we observe reductions below $-2\%$. Allowing all insurers to reduce reserves may threaten the market stability and conflict with regulatory interests. In contrast, the other Multi-PELVE methods -- beside BC-PELVE -- account for all six insurers, see e.g.,~the MSE-PELVE on the RHS in Figure~\ref{fig:capitalChanges_gamma}.

The MSE-PELVE closely follows the PELVE curve of insurer $1$, so the relative changes for insurer~$1$ are near zero.  With its large equity capital and asset amount, insurer $1$ acts as a market leader. However, 
small relative changes for the market leader under the MSE-PELVE are not a common feature, as we will see for model 2 in the next subsection. \newline
\begin{figure}
	\begin{subfigure}[b]{0.48\textwidth}
		\includegraphics[width=\linewidth]{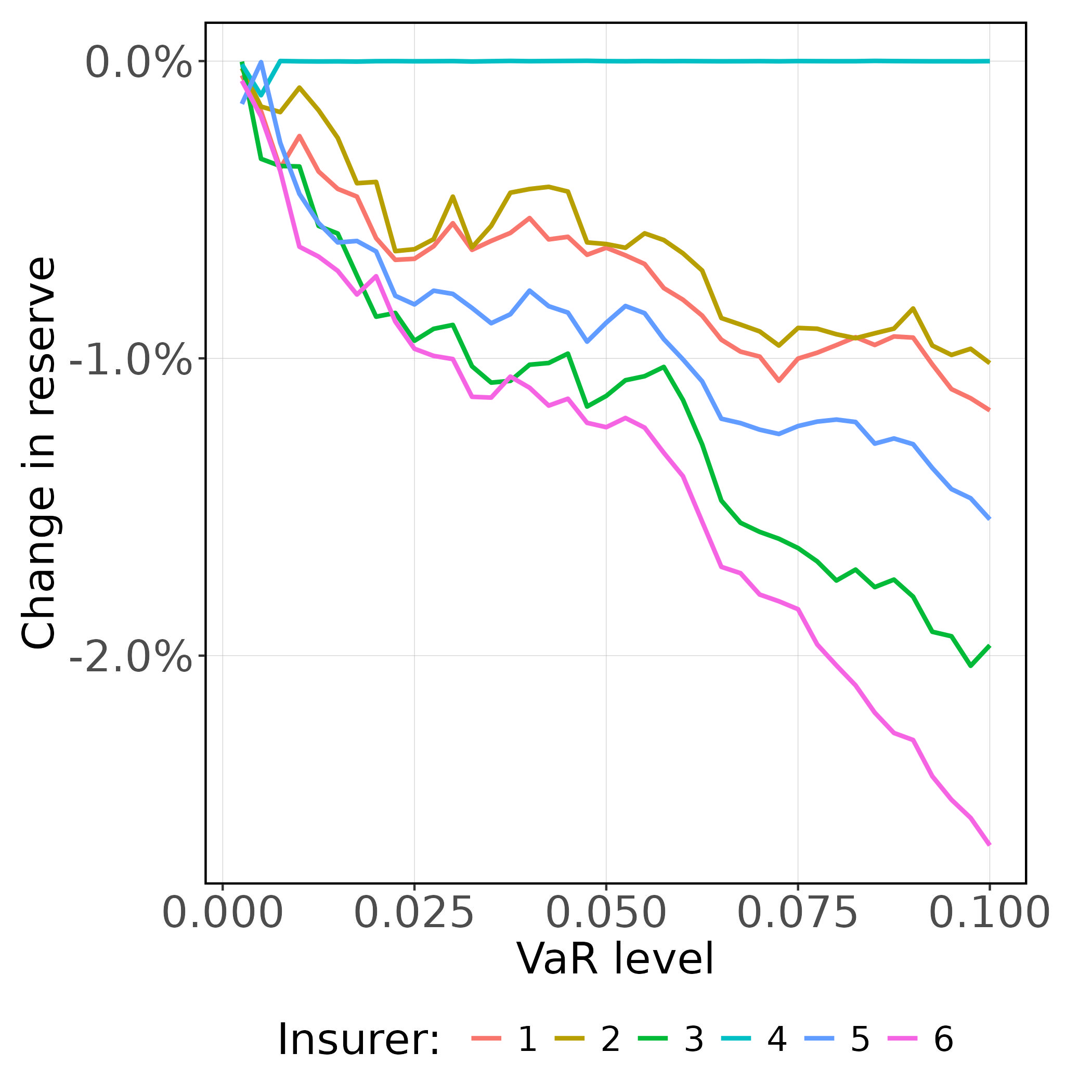}
		\caption{WC-PELVE}
	\end{subfigure}
	\hfill
	\begin{subfigure}[b]{0.48\textwidth}
		\includegraphics[width=\linewidth]{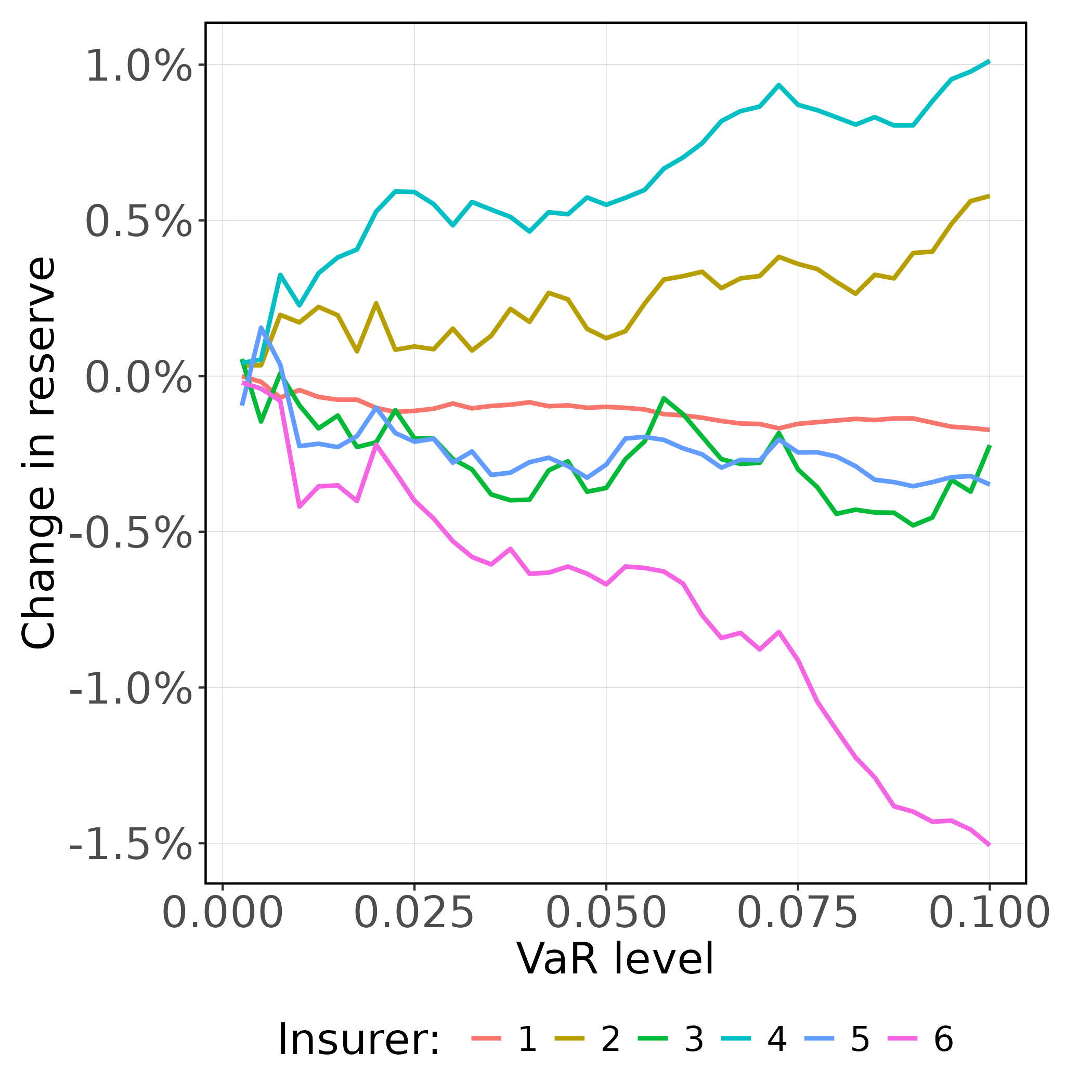}
		\caption{MSE-PELVE (V1)}
	\end{subfigure}
	\captionsetup{font=footnotesize}
	\caption{Changes in the capital reserve of the insurers under model 1, if they use the new ES level stemming from the WC-PELVE or the MSE-PELVE (V1).}
	\label{fig:capitalChanges_gamma}
\end{figure}

\noindent\fbox{\parbox{\textwidth}{\textit{Before we continue with model 2, let us conclude:} For model 1, where all insurers' distributions are of the same type (shifted lognormal minus gamma), all Multi-PELVE methods -- except the extreme ones, namely WC- and BC-PELVE -- yield similar results. This is favorable for the regulator, as the choice of the Multi-PELVE method is practically inconsequential.}}

\subsection{Results: Model 2}\label{sec:model2}

Recall that on the RHS in~\prettyref{fig:single_pelve}, we illustrate the PELVE curves for model 2. In contrast to model 1, the PELVE curves of insurers $1$, $2$, $4$ and $6$ admit significantly larger values. In particular, the GPD distribution of insurer $1$ leads to PELVE values exceeding $4$ at small levels. 

\begin{figure}
	\begin{subfigure}[b]{0.48\textwidth}
		\includegraphics[width=\linewidth]{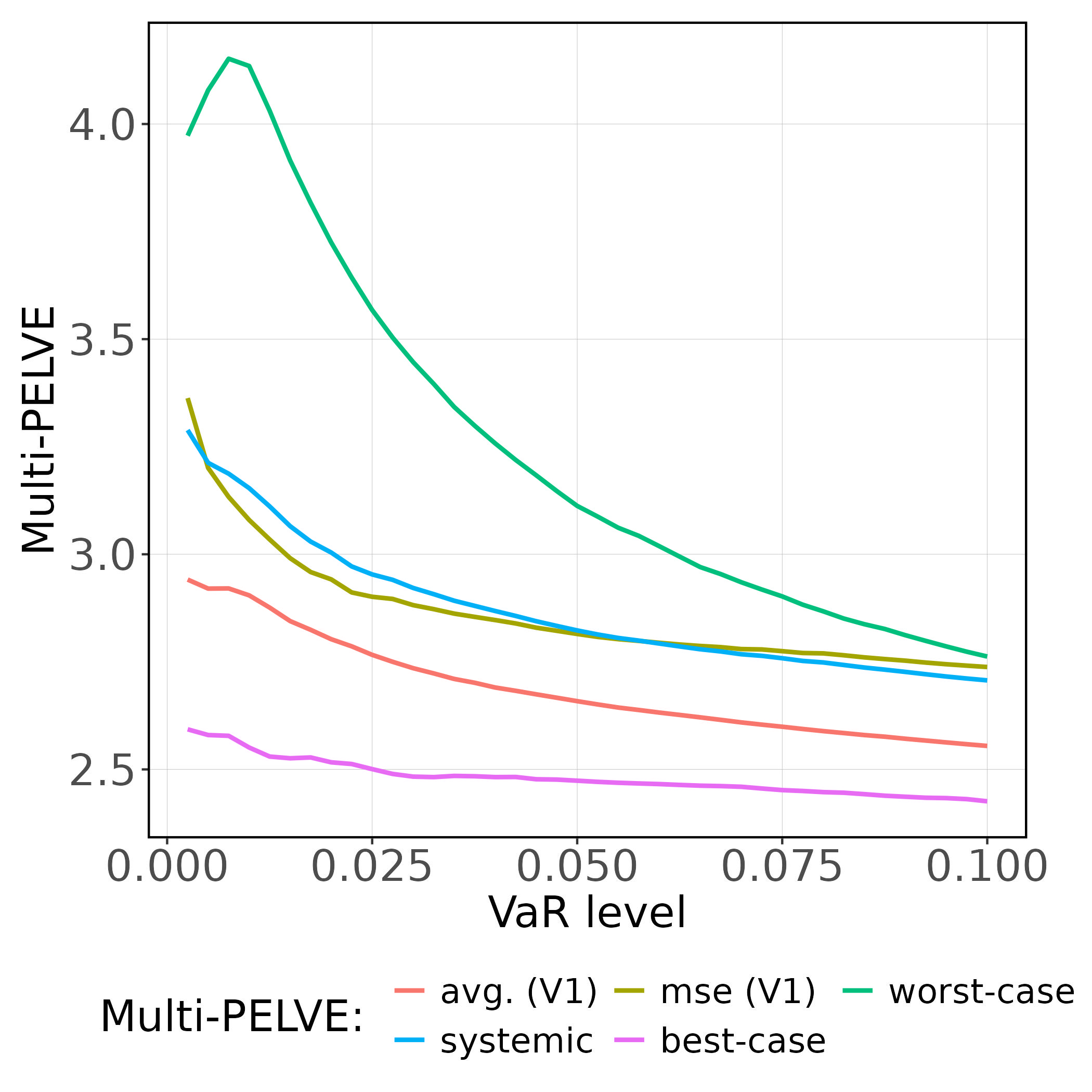}
	\end{subfigure}
	\hfill
	\begin{subfigure}[b]{0.48\textwidth}
		\includegraphics[width=\linewidth]{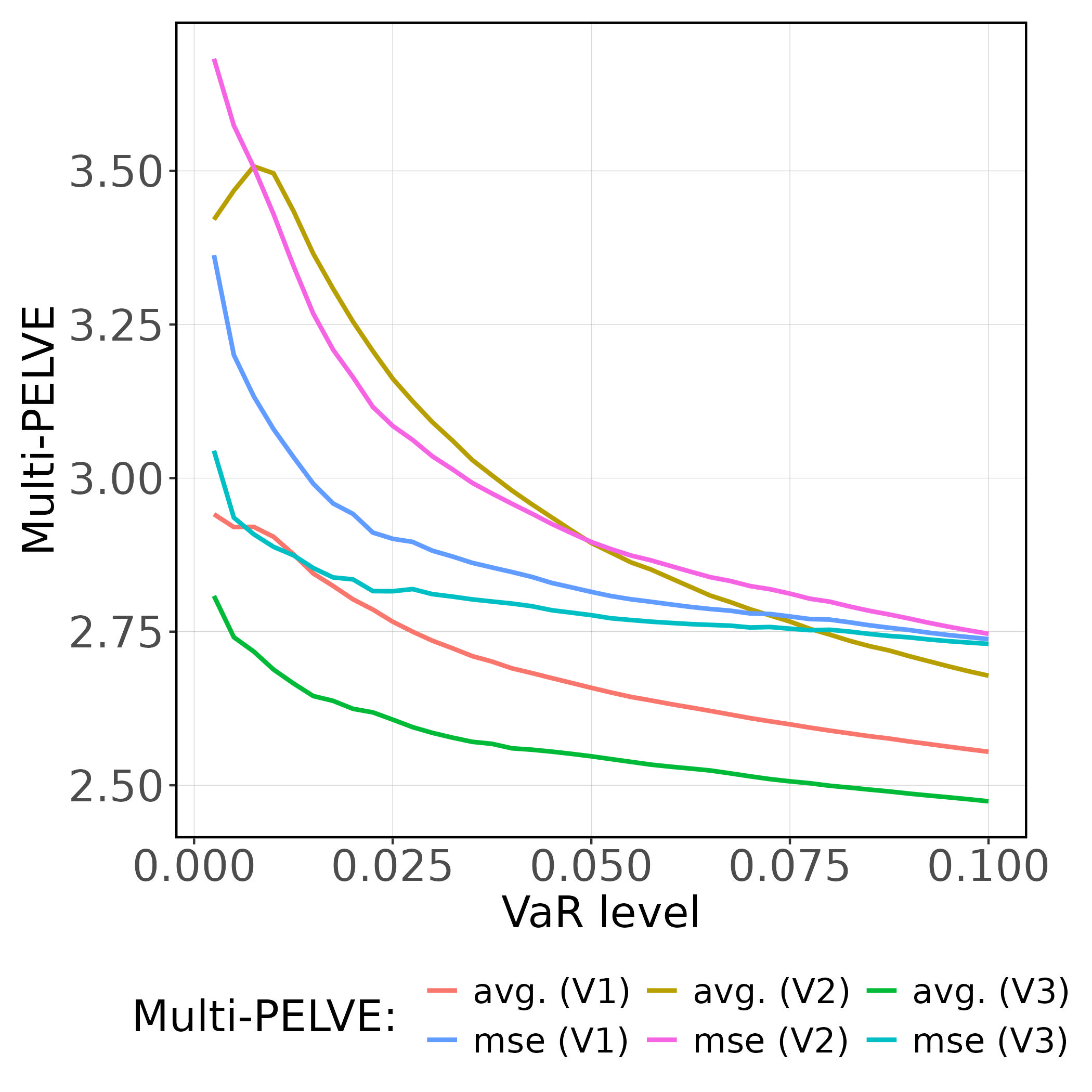}
	\end{subfigure}
	\captionsetup{font=footnotesize}
	\caption{Multi-PELVE methods for different VaR levels in model 2.}
	\label{fig:pelve_perturbed}
\end{figure}

The values of our Multi-PELVE methods differ substantially from those for model $1$ (cf.~\prettyref{fig:pelve_perturbed}) making the choice of the method important: (1) On the LHS of~\prettyref{fig:pelve_perturbed}, except for the WC-PELVE, the MSE-PELVE and Sys-PELVE admit the largest values. Both are significantly larger than the A-PELVE, which is itself larger than the BC-PELVE; (2) The RHS shows that A-PELVE and MSE-PELVE for (V2) are quite large and behave similarly. By increasing the VaR level, the MSE-PELVE (V2) is close to the other two MSE-PELVE versions (V1) and (V3); (3) Sys-PELVE and MSE-PELVE (V1) are comparable.

We elaborate on the reasons for these effects. For (1), note that the MSE-PELVE and the Sys-PELVE focus on sums of (transformed) differences in ES and VaR values of the individual insurers. These (non-robust) statistics are sensitive to large differences between ES and VaR for individual insurers and therefore the focus on the extreme insurer $1$ is stronger than for the A-PELVE, which simply averages new ES-levels. Regarding (2), the A-PELVE (V1) is smaller than its weighted version (V3) due to insurer $1$'s influence as a market leader. Furthermore, by increasing the VaR level, the influence of the extreme tail in the calculation of ES and VaR is reduced (shifting the focus to the body of the distribution). Therefore, the MSE-PELVE versions are close to each other. This does not hold for the A-PELVE versions, because the single PELVE curves for VaR levels are not close to each other, cf.~\prettyref{fig:single_pelve} (B).

For (3), the similarity between the Sys-PELVE and MSE-PELVE (V1) is surprising, given their apparently different methodologies. However, given that ES and VaR values are positive, the Sys-PELVE is based on the sum of ES values minus VaR values, i.e.,~it seeks the minimal $c$ subject to $\frac{1}{6}\sum_{i=1}^{6} \Big(\expectedShortfall{c\lambda}{X_i}-\valueAtRisk{\lambda}{X_i}\Big)\leq 0$, where the left side of this inequality is small when the ES-VaR difference for insurer $1$ is small. On the other hand, the MSE-PELVE is also based on (squared) ES-VaR differences. Furthermore, both Sys- and MSE-PELVE (V1) use equal weights, namely $\frac{1}{6}$, so the importance of reducing the ES-VaR difference for an individual insurer is the same, explaining their similarity.

An interesting question is about the dependence of the results on the correlation structure of the equity capital distributions. In our setting, we can control this by the proportions invested in the common and the idiosyncratic stock, cf.~with the (idiosyncratic) factor of $0.15$ in the formula $\pi_2^i = 0.15 \cdot S_i/x_0^i$ from Table~\ref{tab:equity_insurers}. On the LHS in~\prettyref{fig:pelve_diversification}, we plot the Multi-PELVE methods for a fixed VaR level of $1\%$ and different idiosyncratic factors. The curves only show a slight decreasing behavior. On the RHS we present the absolute changes in the capital reserves, which also admit a decreasing behavior, i.e.,~we obtain the smallest values when equity distributions are independent. Nevertheless, one should be careful, as different log-standard deviations of the common and idiosyncratic stocks may also lead to the opposite effect, i.e.,~increasing curves. Furthermore, for a new Sys-PELVE that takes the dependence structure into account, we refer the reader to Section~\ref{sec:beyondRegulatoryChanges} below.

\begin{figure}
	\begin{subfigure}[b]{0.48\textwidth}
		\includegraphics[width=\linewidth]{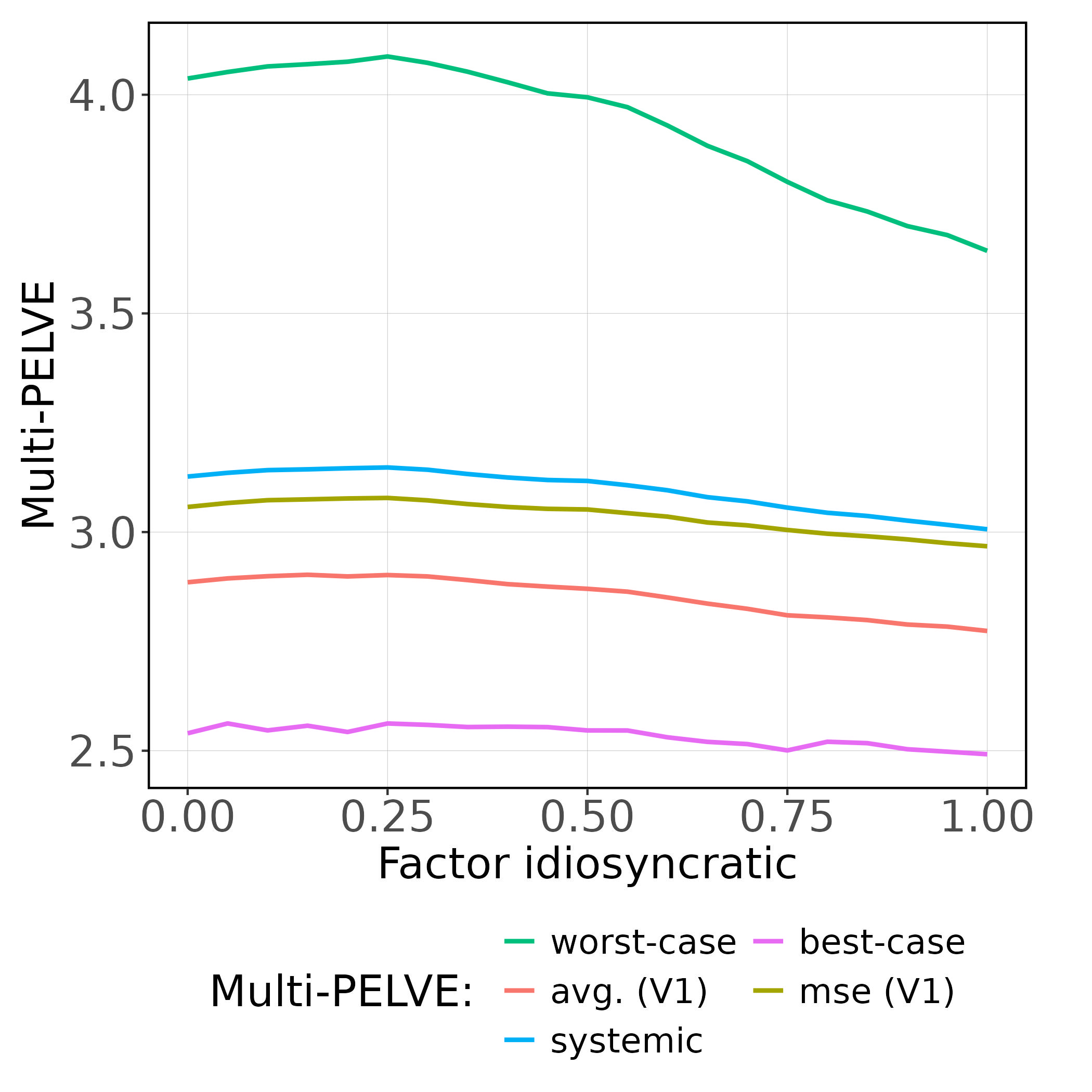}
	\end{subfigure}
	\hfill
	\begin{subfigure}[b]{0.48\textwidth}
		\includegraphics[width=\linewidth]{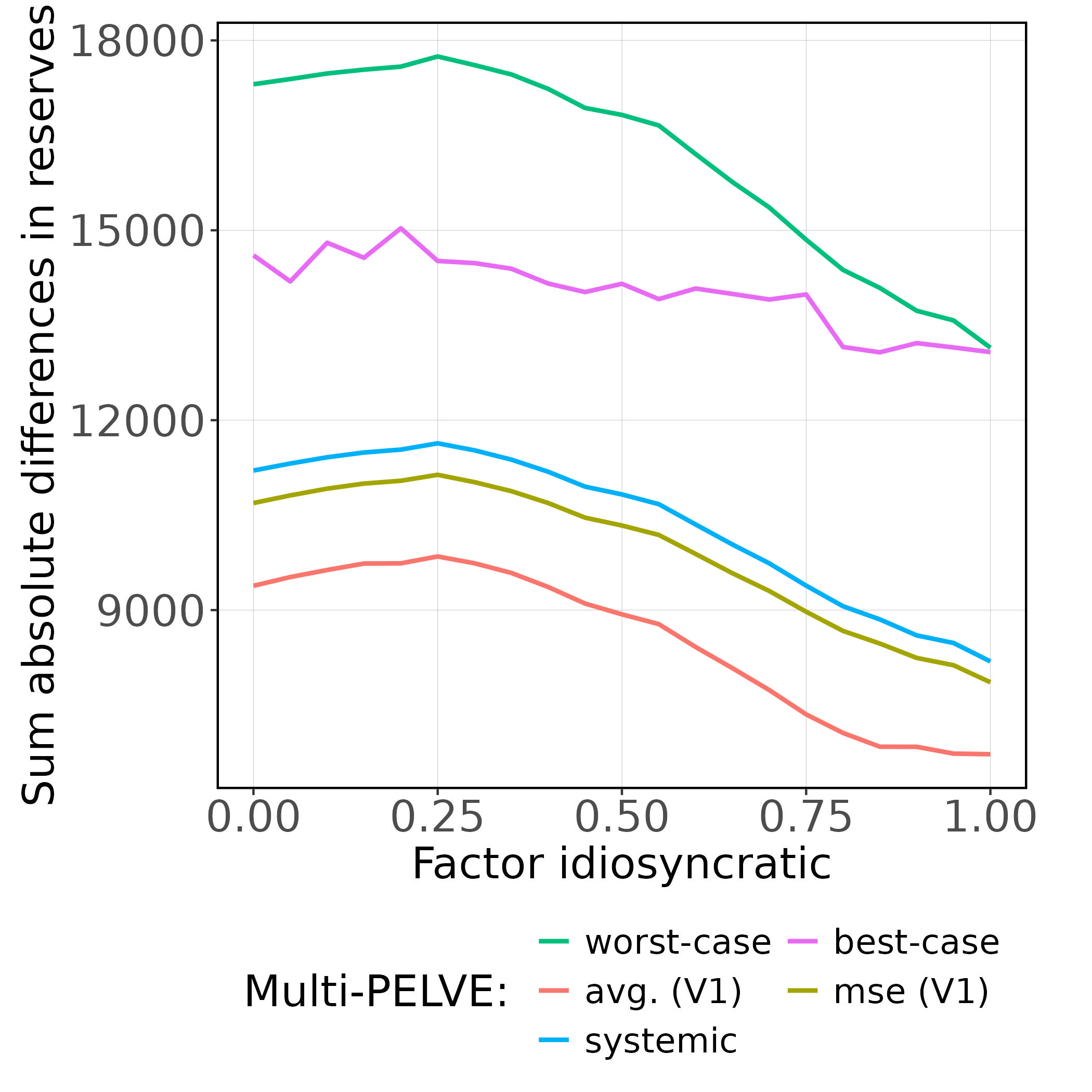}
	\end{subfigure}
	\captionsetup{font=footnotesize}
	\caption{Multi-PELVE methods for different idiosyncratic factors (LHS) and sum of absolute differences in capital reserves (RHS) in model 2. The applied VaR level is $1\%$.}
	\label{fig:pelve_diversification}
\end{figure}

Next, with the help of~\prettyref{fig:comparing_diffs_in_sum_es_var_perturbed}, we analyse the overall change in capital reserves. As for model~1, the Sys-PELVE is optimal for the total change in the sum of reserves (LHS). By plot (A), the A-PELVE requires a substantial increase in the overall capital reserves, meaning that all insurers together must provide a large amount of additional capital immediately. The MSE-PELVE is again close to the Sys-PELVE. Regarding plot (C), for levels above $5\%$, the remaining MSE-PELVE methods are close to each other; below $5\%$, the Sys- and MSE-PELVE (V1) outperform the other A- and MSE-PELVE versions in this criterion.

So far, the Sys-PELVE and the MSE-PELVE (V1) seem to be a good choice. The plots (B) and~(D) on the RHS in~\prettyref{fig:comparing_diffs_in_sum_es_var_perturbed} confirm this choice in terms of absolute changes in the capital reserves. In addition, also the MSE-PELVE (V3) leads to small values. Surprisingly, at small levels, the A-PELVE (V1) yields the smallest values. Compared to the LHS, where the MSE-PELVE (V1) is near zero -- indicating offsetting increases and decreases in insurers' reserves -- the A-PELVE (V1) results in smaller overall capital flows, with a tendency toward large reserve increases for some insurers, cf.~\prettyref{fig:capitalChanges_perturbed}~(A) below. However, beyond the $2.5\%$ level on the RHS of~\prettyref{fig:comparing_diffs_in_sum_es_var_perturbed}, the A-PELVE (V1) remains nearly constant. At these levels, both LHS and RHS show a clear preference for the Sys- and the MSE-PELVE (V1 and V3).

\begin{figure}
	\begin{subfigure}[b]{0.48\textwidth}
		\includegraphics[width=\linewidth]{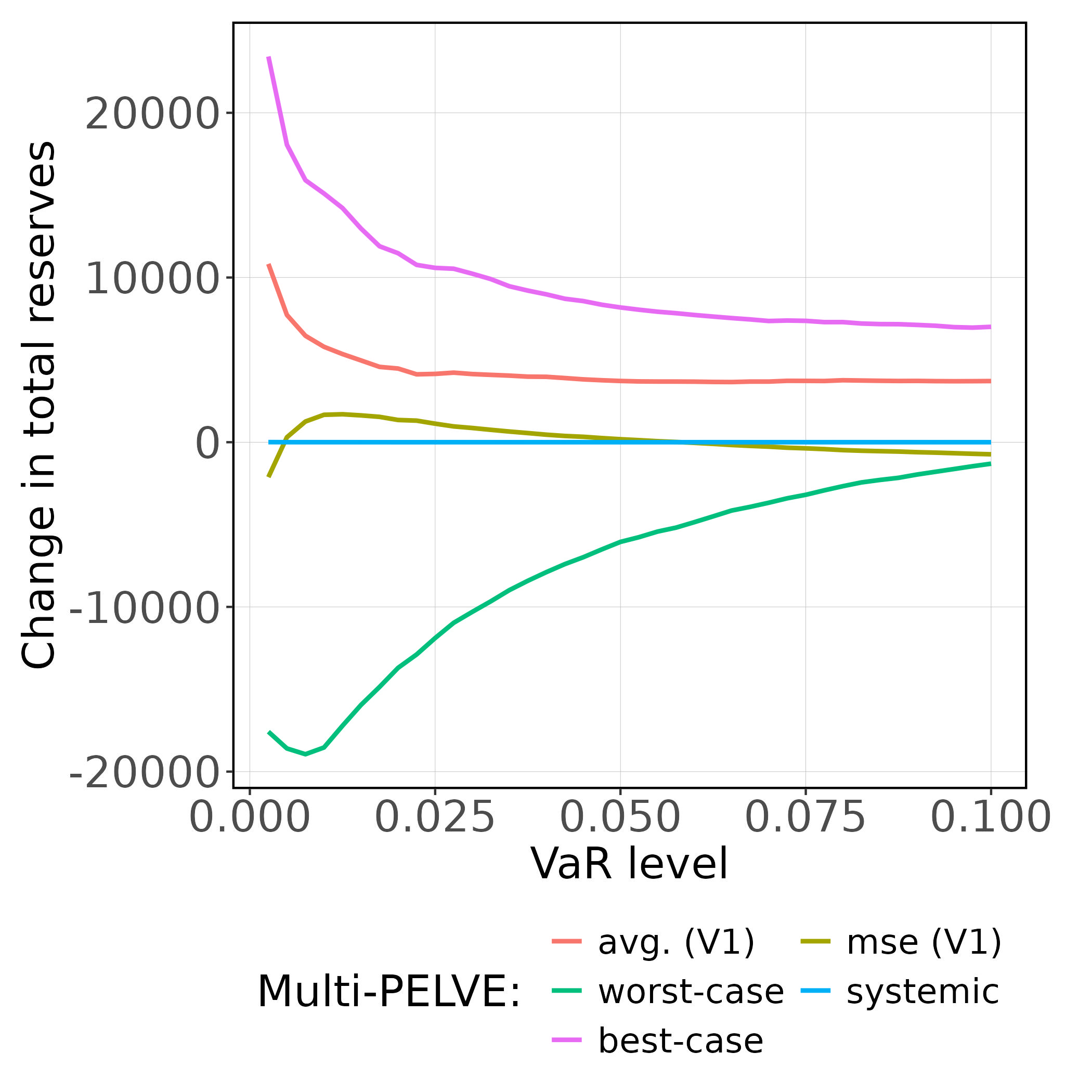}
		\caption{Total change in sum of reserves}
	\end{subfigure}
	\hfill
	\begin{subfigure}[b]{0.48\textwidth}
		\includegraphics[width=\linewidth]{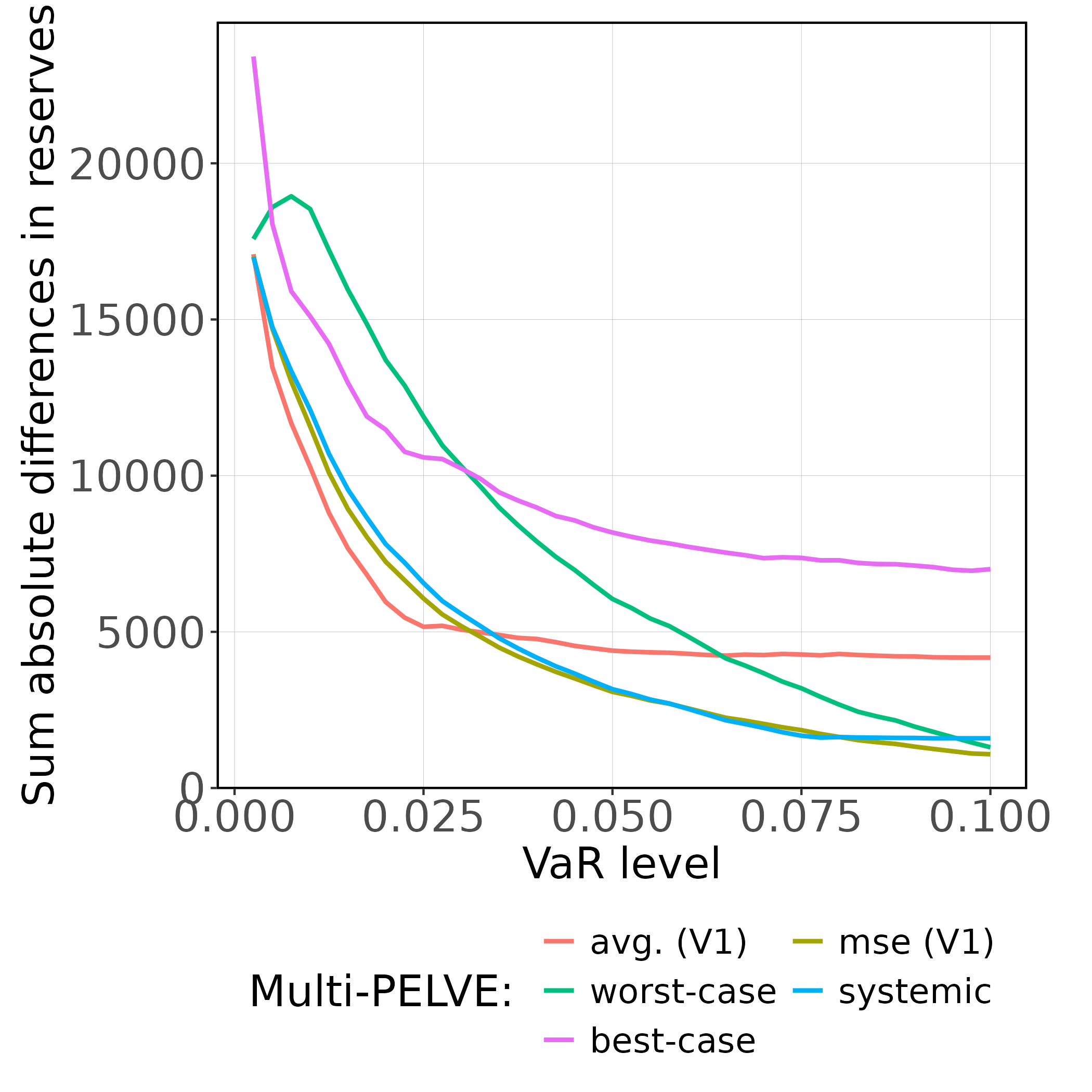}
		\caption{Sum of absolute differences in reserves}
	\end{subfigure}
	\begin{subfigure}[b]{0.48\textwidth}
		\includegraphics[width=\linewidth]{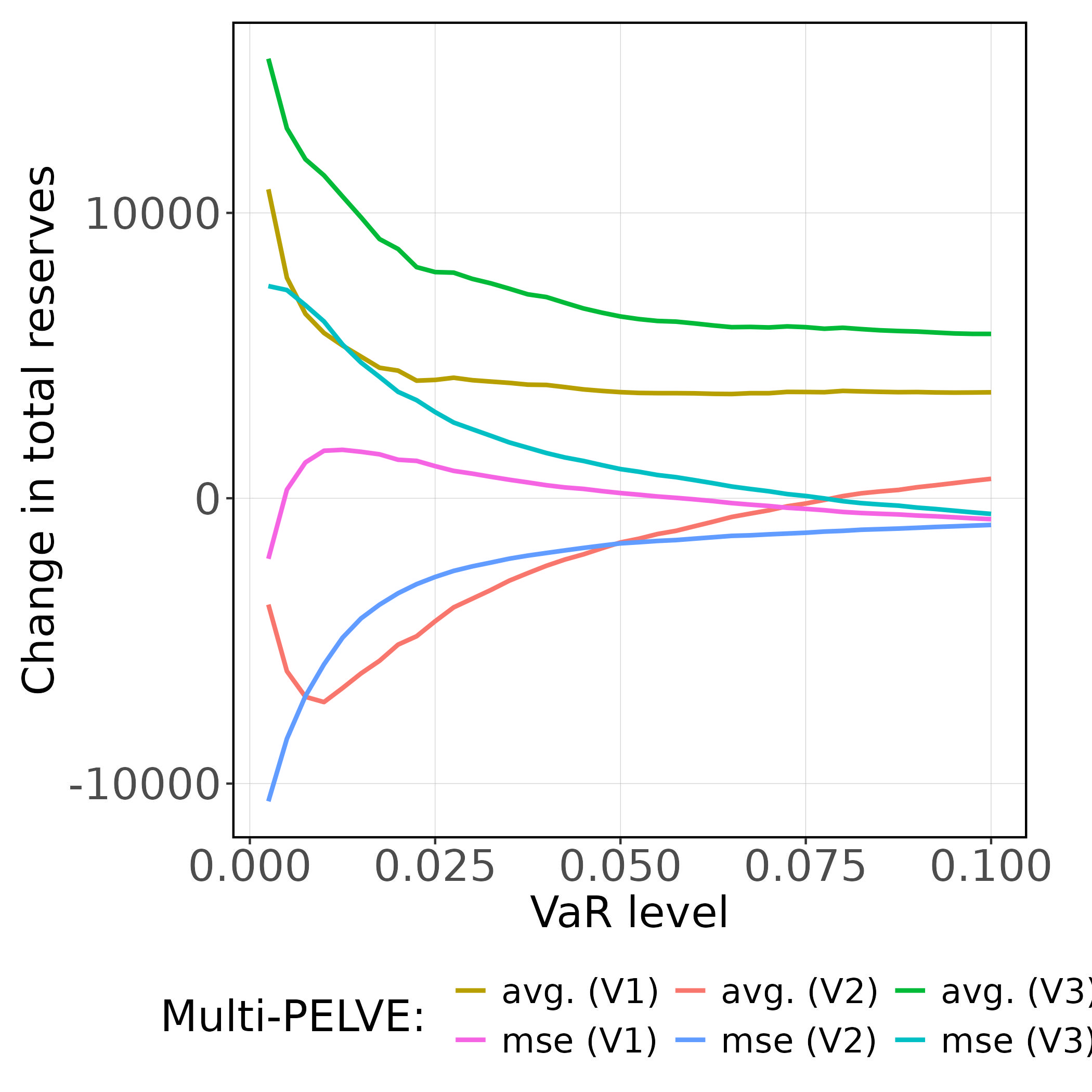}
		\caption{Comparing weighting vectors (V1, V2, V3)}
	\end{subfigure}
	\hfill
	\begin{subfigure}[b]{0.48\textwidth}
		\includegraphics[width=\linewidth]{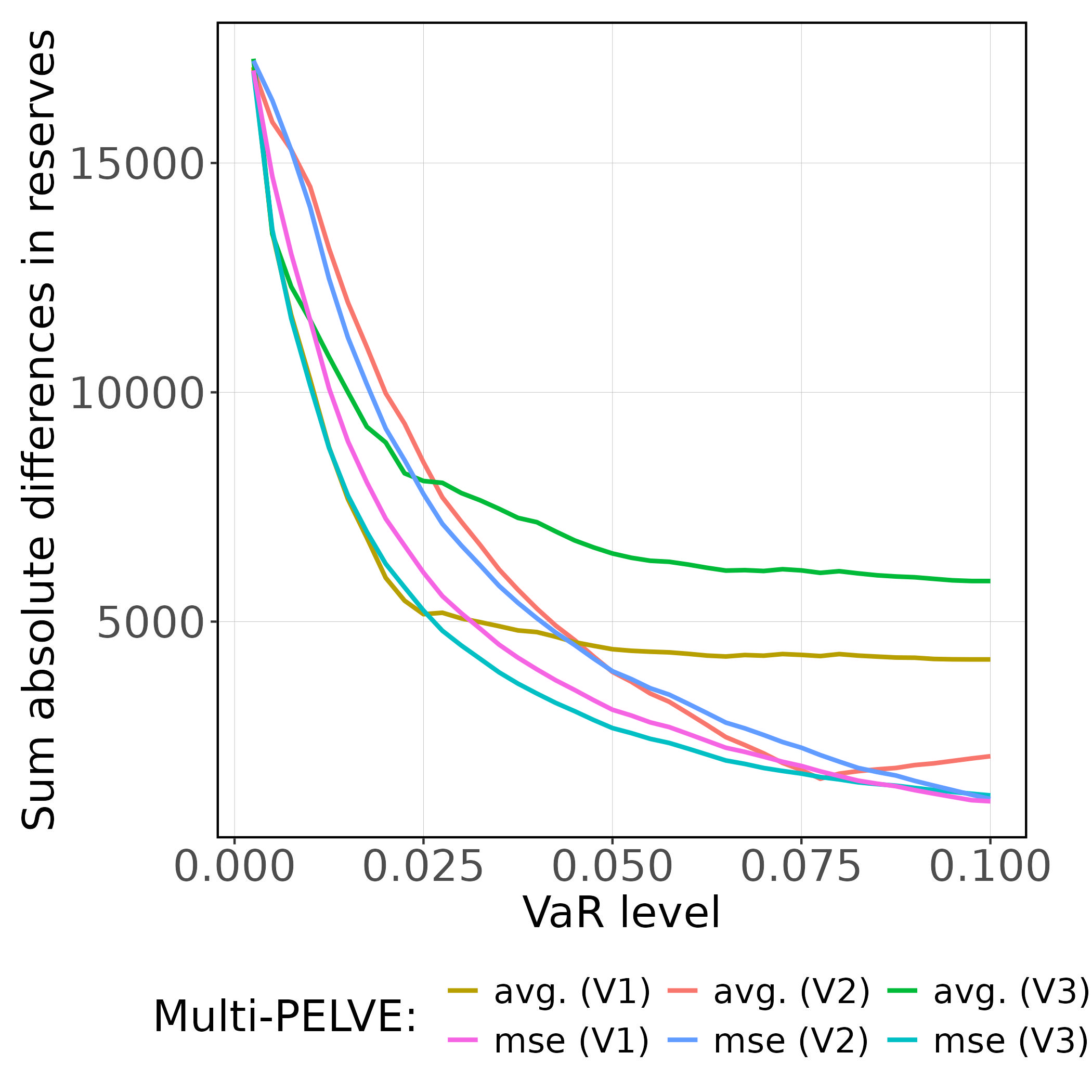}
		\caption{Comparing weighting vectors (V1, V2, V3)}
	\end{subfigure}
	\captionsetup{font=footnotesize}
	\caption{Total change in sum of capital reserves and sum of absolute differences in the capital reserves of all insurers under model 2.}
	\label{fig:comparing_diffs_in_sum_es_var_perturbed}
\end{figure}

Regarding the surprising A-PELVE graph, consider the relative changes in insurers' reserves;~\prettyref{fig:capitalChanges_perturbed}. In plot (A) -- the A-PELVE (V1) case -- the decrease for insurer $1$ is compensated by an increase for insurer $4$. Comparing the actual equity values from~\prettyref{tab:equity_insurers}, insurer $4$ has the second-largest equity ($19\%$ market share) and insurer $1$ has a market participation of $40\%$. Hence, the decrease in the curve of insurer $1$ is partly compensated by the increase in the curve of insurer $4$. 

Furthermore, the curves for insurers $2$, $3$, $5$ and $6$ in plots (A) and (E) are closer to zero than in plots (B), (C), (D) and (F). Hence, for these insurers, the A-PELVE (V1) and (V3) lead to smaller relative changes in the capital reserve than one of the other Multi-PELVE methods. This is of particular importance for VaR levels which are not close to zero, because in such cases the MSE-PELVE versions lead to significantly larger changes in the reserves for insurers $2$, $3$, $5$ and $6$. Only for small levels, the MSE-PELVE (V3) tends to be closer to the A-PELVE (V1). The largest relative deviations at small levels for insurers $2$, $3$, $5$ and $6$ are observed for the A-PELVE (V2) and the MSE-PELVE (V2), due to the strong influence of insurer $1$ ($\omega_1=0.5608$). Consequently, the versions (V2) put the focus on insurer~$1$, also visible by the relative changes for this insurer, which are significantly smaller in the plots (C) and (D) than in the plots (A), (B), (E) and (F). 

Nevertheless, this prioritization of insurer $1$ comes at the expense of the other insurers, explaining their large deviations from zero in the plots (C) and (D). In contrast, plots (E) and (F) show the largest deviations for insurer $1$. This is due to the small weight $\omega_1=0.0231$, producing the opposite effect than before, i.e.,~neglecting insurer $1$. In other words, the capital reserves of the remaining insurers are not artificially reduced by insurer $1$'s financial distress, as modeled via the GPD. So, in this sense, the versions (V3) are preferable, as they require insurer $1$ to increase reserves for its more riskier tail, rather than benefiting at the expense of the others.\newline
\begin{figure}
	\begin{subfigure}[b]{0.45\textwidth}
		\includegraphics[width=\linewidth]{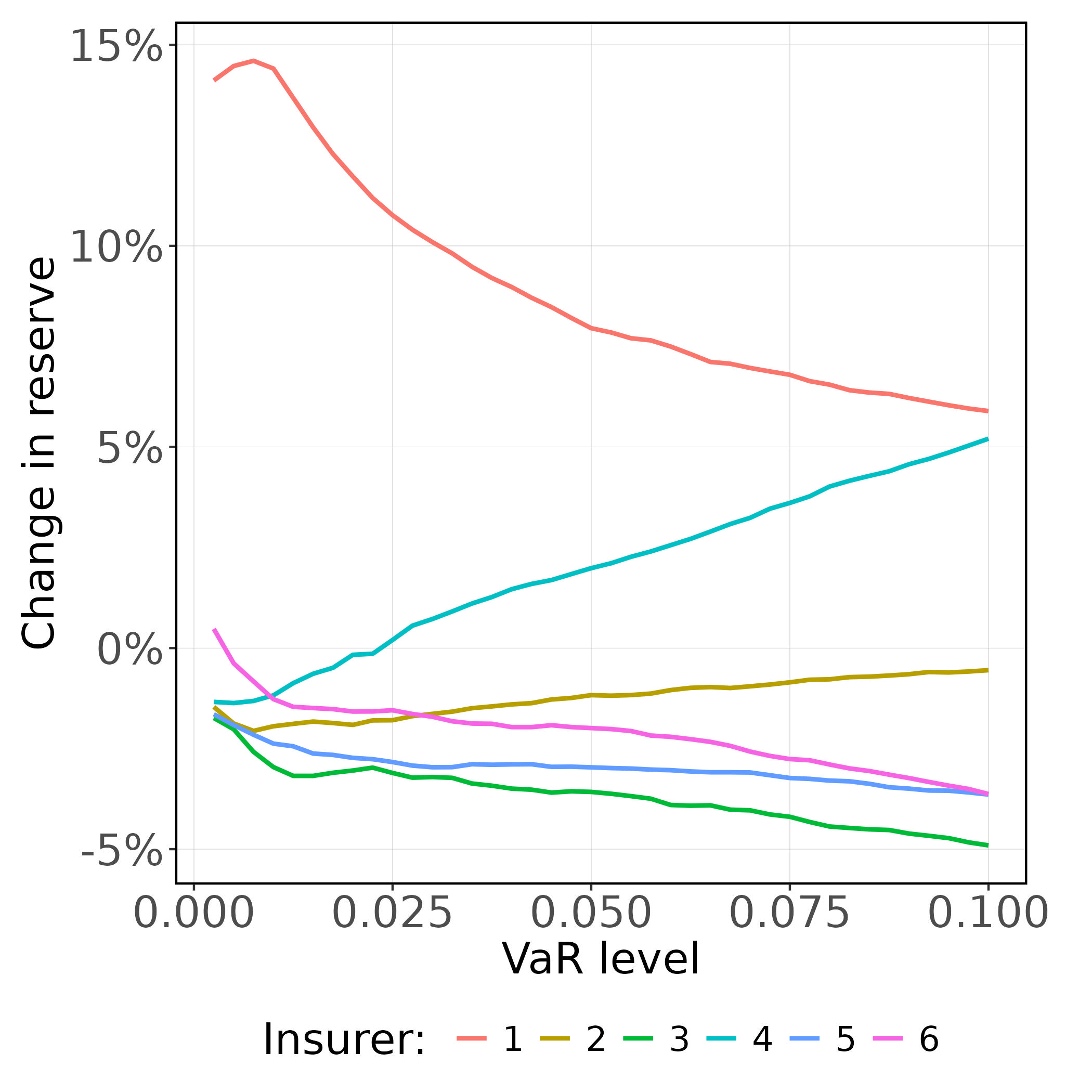}
		\caption{A-PELVE (V1)}
	\end{subfigure}
	\hfill
	\begin{subfigure}[b]{0.45\textwidth}
		\includegraphics[width=\linewidth]{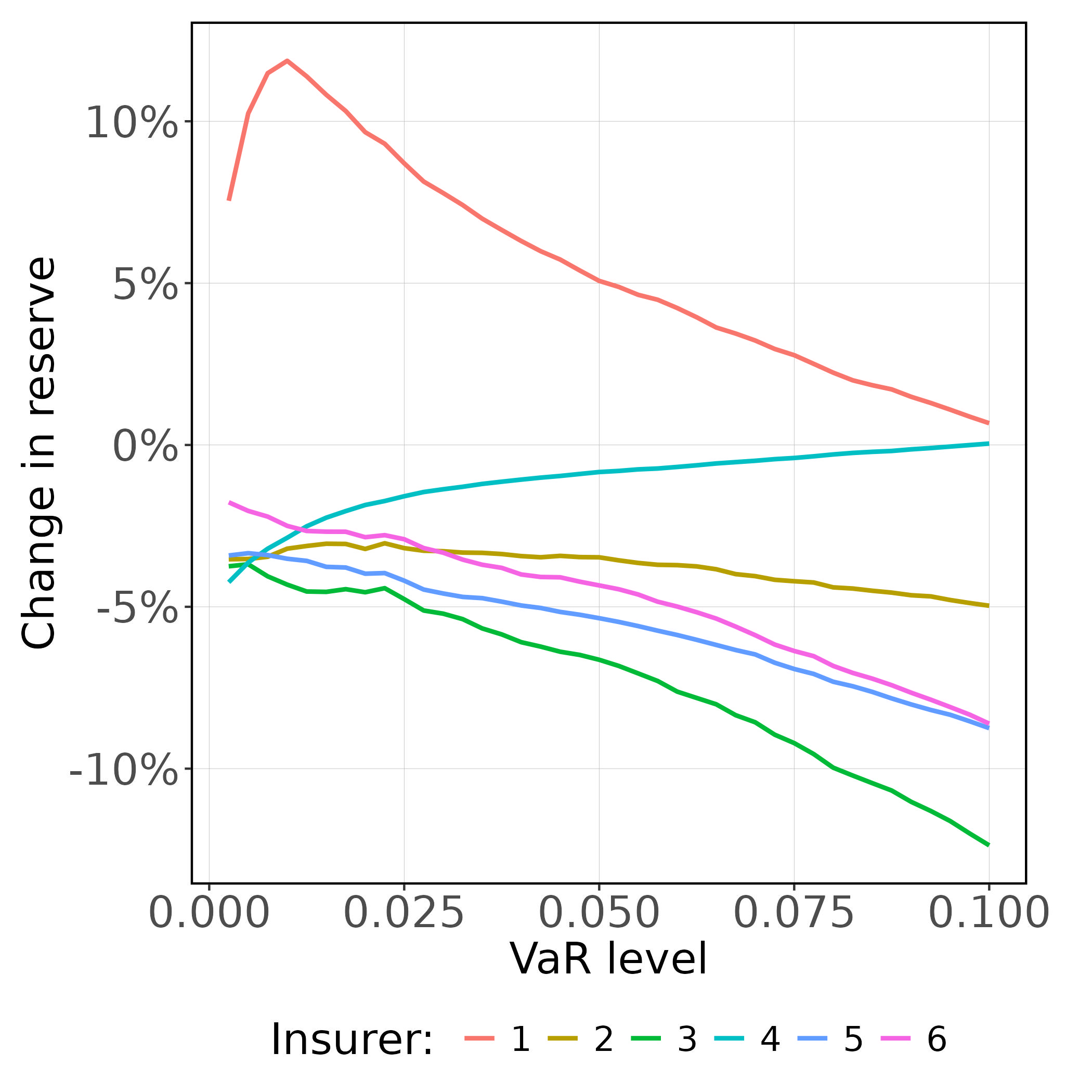}
		\caption{MSE-PELVE (V1)}
	\end{subfigure}
	\begin{subfigure}[b]{0.45\textwidth}
		\includegraphics[width=\linewidth]{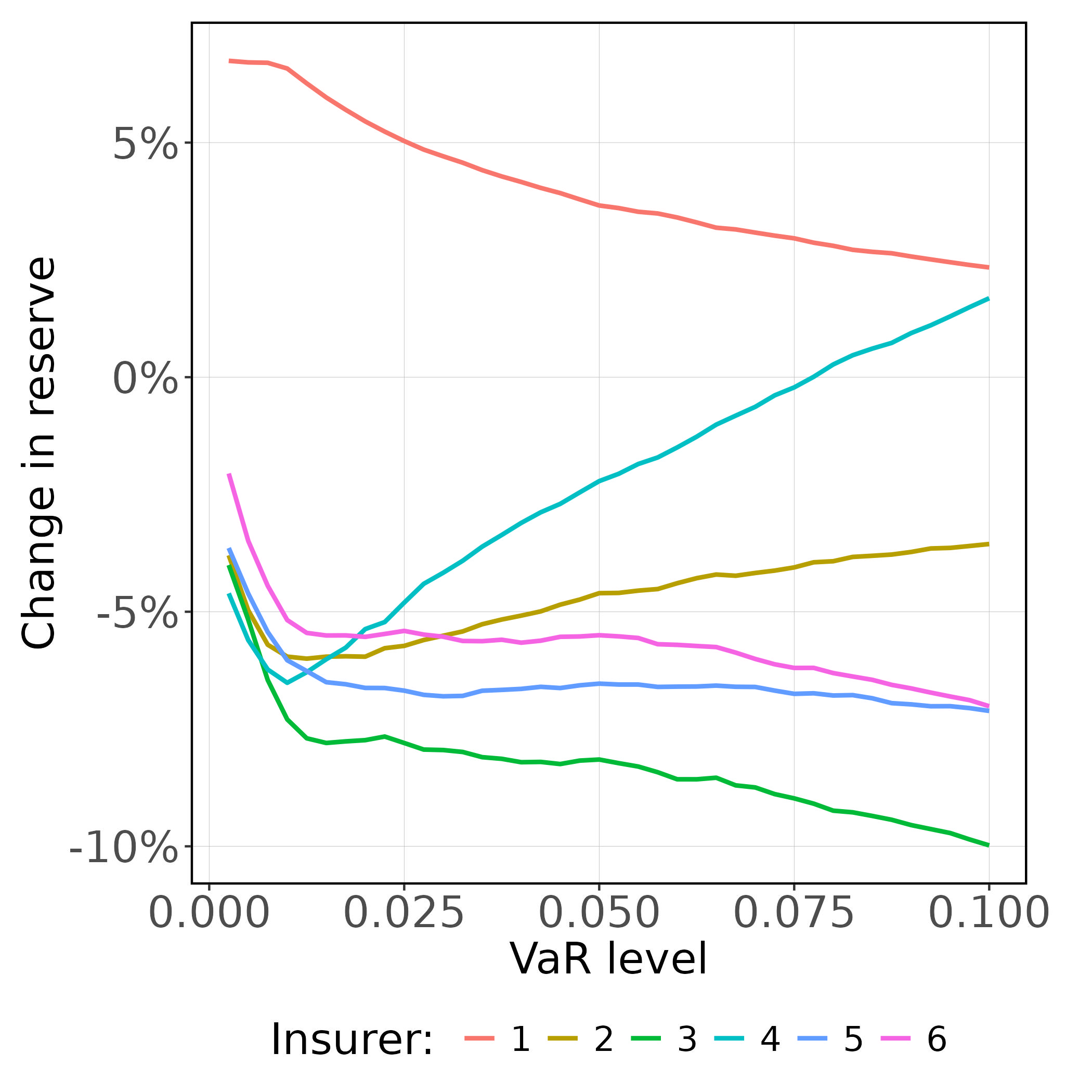}
		\caption{A-PELVE (V2)}
	\end{subfigure}
	\hfill
	\begin{subfigure}[b]{0.45\textwidth}
		\includegraphics[width=\linewidth]{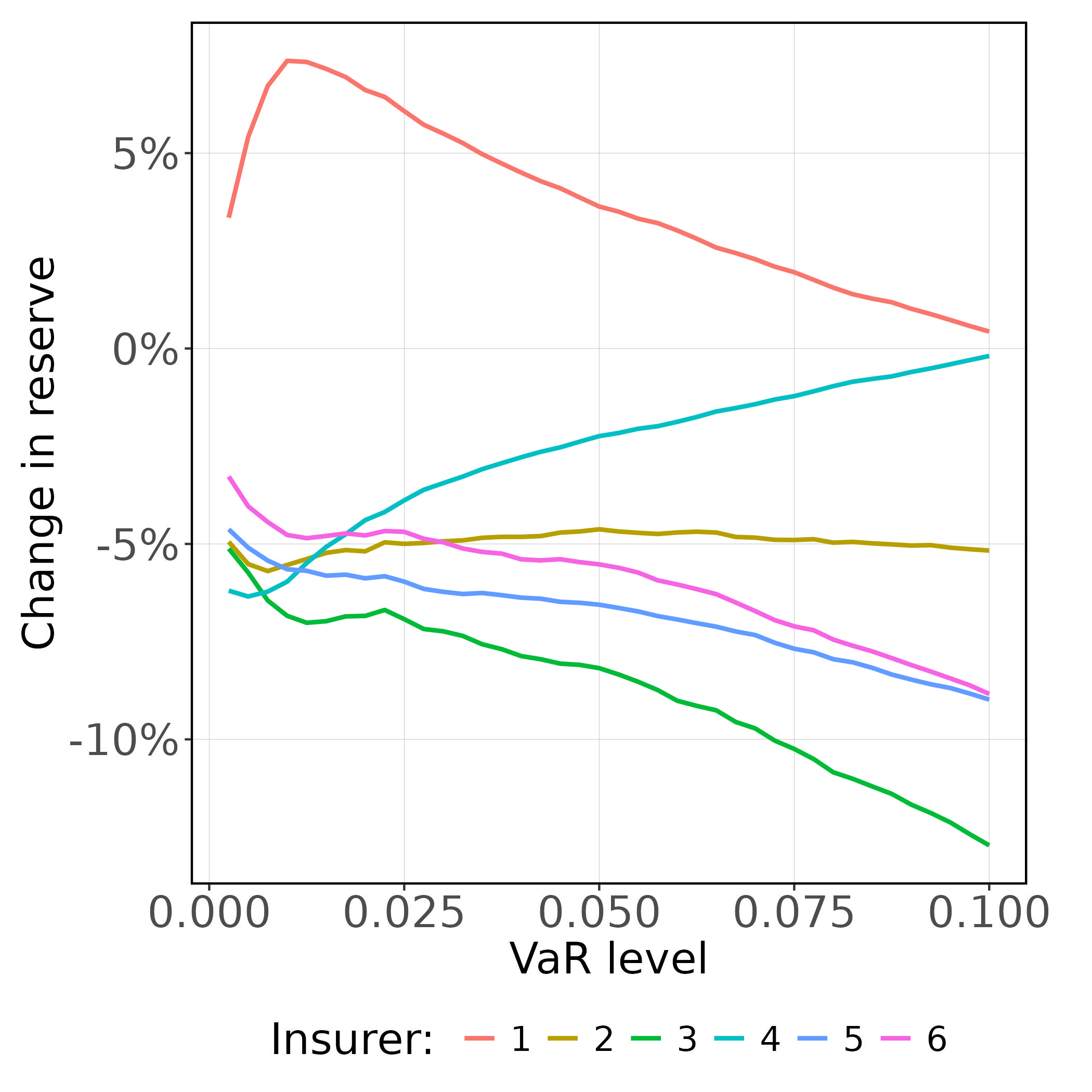}
		\caption{MSE-PELVE (V2)}
	\end{subfigure}
	\begin{subfigure}[b]{0.45\textwidth}
		\includegraphics[width=\linewidth]{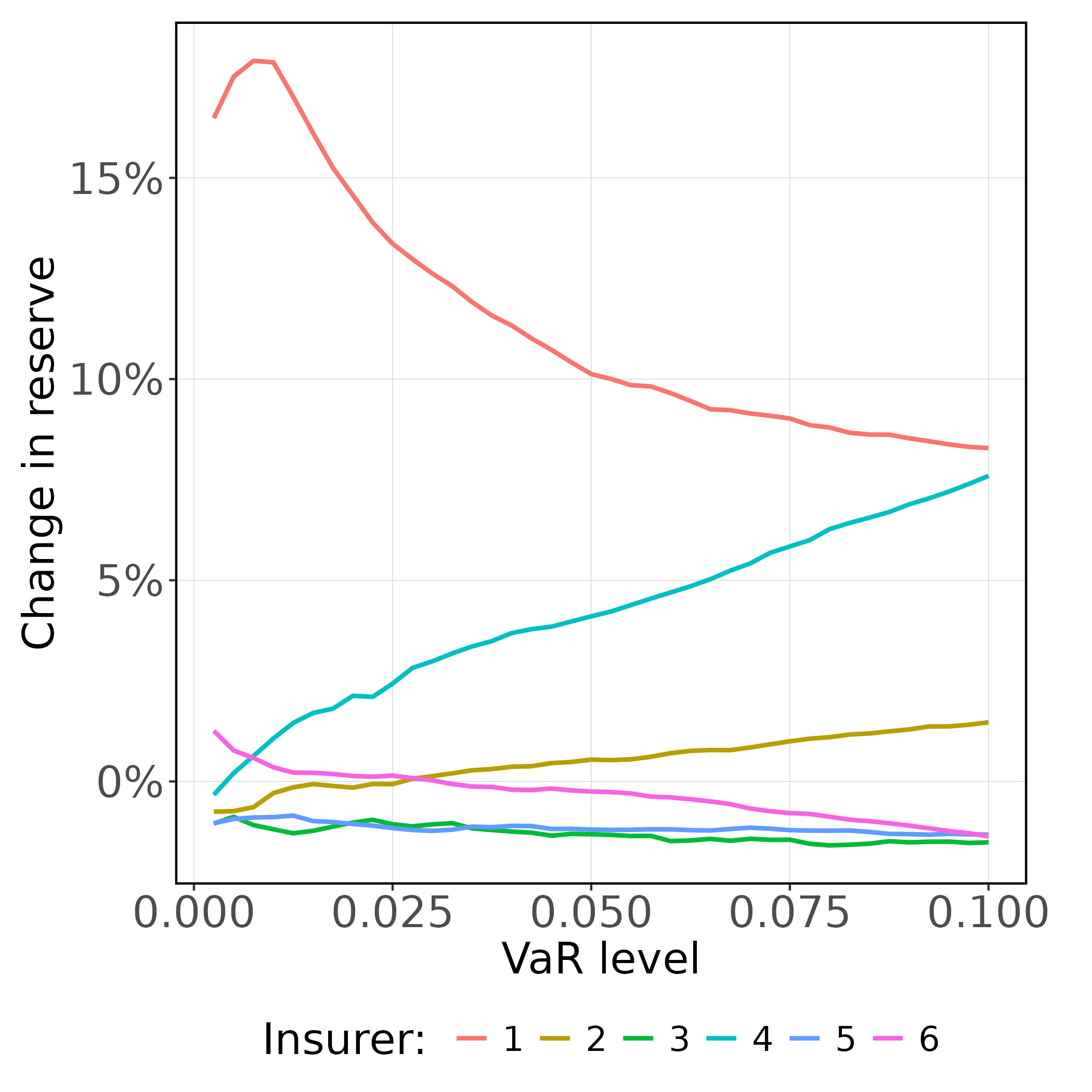}
		\caption{A-PELVE (V3)}
	\end{subfigure}
	\hfill
	\begin{subfigure}[b]{0.45\textwidth}
		\includegraphics[width=\linewidth]{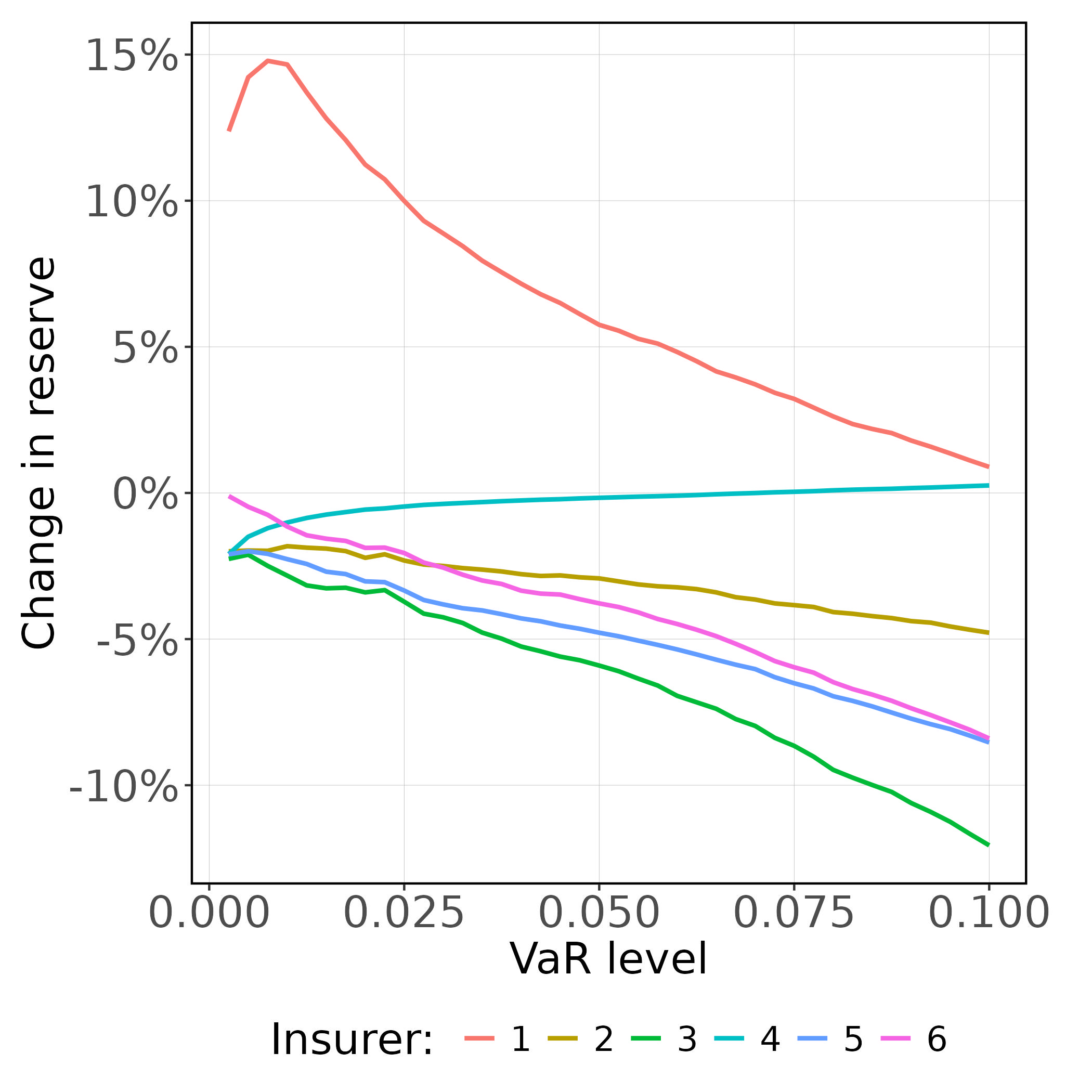}
		\caption{MSE-PELVE (V3)}
	\end{subfigure}
	\captionsetup{font=footnotesize}
	\caption{Changes in the capital reserves of the insurers under model 2, if they use the new ES level stemming from different versions (V1, V2 and V3) of A-PELVE or MSE-PELVE.}
	\label{fig:capitalChanges_perturbed}
\end{figure}

\noindent\fbox{\parbox{\textwidth}{\textit{The analysis of model 2 can be concluded by the following main observations:} Considering overall capital changes, the Sys- and MSE-PELVE (V1) and (V3) perform well. In case of VaR levels close to zero and given their strong overall performance, the A-PELVE (V1) and the MSE-PELVE (V3) seem to be good choices for model 2. For other VaR levels, the application of an MSE-PELVE is questionable because it favors insurers with heavier-tailed risk profiles.}}

\subsection{Beyond changes in regulation}\label{sec:beyondRegulatoryChanges}

Recall that the Multi-PELVE methods are intended to deal with the situation of a possible change from VaR to ES in regulation. Such a change affects the capital reserve of each insurer individually. This means that the methods rely on the marginal distributions ignoring possible dependence structures. 
	
It is natural to ask if multivariate extensions of PELVE are also helpful in other contexts. One could be the usage of risk measures to describe the risk of a network of insurers in total, i.e.,~utilize VaR or ES as management and analysis tools for a regulator. This is done by using VaR or ES acceptability criteria for systemic risk measures. In this regard, dependencies between the marginal distributions matter and should definitely not be ignored. To illustrate this possible extension, we introduce the following new version of a Sys-PELVE for $\lambda\in(0,1)$ and $\mathbf{X}\in L^1_n$:
\begin{align}\label{eq:systemic_pelve_2}
	\syspelve{\lambda,\Lambda}(\mathbf{X})\defgl \inf\{c\in[1,\lambda^{-1}]\,|\,\symbolExpectedShortfall{c\lambda}(\Lambda(\mathbf{X}))\leq\symbolValueAtRisk{\lambda}(\Lambda(\mathbf{X}))\},
\end{align}
where the map $\Lambda:\mathbb{R}^n\rightarrow \mathbb{R}$ is an aggregation rule, also called aggregation map. Here, the ES and the VaR expressions in~\eqref{eq:systemic_pelve_2} refer to systemic risk measures following the ``first aggregate, then allocate'' principle. This new Sys-PELVE gives an ES level such that the systemic risk of the bunch of insurers remains equal when switching from a VaR to an ES acceptability criteria. Note, this simplifies to a PELVE of the aggregated position, i.e.,~$\syspelve{\lambda,\Lambda}(\mathbf{X}) = \pelve{\lambda}(\Lambda(\mathbf{X}))$.

There are now several possibilities to choose an aggregation function $\Lambda$. In~\prettyref{fig:sys_pelve_versions}, we apply the following standard aggregation function (see~\cite{biagini_unified_2019, brunnermeier_measuring_2019}):
\[
	\Lambda(x) = \sum_{i=1}^{n}\alpha_i\min\{x_i,0\} + \beta\max\{x_i,0\}.
\] 
The parameters $\alpha_i$ and $\beta$ describe the preferences of the regulator. On the LHS we illustrate the impact of different values for $\beta$. The impact of $\beta$ is negligible for VaR levels in a small neighborhood of zero. Otherwise, larger values of $\beta$ reduce the Sys-PELVE value. 

On the RHS we set $\beta=0$ and choose three alternatives for the weights $\alpha_i$, namely the ones from Table~\ref{tab:optim_weights}. As expected, for version (V3), for which $\alpha_1$ has the largest value, we obtain the largest Sys-PELVE values. For comparison, we also plot the Sys-PELVE as introduced earlier in~\prettyref{sec:def_systemic_pelve}. We see that the shape of this original Sys-PELVE differs significantly from the Sys-PELVE based on ``first aggregate'' (V1). We compare these two methods by the changes in the capital reserves for the individual insurers in~\prettyref{fig:sys_pelve_capital_changes}. It is important to note that for the Sys-PELVE ``first aggregate'' (V1) we have situations for which all insurers are allowed to reduce their capital reserve, which is of course an undesirable behavior for a possible change in regulation. Hence, the main message out of  Figure~\ref{fig:sys_pelve_capital_changes} is the following: Systemic risk measures that in some way aggregate the individual positions before calculating VaR or ES are inappropriate to calculate a new ES level for a regulatory change from VaR to ES. 

\begin{figure}
	\begin{subfigure}[b]{0.48\textwidth}
		\includegraphics[width=\linewidth]{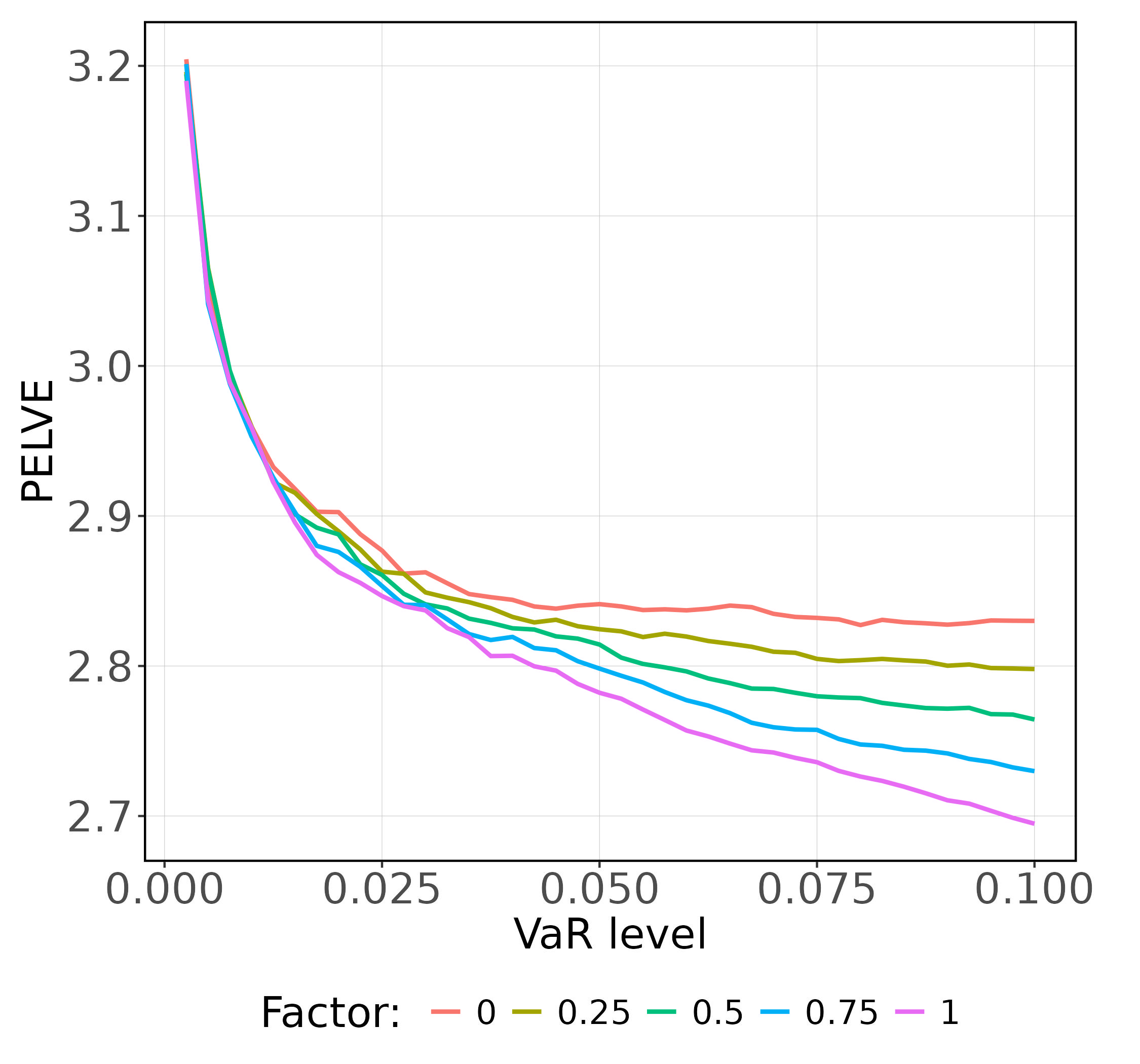}
	\end{subfigure}
	\hfill
	\begin{subfigure}[b]{0.48\textwidth}
		\includegraphics[width=\linewidth]{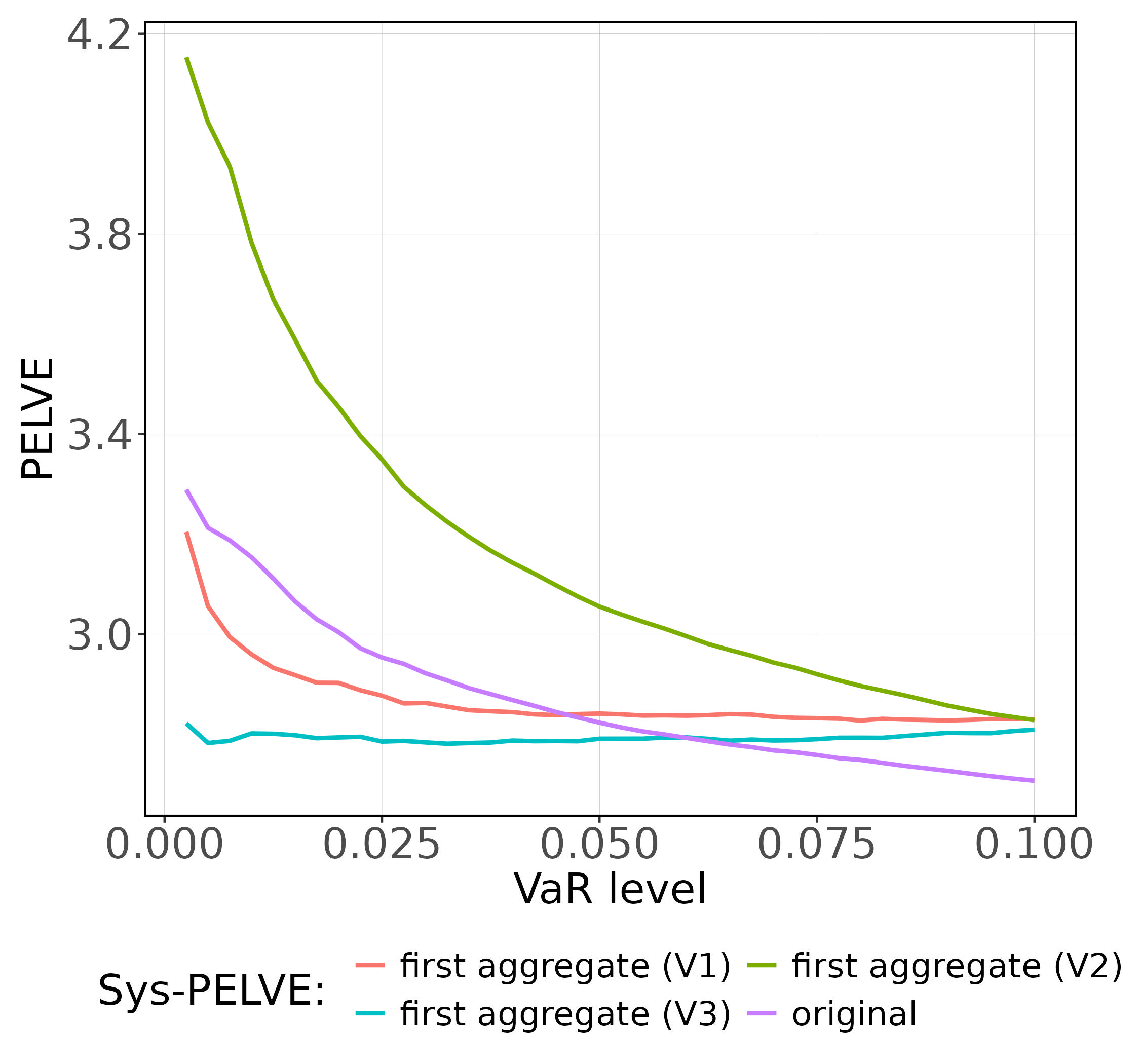}
	\end{subfigure}
	\captionsetup{font=footnotesize}
	\caption{Comparison of systemic versions of PELVE in model 2. LHS: Aggregation function $\Lambda(x) =\sum\limits_{i=1}^n \min\{x_i,0\} + \beta\max\{x_i,0\}$ is applied for different values of the factor $\beta$. RHS: Aggregation function $\Lambda(x) = \sum\limits_{i=1}^n \alpha_i\min\{x_i,0\}$ and weights $\alpha_i$ from  Table~\ref{tab:optim_weights} are used. The terminology ``original'' refers to the Sys-PELVE from Section~\ref{sec:def_systemic_pelve}. }
	\label{fig:sys_pelve_versions}
\end{figure}

\begin{figure}
	\begin{subfigure}[b]{0.48\textwidth}
		\includegraphics[width=\linewidth]{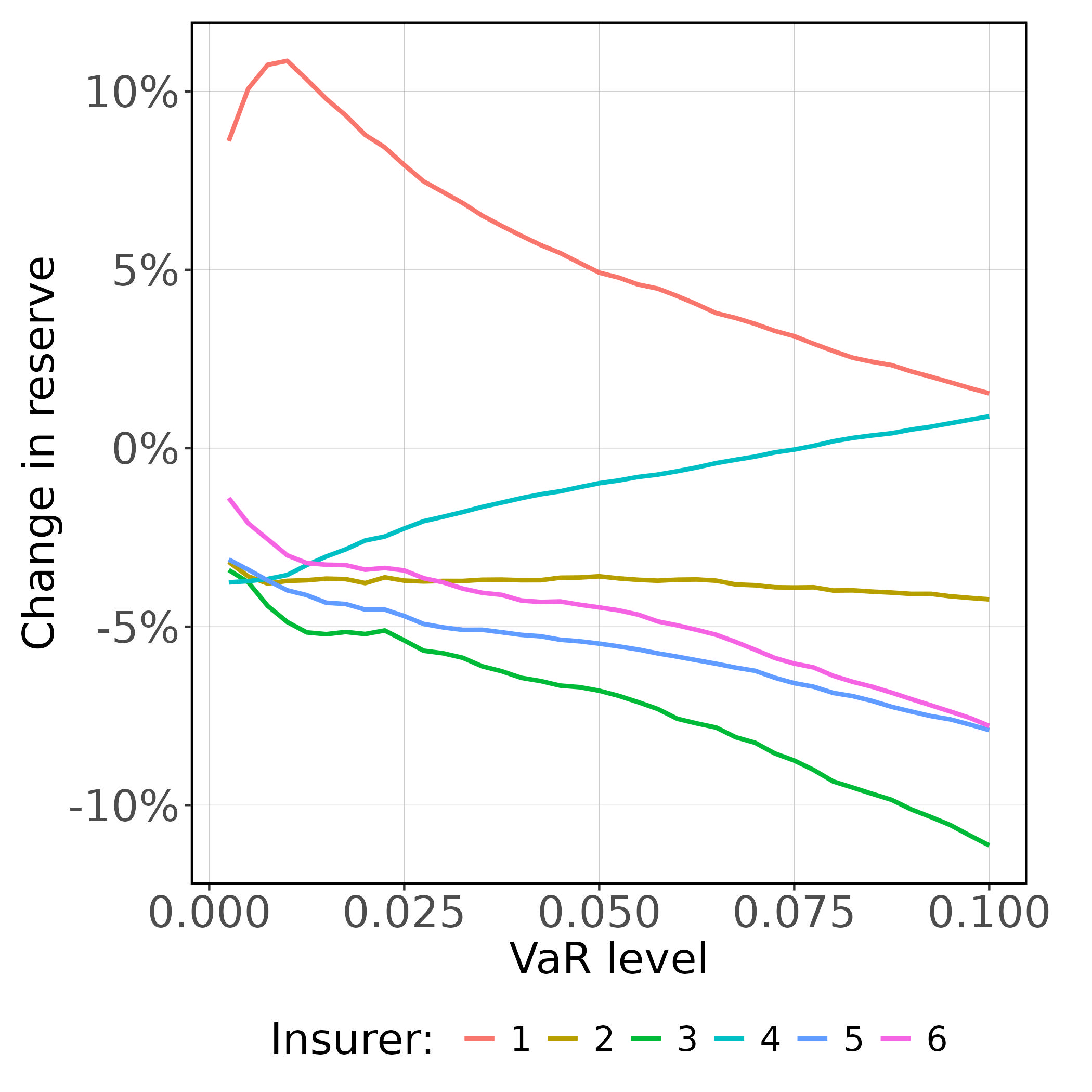}
		\caption{Sys-PELVE ``original''}
	\end{subfigure}
	\hfill
	\begin{subfigure}[b]{0.48\textwidth}
		\includegraphics[width=\linewidth]{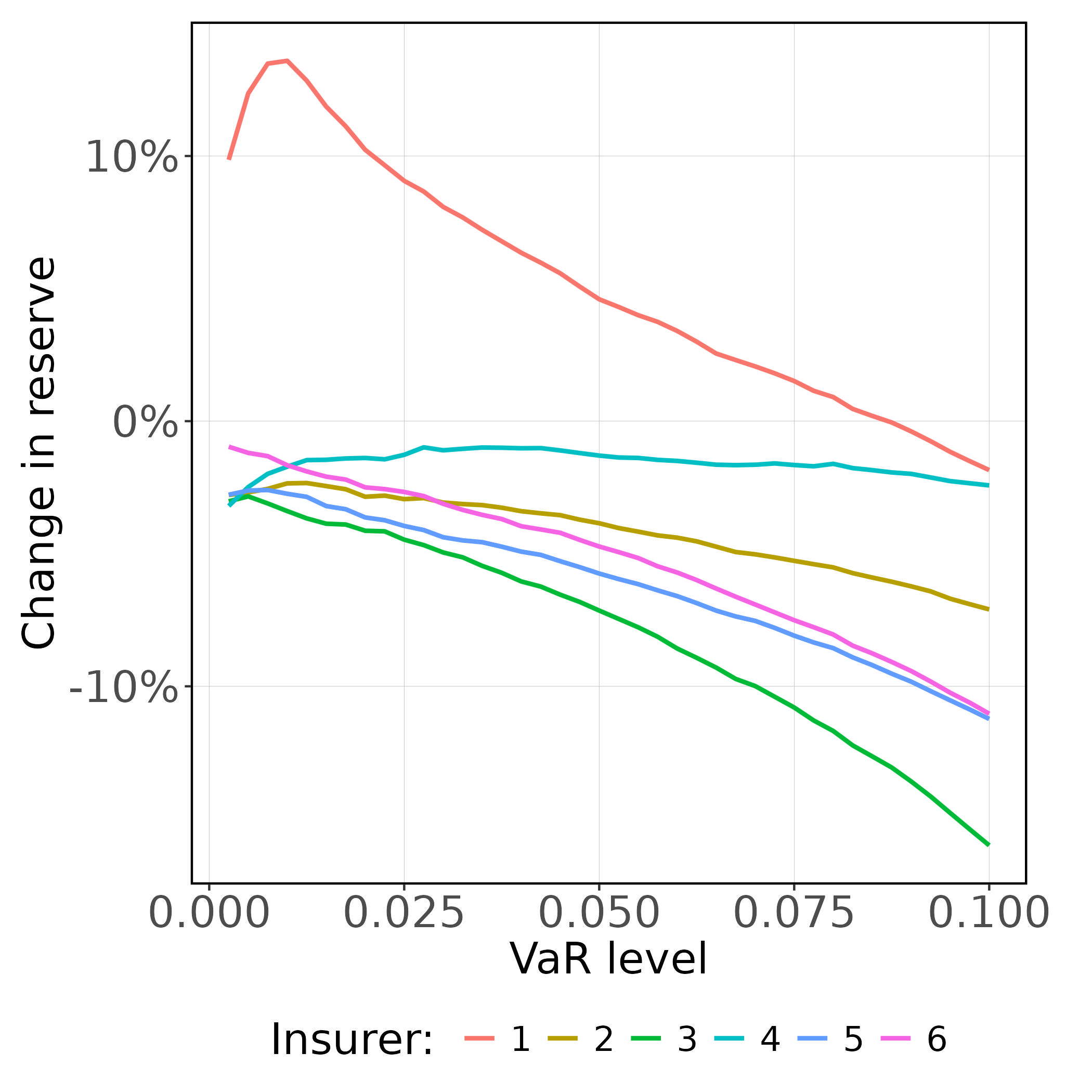}
		\caption{Sys-PELVE ``first aggregate (V1)''}
	\end{subfigure}
	\captionsetup{font=footnotesize}
	\caption{Changes in the capital reserves of the insurers under model 2, if thes use the new ES level stemming from different systemic versions of PELVE.}
	\label{fig:sys_pelve_capital_changes}
\end{figure}

%% file: sections/conclusion.tex
\section{Conclusion and outlook}

\subsection{Conclusion}

We have introduced PELVE-inspired methods to determine an appropriate ES level for the transition from VaR to ES. Unlike the original PELVE, our Multi-PELVE methods account for multiple insurers. Under suitable conditions, most of them yield unique solutions. For the MSE-PELVE, we provide a comprehensive example showing non-uniqueness in general. Continuity results are available for all methods. For elliptically distributed payoffs, most Multi-PELVE methods reduce to the PELVE of a one-dimensional elliptically distributed random variable. For multivariate regularly varying distributions, all methods converge to the PELVE of a Pareto distribution. As a by-product, we obtain that PELVE curves are discontinuous for empirical distribution functions and discuss several characterizations of ES curves.

In a case study where all equity capitals follow the same distribution family, the Multi-PELVE methods -- except the BC- and WC-PELVE -- closely match the PELVE of the individual insurers. So, moving from VaR to ES does not substantially change reserves, which is convenient for a regulator, since the ES level is independent of the chosen Multi-PELVE method. This changes drastically, if equity distributions come from different distribution families. In our case, one insurer acts as a market leader with heavy-tailed distribution. The smallest overall deviations of the reserves are obtained by the Sys- and the (weighted) MSE-PELVE. Unlike other versions, the A-PELVE (equal or inverse proportional weights with respect to asset volume) and for small VaR levels the corresponding MSE-PELVE versions do not prioritize insurers with risky equity profiles at the expense of the other insurers. Thus, in the mentioned situations, these methods seem to be a fair choice to determine a new ES level.

\subsection{Outlook}

Considering multivariate extensions of PELVE opens an unexpectedly rich landscape of future research directions, such as: (1) Testing the Multi-PELVE methods using a more advanced balance sheet model, like e.g.,~the one in~\cite{diehl_long-term_2023}. (2) Developing additional PELVE variants to address unresolved and undesirable effects, such as robustified versions. (3)~Accounting for insurers' reactions, for instance, portfolio restructuring after a transition from VaR to ES, motivated by the considerations in~\cite{embrechts_robustness_2022}. (4) Inspired by the discussion in~\prettyref{sec:beyondRegulatoryChanges}, a next step could be to develop other multivariate PELVE-inspired method to describe the difference in systemic risk measures for financial networks (of banks or (re-)insurers). Models incorporating corresponding contagion effects are of particular interest, most notably the  seminal one in~\textcite{EisenbergNoe}. Its combination with ES to construct a systemic risk measure is demonstrated in~\cite{FeinsteinRudloffWeber}.

\section*{Acknowledgements}

The authors are grateful to the two anonymous referees for their careful reading and insightful comments, which substantially improved the quality and clarity of the manuscript. The authors thank Philip Biegel for his  assistance on the proof of~\prettyref{prop:ratio_es_var_gamma}.

%% file: sections/declarations.tex
\section*{Declarations}

\textbf{Competing interests:} The authors have no competing interests to declare that are relevant to the content of this article.\newline

\noindent\textbf{Data availability:} The balance sheet figures were taken from the official reports of the corresponding insurance companies. Apart from these publicly available numbers, no additional data were used in this contribution; simulated numbers are considered for illustrative purposes only.

%% file: sections/appendix.tex
\section{Results and proofs accompanying Section~\ref{sec:esCurves}}\label{sec:auxiliaryEScurves}

\begin{lem}\label{lem:nonConcavityEsCurve}
	Let $X\in L^1$. The map $(0,1]\rightarrow\mathbb{R},t\mapsto t\expectedShortfall{t}{X}$ is concave and continuous.
\end{lem}

\begin{proof}
	Concavity follows since the function $h:(0,1]\ni t\mapsto-t\expectedShortfall{t}{X}$ is a restriction of the Fenchel-Legendre transformation of the function $g(x):=E[(x-X)^{+}]$ (with domain $\mathbb{R}$), see~\cite[Lemma A.26]{follmer_stochastic_2016}. Then, $g\not\equiv\infty$ implies that $h$ is convex (and lower semicontinuous). Continuity of $h$ is a consequence of dominated convergence. 
\end{proof}

In contrast to the previous lemma, \prettyref{exam:nonConcavityES} provides a random variable $X\in L^1$ such that $(0,1]\rightarrow\mathbb{R},t\mapsto \expectedShortfall{t}{X}$ is neither concave nor convex. 
\begin{exam}\label{exam:nonConcavityES}
	Consider the random variable $X$ with cumulative distribution function:
	\begin{align*}
		F^X(x) = \sum_{i=1}^{3} p_i 1_{[x_i,\infty)}(x),
	\end{align*}
	by choosing the concrete values $x_i=i$ and $p_i = 3^{-1}$ for all $i\in[3]$, we obtain:
	\begin{align*}
		f(t):=\expectedShortfall{t}{X} = \begin{cases}
			0, & t\in\left(0,\frac{1}{3}\right],\\
			\frac{1}{3t}-1, & t\in\left(\frac{1}{3},\frac{2}{3}\right],\\
			-2+\frac{1}{t},& t\in\left(\frac{2}{3},1\right].
		\end{cases}
	\end{align*}
	Note, $f$ is differentiable on $(0,1]\setminus\left\{\frac{1}{3},\frac{2}{3}\right\}$ and it holds that
	\begin{align*}
		f^{\prime}(t)= \begin{cases}
			0, & t\in\left(0,\frac{1}{3}\right),\\
			-\frac{1}{3t^2}, & t\in\left(\frac{1}{3},\frac{2}{3}\right),\\
			-\frac{1}{t^2},& t\in\left(\frac{2}{3},1\right],
		\end{cases}\quad f^{\prime\prime}(t)= \begin{cases}
			0, & t\in\left(0,\frac{1}{3}\right),\\
			\frac{2}{3t^3}, & t\in\left(\frac{1}{3},\frac{2}{3}\right),\\
			\frac{2}{t^3},& t\in\left(\frac{2}{3},1\right].
		\end{cases}
	\end{align*}
	This shows that $f$ is strictly convex on $\left(\frac{1}{3},1\right]\setminus\left\{\frac{2}{3}\right\}$. Furthermore, the limits of the left and right derivative of $f$ in $\frac{1}{3}$ are 
	\begin{align*}
		\lim\limits_{t\uparrow\frac{1}{3}}f^{\prime}(t) = 0,\quad \lim\limits_{t\downarrow\frac{1}{3}}f^{\prime}(t) = -3.
	\end{align*}
	This shows that $f$ is not locally convex in $\frac{1}{3}$.  
\end{exam}

\prettyref{exam:nonConcavityES} also shows that ES curves are not necessarily strictly decreasing. Accordingly, assumptions guaranteeing this property are important. In this regard, the next result presents a slight refinement of the statement in~\prettyref{rem:existenceUniquenessPELVE}.  For completeness, we include a proof of this tailored version, which is used in the proof of~\prettyref{prop:continuityPELVE}. 
\begin{lem}\label{lem:esStrictlyDecreasing}
	Let $X\in L^1$ and $\alpha\in(0,1]$. Assume there exists  $\gamma\in[0,\alpha)$ such that for all $\lambda\in(\gamma,\alpha)$ it holds that $\symbolUpperQuantile{X}$ is not constant on $(0,\lambda)$. Then, $(\gamma,1]\rightarrow \mathbb{R},t\mapsto \expectedShortfall{t}{X}$ is strictly decreasing.
\end{lem}

\begin{proof}
	Assume $\beta,\kappa\in(\gamma,1]$ with $\kappa <\beta$. Then we obtain:
	\begin{align*}
		\expectedShortfall{\beta}{X} &= \frac{1}{\beta}\int\limits_{0}^{\beta}\text{VaR}_{u}(X)\,\diff u =\frac{\kappa}{\beta} \expectedShortfall{\kappa}{X}+ \frac{\beta-\kappa}{\beta}\left(\frac{1}{\beta-\kappa}\int\limits_{\kappa}^{\beta}\text{VaR}_{u}(X)\,\diff u\right)\\
		&\leq \frac{\kappa}{\beta} \expectedShortfall{\kappa}{X}+ \frac{\beta-\kappa}{\beta}\text{VaR}_{\kappa}(X)\\
		&<\frac{\kappa}{\beta} \expectedShortfall{\kappa}{X}+ \frac{\beta-\kappa}{\beta}\expectedShortfall{\kappa}{X} = \expectedShortfall{\kappa}{X}.
	\end{align*} 
	The strict inequality is obtained as follows: By the assumption that for all $\lambda\in(\gamma,\alpha)$ the map $\symbolUpperQuantile{X}$ is not constant on $(0,\lambda)$, we obtain that there exists  $\eta\in(0,\kappa)$ such that $\valueAtRisk{\eta}{X}>\valueAtRisk{\kappa}{X}$ and we conclude that 
	\begin{align*}
		\valueAtRisk{\kappa}{X} &\leq  \frac{1}{\kappa}\left(\int\limits_{0}^{\eta}\text{VaR}_{\kappa}(X)\,\diff u + \int\limits_{\eta}^{\kappa}\text{VaR}_{\kappa}(X)\,\diff u \right)\\ &<\frac{1}{\kappa}\left(\int\limits_{0}^{\eta}\text{VaR}_{u}(X)\,\diff u + \int\limits_{\eta}^{\kappa}\text{VaR}_{\kappa}(X)\,\diff u \right)\leq \expectedShortfall{\kappa}{X}.\qedhere
	\end{align*}
\end{proof}

\begin{proof}[Proof of~\prettyref{thm:existenceOfEScurve}.]
	We first prove that (a) implies (b). We have $$\lim\limits_{t\downarrow 0}\tES(t) = \lim\limits_{t\downarrow 0}tf(t) = -\lim\limits_{t\downarrow 0}\int_{0}^{t}q_X^{+}(u)\diff u =0,$$
	by integrability of $X$, which is equivalent to $\int_0^1|\upperQuantileFunction{X}{u}|\diff u<\infty$. Then, for any finite collection of disjoint intervals $(a_1,b_1),\dots, (a_n,b_n)\subseteq [0,1]$, we get
	\[
	\sum_{k=1}^n \big| \tES(b_k) - \tES(a_k) \big|
	= \sum_{k=1}^n \left| \int_{a_k}^{b_k} q_X(u)\,du \right|
	\le \sum_{k=1}^n \int_{a_k}^{b_k} |q_X(u)|\,du
	= \int\limits_{\bigcup_{k=1}^{n} (a_k,b_k)} |q_X(u)|\,du.
	\]
	Together with $X\in L^1$, i.e.,~$\int_0^1|\upperQuantileFunction{X}{u}|\diff u<\infty$, we obtain that $\tES$ is absolutely continuous. By $\tES$ being absolutely continuous, the fundamental theorem of calculus for Lebesgue integrals implies that~$-\tES^{\prime}$ exists a.e.~on $[0,1]$ and it is a.e.~equal to $\symbolUpperQuantile{X}$, which means we can find a null set $A\subseteq [0,1]$ such that $-\tES^{\prime}(t) = \upperQuantileFunction{X}{t}$ for all $t\in[0,1]\setminus A$. This also implies that $-\tES^{\prime}$ is increasing and right-continuous on $[0,1]\setminus A$. 
	
	Now, we prove that (b) implies (a). Note, $(0,1]\setminus A$ is dense in $(0,1]$. Indeed, let $I = (a,b)\subseteq (0,1]$ be nonempty. If $I\subseteq A$, then $A$ cannot be a null set. Otherwise, $I\setminus A \neq \emptySet$, which means $I$ and $(0,1]\setminus A$ cannot be disjoint. This shows that $(0,1]\setminus A$ is dense in $(0,1]$. 
	
	For all $t\in(0,1]\setminus A$ set $q(t):=-\tES^{\prime}(t)$. For $t\in A$ set $$q(t):=\lim\limits_{\substack{u\downarrow t\\ u\in (0,1]\setminus A}} q(u).$$
	
	Then, by the fundamental theorem of calculus for the Lebesgue integral we obtain 
	\begin{align*}
		-\int_{0}^{t}q(u)\diff u = \int_{(0,t]\setminus A}(-q(u))\diff u= \int_{[0,t]\setminus A}\tES^{\prime}(u)\diff u = \tES(t)-\tES(0) = \tES(t).
	\end{align*}
	
	So, if $q$ satisfies all properties of an upper quantile function, then we are done. To prove this, note first that $q$ being increasing is a direct consequence of the assumption that $-\tES^{\prime}$ is increasing on $(0,t]\setminus A$. To show that $q$ is right-continuous, we point out that for all $t\in(0,1]$ it holds that
	\begin{align*}
		\lim\limits_{\substack{u\downarrow t\\ u\in (0,1]\setminus A}} q(u) = \lim\limits_{u\downarrow t} q(u). 
	\end{align*}
	This equality is derived as follows: If necessary, subdivide a decreasing sequence $(u_n)_n$ with $u_n\downarrow t$ into two subsequences $(v_n)_n\subseteq (0,1]\setminus A$ and $(w_n)_n\subseteq A$. The case of $(v_n)_n$ is then clear. For the case of $(w_n)_n$, note that for each $n$ we can create sequences $(w_n^l)_n,(w_n^u)_n\subseteq (t,1]\setminus A$ converging to~$t$ such that for each $n$ it holds that $w_n^l< w_n< w_n^u$. This then shows $q(w_n)\downarrow q(t)$. Finally, for $t\in A$, $$q(t) =  \lim\limits_{\substack{u\downarrow t\\ u\in (0,1]\setminus A}} q(u) = \lim\limits_{u\downarrow t} q(u),$$ where the first inequality holds by definition of $q$ in $t\in A$, as well as by the fact that for all $t\in(0,1]\setminus A$ we obtain
	\begin{align*}
		\lim\limits_{u\downarrow t} q(u) = \lim\limits_{\substack{u\downarrow t\\ u\in (0,1)\setminus A}} (-\tES^{\prime}(u))=-\tES^{\prime}(t) = q(t), 
	\end{align*}
	where the second equality follows by the right-continuity of $-\tES^{\prime}$ on $(0,1]\setminus A$. 
\end{proof}
	
\begin{proof}[Proof of~\prettyref{cor:conclusionTheoremEScurves}]
	Towards a contradiction, assume there exist $0<b<c\leq 1$ and $\delta>0$ such that $f$ is constant on $(b,c]$ and strictly decreasing on $(\max\{0,b-\delta\},b)$. Note, $\lim\limits_{u\uparrow b} (-\tES^{\prime}(u)) = f(b)-b\lim\limits_{u\uparrow b}f^{\prime}(u)$. 
	If $\lim\limits_{u\uparrow b}f^{\prime}(u) <0$, then it is obvious that the function $(0,1]\setminus A\rightarrow\mathbb{R},t\mapsto -\tES^{\prime}(t)$ cannot be increasing, contradicting the properties of $\tES$ in~\prettyref{thm:existenceOfEScurve}.  
	The case $\lim\limits_{u\uparrow b}f^{\prime}(u) =0$ is only possible if $f$ is convex in a small neighborhood of $b$, i.e.~convex on $(-b-\epsilon,b+\epsilon)\subseteq(0,1)$ for some $\epsilon>0$. Now, $t\mapsto -\tES^{\prime}(t)$ being increasing on $(0,1]\setminus A$ implies that for all $t\in(0,b)\setminus A$ it has to hold that 
	\begin{align*}
		-f(t)-tf^{\prime}(t)\leq -f(b) \quad\Longleftrightarrow\quad f(b)-f(t) - tf^{\prime}(t)\leq 0.
	\end{align*}
	So, for $t\in(b-\epsilon,b)$ we obtain from the convexity of $f$ on $(-b-\epsilon,b+\epsilon)$ that
	\begin{align*}
		f(b)-f(t) - t\frac{f(b)-f(t)}{b-t}\leq 0 \quad\Longleftrightarrow\quad 	 \frac{b-2t}{b-t}(\underbrace{f(b)-f(t)}_{<0}) \leq 0.
	\end{align*}
	But, for $t\in(b-\epsilon,b)$ close enough to $b$, the left-hand side in this inequality becomes strictly positive, a contradiction.
\end{proof}

\begin{proof}[Proof of~\prettyref{cor:existenceOfEScurve}]
	For all $t\in(0,1]$ set $\tES(t):=tf(t)$. Then, $f$ being decreasing and concave on $(0,1]$ implies $f$ being bounded. Hence, we can extend $\tES$ to a continuously differentiable function on $[0,1]$, by setting $\tES(0):=\lim\limits_{u\downarrow 0}\tES(u) = 0$. This extended version of $\tES$ to the compact interval $[0,1]$ is continuously differentiable and therefore, it is also absolutely continuous. In particular, $-\tES^{\prime}(t) = -tf^{\prime}(t)-f(t)$ is continuous on $(0,1]$. By $f$ being decreasing and concave, we have for all $t\in(0,1)$ that $-2f^{\prime}(t)-tf^{\prime\prime}(t)\geq 0$, i.e.,~$-\tES^{\prime}$ is increasing on $(0,1)$. The claim follows then by~\prettyref{thm:existenceOfEScurve}.
\end{proof}

\begin{exam}\label{exam:counterxampleCorollaryEScurve}
	Let $f:(0,1]\rightarrow\mathbb{R},t\mapsto (1-t)^{2}$. Then, for all $t\in(0,1]$ we have $f^{\prime}(t) = 2t-2$ and $f^{\prime\prime}(t) = 2$. Hence, $2f^{\prime}(t)+tf^{\prime\prime}(t) = 6t-4$, which is strictly positive for all $t\in\left(\frac{2}{3},1\right]$. So, $t\mapsto -f(t)-tf^{\prime}(t)$ is decreasing on $\left(\frac{2}{3},1\right]$.
\end{exam}

For brevity of the next proof, we denote the restriction of a function $f$ to a set $A$ by $f|_{A}$.

\begin{proof}[Proof of~\prettyref{prop:continuityPELVE}.]
	We first prove (a) implies (b), by a proof of contradiction. To do so, we assume first that there exists $\lambda\in(0,\alpha)$ such that $\expectation{}{-X}>\valueAtRisk{\lambda}{X}$. By~\cite[Proposition 1]{li_pelve_2023} we obtain $f(\lambda)=\pelve{\lambda}(X)=\infty$, which contradicts the finiteness of $f$. 
	
	Then, assume that for all $\lambda\in (0,\alpha)$ we have $\expectation{}{-X}\leq \valueAtRisk{\lambda}{X}$ and $q_X^{+}$ is discontinuous at a point $\gamma\in(0,\alpha)$. By~\cite[Lemma A.19]{follmer_stochastic_2016} we know that  $\symbolUpperQuantile{X}$ is right-continuous. Hence, $\lim\limits_{\lambda\uparrow \gamma}\upperQuantileFunction{X}{\lambda}<\upperQuantileFunction{X}{\gamma}$. Now, by~\cite[Proposition 2]{li_pelve_2023} we obtain that $c^{\star}:=\pelve{\gamma}(X)$ is the unique value such that $\expectedShortfall{c^{\star}\gamma}{X} = \valueAtRisk{\gamma}{X}$. By the continuity of $f$, we get for each sequence $(\gamma_n)_n\subseteq (0,\alpha)$ with $\gamma_n\uparrow\gamma$ that $c_n:=\pelve{\gamma_n}(X)\rightarrow c^{\star}$. In total,  $c_n\gamma_n\rightarrow c^{\star}\gamma$ and $$\valueAtRisk{\gamma}{X}=\expectedShortfall{c^{\star}\gamma}{X}=\lim\limits_{n\rightarrow\infty}\expectedShortfall{c_n\gamma_n}{X}=\lim\limits_{n\rightarrow\infty}\valueAtRisk{\gamma_n}{X},$$ 
	where the second equality follows by the continuity of  $(0,1]\ni\lambda\mapsto\expectedShortfall{\lambda}{X}$ and the last equality by~\cite[Proposition 1]{li_pelve_2023}. But this contradicts $\lim\limits_{\lambda\uparrow \gamma}\upperQuantileFunction{X}{\lambda}<\upperQuantileFunction{X}{\gamma}$, i.e.,~the discontinuity of $\symbolUpperQuantile{X}$ in $\gamma$. 
	
	Now, we prove that (b) implies (a). By assumption for all $\lambda\in(0,\alpha)$ it holds $\expectation{}{-X}\leq \valueAtRisk{\lambda}{X}$, which implies by \cite[Proposition 1]{li_pelve_2023} that for each such $\lambda$ there exists $c\in[1,\lambda^{-1}]$ with $\expectedShortfall{c\lambda}{X}=\valueAtRisk{\lambda}{X}$. The latter implies that $f$ is finite-valued.
	
	Next, we prove continuity of $f$. To do so, assume first that  $\symbolUpperQuantile{X}$ is constant on $(0,\gamma)$ for some $\gamma\in(0,\alpha]$, then the map $(0,\gamma)\ni\lambda\mapsto \expectedShortfall{\lambda}{X}$ is constant (with value $-\upperQuantileFunction{X}{\lambda}$). This implies that $\pelve{\lambda}(X) = 1$ for all $\lambda\in(0,\gamma)$, which means that $f|_{(0,\gamma)}$ is continuous. In particular, if $\gamma=\alpha$, then we are done. Furthermore, if $\gamma <\alpha$, then the continuity of $\symbolUpperQuantile{X}$ also implies that $f|_{(0,\gamma]}$ is constant.
	
	Second, if there exists $\gamma\in[0,\alpha)$ such that for all $\lambda\in(\gamma,\alpha)$, $\symbolUpperQuantile{X}$ is not constant on $(0,\lambda)$, then for all $\beta\in(\gamma,\alpha)$, we obtain from~\cite[Proposition 2]{li_pelve_2023} that $\pelve{\beta}(X)$ is the unique value in $c\in[1,\beta^{-1}]$ for which $\expectedShortfall{c\beta}{X} = \valueAtRisk{\beta}{X}$. So, assuming a sequence $(\lambda_n)_n\subseteq (\gamma,\alpha)$ with $\lambda_n\rightarrow \lambda^{\star}$ for some $\lambda^{\star}\in(\gamma,\alpha)$ and setting $c_n\defgl \pelve{\lambda_n}(X)$ and $c^{\star}\defgl\pelve{\lambda^{\star}}(X)$, then
	\begin{align*}
		\lim\limits_{n\rightarrow\infty}\expectedShortfall{c_n\lambda_n}{X} = \lim\limits_{n\rightarrow\infty}\valueAtRisk{\lambda_n}{X}=\valueAtRisk{\lambda^{\star}}{X} = \expectedShortfall{c^{\star}\lambda^{\star}}{X}.
	\end{align*}
	
	Now, by~\prettyref{lem:esStrictlyDecreasing}, we know that the map  $(\gamma,1]\ni\lambda\mapsto \expectedShortfall{\lambda}{X}$ is strictly decreasing. Further, by the continuity of this map, it has to hold that $c_n\lambda_n\rightarrow c^{\star}\lambda^{\star}$, which in turn implies $c_n\rightarrow c^{\star}$. The latter means that $f|_{(\gamma,\alpha)}$ is continuous. If $\gamma=0$, the desired result follows.
	
	Third, assume that there exists $\gamma\in(0,\alpha)$ such that $\symbolUpperQuantile{X}$ is constant on $(0,\gamma]$ and for all $\lambda\in(\gamma,\alpha)$ it holds that $\symbolUpperQuantile{X}$ is not constant on $(0,\lambda)$. Then, it remains to show that $f$ is right-continuous at $\gamma$. To prove this, let $(\lambda_n)_n\subseteq(\gamma,\alpha)$ be a sequence with $\lambda_n\rightarrow \gamma$. Then (recall $c_n\defgl \pelve{\lambda_n}(X)$),
	\begin{align*}
		\lim\limits_{n\rightarrow\infty}\expectedShortfall{c_n\lambda_n}{X} = \lim\limits_{n\rightarrow\infty}\valueAtRisk{\lambda_n}{X}=\valueAtRisk{\gamma}{X} = \expectedShortfall{\gamma}{X}.
	\end{align*}
	The result $c_n\rightarrow 1$ follows by the fact that the map $(0,1]\ni\lambda\rightarrow\expectedShortfall{\lambda}{X}$ is continuous in combination with~\prettyref{lem:esStrictlyDecreasing}.
\end{proof}

\begin{proof}[Proof of~\prettyref{cor:discontinuityPELVEcurves}.]
	This follows from~\prettyref{prop:continuityPELVE} in combination with the fact that a constant area of $F^{X}$ transfers to a jump of $\symbolUpperQuantile{X}$.
\end{proof}

\section{Proofs accompanying Section~\ref{sec:problemRegulator}}\label{sec:exampleMsePelve}

\begin{proof}[Proof of~\prettyref{prop:existenceMsePELVE}.]
	Let $Y\in L^1$. The map $(0,1)\mapsto \mathbb{R},p\mapsto \symbolExpectedShortfall{p}(Y)$ is continuous. This implies that~\prettyref{eq:mse} is also continuous. Together with the fact that $[1,\lambda^{-1}]$ is compact, the extreme value theorem implies that the map~\prettyref{eq:mse} attains a minimum. 
\end{proof}

For the proofs of~\prettyref{thm:convergenceMSE} and~\prettyref{prop:convergenceSys}, we need the following auxiliary result. It utilizes slightly different assumptions than the classical Dini Theorem~\parencite[Theorem 7.13]{rudinPrinciplesMathematicalAnalysis1976}. Instead of assuming a monotone sequence of functions, every function is assumed to be decreasing. This necessitates a distinct approach in order to establish the result.

\begin{lem}\label{lem:alternativeDini}
	Let $[a,b]\subset \mathbb{R}$ with $a,b\in\mathbb{R}$ such that $a<b$. Suppose $(f_n)_n$ is a sequence of functions $f_n:[a,b]\rightarrow \mathbb{R}$ such that every $f_n$ is decreasing on $[a,b]$. If $(f_n)_n$ converges pointwise to a continuous function $f:[a,b]\rightarrow\mathbb{R}$, then the convergence is uniform. 
\end{lem}

\begin{proof}
	By a proof of contradiction, assume that the convergence is not uniform. This means, there exist $\epsilon>0$, a subsequence $(f_{n_k})_k$ and corresponding points $x_{k}\in[a,b]$ with 
	\begin{align*}
		|f_{n_k}(x_k)-f(x_k)|\geq \epsilon\quad \text{for all } k\in\mathbb{N}.
	\end{align*} 
	By passing to a subsequence of $(x_k)_k$ if necessary, we can assume that $x_k\rightarrow x^{\star}$ for some $x^{\star}\in[a,b]$. W.l.o.g.~assume that $f_{n_k}(x_k)\geq f(x_k)+\epsilon$ occurs infinitely often. If $x^{\star}>a$, choose an arbitrary $t\in[a,x^{\star})$. Then, for $k$ large enough we have $x_k>t$ and hence,
	\begin{align*}
		f_{n_k}(t)\geq f_{n_k}(x_k)\geq f(x_k)+\epsilon.
	\end{align*}
	Pointwise convergence of $(f_n)_n$ to $f$ and continuity of $f$ give $f(t)\geq f(x^{\star})+\epsilon$, contradicting continuity of $f$. If $x^{\star} = a$, the same line of arguments leads $f(a)\geq f(a)+\epsilon$, again a contradiction.
\end{proof}

\begin{proof}[Proof of~\prettyref{thm:convergenceMSE}.]
	We perform this proof in three steps.
	
	\textit{Step 1:} For $i\in[n]$ the function
	\begin{align*}
		f_m^{i}:[\lambda,1]\rightarrow \mathbb{R},t\mapsto \expectedShortfall{t}{X_i^m}-\valueAtRisk{\lambda}{X_i^m}
	\end{align*}
	is decreasing and
	\begin{align*}
		f^i:[\lambda,1]\rightarrow \mathbb{R},t\mapsto \expectedShortfall{t}{X_i}-\valueAtRisk{\lambda}{X_i}
	\end{align*}
	is continuous. Then, assumption (a) guarantees that $(f_m^i)_m$ converges pointwise to $f^i$, which is by~\prettyref{lem:alternativeDini} even a uniform convergence.
	 
	Now, set $u_m(t):=(f_m^1(t),\dots,f_m^n(t))^{\intercal}$, $u(t):=(f^1(t),\dots,f^n(t))^{\intercal}$ and $\Phi:\mathbb{R}^n\rightarrow\mathbb{R},x\mapsto \sqrt{\sum_{i=1}^{n}x_i^2}$. Note, $\Phi$ is continuous.
	
	\textit{Step 2:} Now, we show that $\Phi\circ u_m$ uniformly converges to $\Phi\circ u$ as $m\rightarrow\infty$. To do so, note that by Step 1, we have
	\begin{align}\label{eq:proofContinuityMSE}
		\sup_{x\in[\lambda,1]}\max_{i\in[n]}\big(f^{i}_m(x)-f^{i}(x)\big)\xrightarrow{m\rightarrow\infty} 0.
	\end{align}
	Together with compactness of $[\lambda,1]$ and $u_m$ being continuous for all $m\in\mathbb{N}$, we can apply Theorem~4.14 and Theorem~7.12 in~\cite{rudinPrinciplesMathematicalAnalysis1976}. Hence, there exists a compact set $K\subseteq\mathbb{R}^n$ such that $u([\lambda,1])\subseteq K$ and $u_m([\lambda,1])\subseteq K$ for all $m\in\mathbb{N}$. For this $K$ we obtain by the Heine-Cantor theorem~\parencite[Corollary 3.31]{aliprantis_infinite_2006} that the restricted map $\Phi|_{K}$ is uniformly continuous. Now, let $\epsilon>0$ be arbitrary. By uniform continuity of $\Phi|_{K}$ choose $\delta>0$ such that for all $y,z\in K$ with $|y-z|_d<\delta$ it holds that $|\Phi(x)-\Phi(z)|<\epsilon$. In addition, choose $N\in\mathbb{N}$ such that the expression in~\prettyref{eq:proofContinuityMSE} is smaller than $\delta$ for all $m\geq N$. Thus, for all $t\in[\lambda,1]$ and $m\geq N$ we obtain that $|(\Phi\circ u_m)(x)-(\Phi\circ u)(x)|<\epsilon$.	
	
	\textit{Step 3:} Finally, we prove that $(c_m)_m:=(\msepelve{\lambda,\omega}(\mathbf{X}^m))_m$ converges to $c^{\star}:=\msepelve{\lambda,\omega}(\mathbf{X})$. Since $(c_m)_m\subseteq [1,\lambda^{-1}]$, there exists a convergent subsequence $(c_{m_k})_k$ such that $c_{m_k}\rightarrow c$ as $k\rightarrow\infty$ for some~$c\in[1,\lambda^{-1}]$. Now, for each $k$, note that $c_{m_k}\lambda$ is a global minimizer of $\Phi\circ u_{m_k}$, which implies
	\begin{align*}
		(\Phi\circ u_{m_k})(c_{m_k}\lambda)\leq (\Phi\circ u_{m_k})(c^{\star}\lambda).
	\end{align*}
	By uniform convergence of $(\Phi\circ u_{m_k})_k$ to $\Phi\circ u$ (Step 2), we get that 
	\begin{align*}
		(\Phi\circ u)(c\lambda)\leq (\Phi\circ u)(c^{\star}\lambda).
	\end{align*}
	Then, assumption (b) implies that $c = c^{\star}$. Hence, every convergent subsequence of $(c_m)_m$ converges to $c^{\star}$, which shows by compactness of $[\lambda,1]$ that $c_m\rightarrow c^{\star}$. 
\end{proof}

\begin{proof}[Proof of~\prettyref{prop:convergenceSys}.]
	For each $m\in\mathbb{N}$, the function
	\begin{align*}
		h_m:[\lambda,1]\rightarrow \mathbb{R},t\mapsto \sum_{i=1}^{n}g\left(\expectedShortfall{t}{X_i^m}\right)-g\left(\valueAtRisk{\lambda}{X_i^m}\right)
	\end{align*}
	is decreasing and the function 
	\begin{align*}
		h:[\lambda,1]\rightarrow \mathbb{R},t\mapsto \sum_{i=1}^{n}g\left(\expectedShortfall{t}{X_i}\right)-g\left(\valueAtRisk{\lambda}{X_i}\right)
	\end{align*}
	is continuous. Assumption (a) guarantees that $(h_m)_m$ converges pointwise to $h$. \prettyref{lem:alternativeDini}, implies that this convergence is even uniform. From assumption (c), we obtain for $m$ sufficiently large that $\sum_{i=1}^{n}g(\expectation{}{X_i^m})<\sum_{i=1}^{n}g(\valueAtRisk{\lambda}{X_i^m})$. Then, $\syspelve{\lambda,g}(\mathbf{X}^m)$ is the left-most root $c_m$ of the equation $h_m(c_m\lambda) = 0$ and by assumption (b), $\syspelve{\lambda,g}(\mathbf{X})$ is the unique root $c$ of the equation $h(c\lambda)=0$. The claim follows then by the uniform convergence of $(h_m)_m$ to $h$.
\end{proof}	

\section{Proofs accompanying Section~\ref{sec:concreteTheoreticalDistributions}}\label{sec:proofsSectionConcreteTheoreticalDistributions}

\begin{proof}[Proof of~\prettyref{prop:pelveForElliptical}.]
	Let $A\in\mathbb{R}^{n\times n}$ with $\Sigma=A A^{\intercal}$. By~\prettyref{rem:cholesky_elliptical} there exists $\mathbf{Y}\sim E_{n}(0,I_n,\phi)$ such that $\mathbf{X} \stackrel{d}{=} A\mathbf{Y} + \mu$.
	\begin{enumerate}[(i)]
		\item Let $i\in[n]$. By Theorem 6.16 in~\cite{embrechts_quantitative_2015} it holds that $X_i = \unitVec_i^{\intercal} \mathbf{X} \stackrel{d}{=} \mu_i + \unitVec_i^{\intercal} A\mathbf{Y}   \stackrel{d}{=} \mu_i +|A^{\intercal} \unitVec_i|_n Z$. Cash additivity and positive homogeneity of VaR and ES yield that $\valueAtRisk{\alpha}{X_i} = -\mu_i + \sqrt{\unitVec_i^{\intercal}\Sigma \unitVec_i}\valueAtRisk{\alpha}{Z}$ and $\expectedShortfall{\alpha}{X_i} = -\mu_i + \sqrt{\unitVec_i^{\intercal}\Sigma \unitVec_i}\expectedShortfall{\alpha}{Z}$ for a level $\alpha$. This shows that the PELVE of $X_i$ is independent of $\mu$ and $\Sigma$. Hence, $\pelve{\lambda}(Z) = \pelve{\lambda}(X_i)$. 
		\item The statement $\apelve{\lambda,\omega}(\mathbf{X}) = \pelve{\lambda}(Z)$ follows from part (i). $\wcpelve{\lambda}(\mathbf{X}) = \pelve{\lambda}(Z)$ follows also from part (i). The equality $\syspelve{\lambda,g}(\mathbf{X}) = \pelve{\lambda}(Z)$ is a consequence of the specific forms of the VaR and ES for elliptical distributions as already used in the proof of part (i).
		\item For an arbitrary $c\in[1,\lambda^{-1}]$ it holds that
		\begin{align*}
			\sqrt{\sum_{i=1}^{n}\omega_i\Big(\symbolExpectedShortfall{c\lambda}(X_i)-\symbolValueAtRisk{\lambda}(X_i)\Big)^2} = |\symbolExpectedShortfall{c\lambda}(Z)-\symbolValueAtRisk{\lambda}(Z)|\sqrt{\sum_{i=1}^{n}\omega_i (\unitVec_i^{\intercal}\Sigma \unitVec_i)}.  
		\end{align*}
		If $\expectedShortfall{1}{Z} = \expectation{}{-Z}\leq \valueAtRisk{\lambda}{Z}$, then by~\cite[Proposition 1]{li_pelve_2023} we have $\expectedShortfall{\pelve{\lambda}(Z)\lambda}{Z} = \valueAtRisk{\lambda}{Z}$. Otherwise, the map $c\mapsto \expectedShortfall{c\lambda}{Z}$ is decreasing towards $\expectedShortfall{1}{Z} = E[-Z]>\valueAtRisk{\lambda}{Z}$, which proves the second case.
		\item If $\expectation{}{-Z}>\valueAtRisk{\lambda}{Z}$, then by~\cite[Proposition 1]{li_pelve_2023} we have that $\pelve{\lambda}(Z)=\infty$ and the inequality is satisfied. Otherwise, note that the  value $\syspelve{\lambda,g}(\mathbf{X})$ is calculated subject to
		\begin{align}\label{eq:proof_sys_pelve_elliptical}
			\sum_{i=1}^{n}\left(\max\left\{\mu_i,\sqrt{\unitVec_i^{\intercal}\Sigma \unitVec_i}\expectedShortfall{c\lambda}{Z}\right\}-\max\left\{\mu_i,\sqrt{\unitVec_i^{\intercal}\Sigma \unitVec_i}\valueAtRisk{\lambda}{Z}\right\}\right)\leq 0.
		\end{align}
		Using again~\cite[Proposition 1]{li_pelve_2023}, we obtain $\expectedShortfall{c\lambda}{Z} = \valueAtRisk{\lambda}{Z}$ for $c = \pelve{\lambda}(Z)$, which shows that the constraint is satisfied for $c$. This proves the first part of the claim.
		
		For the second part, note that by the assumption $E[-Z]\leq \valueAtRisk{\lambda}{Z}$ we have $\pelve{\lambda}(Z)<\infty$. Furthermore, by the assumption that 
		 $\alpha\mapsto\expectedShortfall{c\lambda}{Z}$ is strictly decreasing we have $\pelve{\lambda}(Z)>1$. Now, choose an arbitrary $c\in[1,\pelve{\lambda}(Z))$. Then, by the existence of $i\in[n]$ such that $\unitVec_i^{\intercal} \Sigma \unitVec_i > 0$ and $\mu_i  \leq \sqrt{\unitVec_i^{\intercal} \Sigma \unitVec_i} \valueAtRisk{\lambda}{Z}$ we have for this $i$ together with the fact that the map $\alpha\mapsto\expectedShortfall{c\lambda}{Z}$ is strictly decreasing:
		\begin{align}\label{eq:proof_sys_pelve_elliptical_2}
			\begin{split}
			&\max\left\{\mu_i,\sqrt{\unitVec_i^{\intercal}\Sigma \unitVec_i}\expectedShortfall{c\lambda}{Z}\right\}-\max\left\{\mu_i,\sqrt{\unitVec_i^{\intercal}\Sigma \unitVec_i}\valueAtRisk{\lambda}{Z}\right\}\\
			&\quad\quad = \underbrace{\sqrt{\unitVec_i^{\intercal}\Sigma \unitVec_i}}_{>0}\big(\expectedShortfall{c\lambda}{Z}-\valueAtRisk{\lambda}{Z}\big)>0.
			\end{split} 
		\end{align}
		Furthermore, note that every summand on the left-hand side in~\prettyref{eq:proof_sys_pelve_elliptical} is greater or equal to zero for such a $c$, which implies by~\prettyref{eq:proof_sys_pelve_elliptical_2} that $\syspelve{\lambda,g}(\mathbf{X})\geq \pelve{\lambda}(Z)$. Together with the first part, i.e.,~$\syspelve{\lambda,g}(\mathbf{X})\leq \pelve{\lambda}(Z)$ we obtain $\syspelve{\lambda,g}(\mathbf{X})= \pelve{\lambda}(Z)$.
		
		To prove the remaining implication, assume that $\syspelve{\lambda,g}(\mathbf{X})= \pelve{\lambda}(Z)$ and for all $j\in[n]$ it holds that $\unitVec_j^{\intercal}\Sigma \unitVec_j = 0$ or $\mu_j\left(\unitVec_j^{\intercal}\Sigma \unitVec_j\right)^{-\frac{1}{2}}>\valueAtRisk{\lambda}{Z}$. If for all $j\in[n]$ it holds that $\unitVec_j^{\intercal}\Sigma \unitVec_j = 0$, we have  $\syspelve{\lambda,g}(\mathbf{X}) = 1<\pelve{\lambda}(Z)$. Otherwise, by the map $\alpha\mapsto\expectation{\alpha}{Z}$ being continuous, there exists $\tilde{c}\in[1,\pelve{\lambda}(Z))$ such that for all $j\in[n]$ with $\unitVec_j^{\intercal}\Sigma \unitVec_j > 0$ it holds that $\mu_j>\sqrt{\unitVec_j^{\intercal}\Sigma \unitVec_j}\expectedShortfall{\tilde{c}\lambda}{Z}$. This implies that the left-hand side in the inequality in~\prettyref{eq:proof_sys_pelve_elliptical} is equal to zero, which means that $\syspelve{\lambda,g}(\mathbf{X})\leq \tilde{c}<\pelve{\lambda}(Z)$, a contradiction. \qedhere
	\end{enumerate}
\end{proof}

\begin{proof}[Proof of~\prettyref{thm:pelveForMVR}.]
	Under the non-degeneracy condition, for $\mathbf{Y}=-\mathbf{X}$ we obtain for each $\xi\in\mathbb{R}^n$ with $\sum_{i=1}^{n}\xi_i = 1$ that $\xi^{\intercal} \mathbf{Y}\in \text{RV}_\gamma$, see~\cite{basrak_characterization_2002}. In particular, for all $i\in[n]$ we have $Y_i\in \text{RV}_{\gamma}$.
	\begin{enumerate}[(i)]
		\item This is~\cite[Theorem 3]{li_pelve_2023}, by noticing that $\valueAtRisk{\lambda}{X_i} = \upperQuantileFunction{Y_i}{1-\lambda}$ (use e.g.,~\cite[Lemma A.27]{follmer_stochastic_2016}) and hence, for $c\in[1,\lambda^{-1}]$ it holds by the Karamata Theorem~\parencite[Theorem A3.6]{embrechts_modelling_1997} and the asymptotic inverse theorem for regularly varying functions~\parencite[Theorem 1.5.12]{bingham_regular_1987} that
		\begin{align*}
			\lim\limits_{\lambda\downarrow 0} \frac{\expectedShortfall{c\lambda}{X_i}}{\valueAtRisk{\lambda}{X_i}} = \lim\limits_{\lambda\downarrow 0} \frac{\frac{1}{c\lambda} \int_{1-c\lambda}^{1}\upperQuantileFunction{Y_i}{u}\diff u}{\upperQuantileFunction{Y_i}{1-c\lambda}}\cdot\frac{\upperQuantileFunction{Y_i}{1-c\lambda}}{\upperQuantileFunction{Y_i}{1-\lambda}}= \frac{\gamma}{\gamma-1}c^{-\frac{1}{\gamma}}.
		\end{align*}
		\item The limit of the A-PELVE is a direct consequence of part (i). For the WC-PELVE, we have
		\begin{align*}
			\lim\limits_{\lambda\downarrow 0}\wcpelve{\lambda}(\mathbf{X}) = \lim\limits_{\lambda\downarrow 0}\sup_{i\in[n]}\pelve{\lambda}(X_i).
		\end{align*}
		By finiteness of the set $[n]$, we can exchange $\lim$ and $\sup$. Together with part (i) we obtain that
		\begin{align*}
			\lim\limits_{\lambda\downarrow 0}\wcpelve{\lambda}(\mathbf{X}) = \sup_{i\in[n]}\lim\limits_{\lambda\downarrow 0}\pelve{\lambda}(X_i) = \left(\frac{\gamma}{\gamma - 1}\right)^{\gamma}.
		\end{align*}
		For the Sys-PELVE with $g(x) = x$, an application of the Karamata Theorem shows that for each $i\in[n]$ it holds that
		\begin{align*}
			\frac{1}{c\lambda} \int_{1-c\lambda}^{1}\upperQuantileFunction{Y_i}{u}\diff u \sim \frac{\gamma}{\gamma-1}\upperQuantileFunction{Y_i}{1-c\lambda}\sim \frac{\gamma}{\gamma-1}c^{-\frac{1}{\gamma}}\upperQuantileFunction{Y_i}{1-\lambda},\quad \lambda\downarrow 0.
		\end{align*}
		From which we deduce that
		\begin{align*}
			\frac{\sum_{i=1}^{n}\expectedShortfall{c\lambda}{X_i}}{\sum_{i=1}^{n}\valueAtRisk{\lambda}{X_i}} =  \frac{\sum_{i=1}^{n}\frac{1}{c\lambda} \int_{1-c\lambda}^{1}\upperQuantileFunction{Y_i}{u}\diff u}{\sum_{i=1}^{n}\upperQuantileFunction{Y_i}{1-\lambda}}\sim \frac{\gamma}{\gamma-1}c^{-\frac{1}{\gamma}},\quad \lambda\downarrow 0.
		\end{align*}
		The result follows by the same reasoning as in the proof of~\cite[Theorem 3]{li_pelve_2023}.
		
		For the Sys-PELVE with $g(x) = \max\{x,0\}$ it suffices to note that  $\upperQuantileFunction{Y_i}{1-\lambda}\rightarrow \infty$ as $\lambda\downarrow 0$ for each $i\in[n]$. Hence, for sufficiently small  $\lambda$, the maximum operator becomes irrelevant, and the same reasoning as for $g(x) = x$ applies.
		
		\item For $i\in[n]$ and $Y_i\in \text{RV}_{\gamma}$, we have that $P(Y_i>x)\sim x^{-\gamma} L_i(x)$ as $x\rightarrow \infty$ for some slowly varying function $L_i$, see~\cite[Definition A3.1]{embrechts_modelling_1997} for the definition of slowly varying functions. Then, we have that
		\begin{align*}
			\upperQuantileFunction{Y_i}{\lambda} = (1-\lambda)^{-\frac{1}{\gamma}} \tilde{L}_{i}((1-\lambda)^{-1}),  
		\end{align*}
		where $\tilde{L}_{i}$ is defined via the connection that  $h(x):=\tilde{L}_{i}(x^{\gamma})$ is the de Bruijn conjugate of $(L_i)^{-\frac{1}{\gamma}}$, see also~\cite[Theorem 1.5.13]{bingham_regular_1987} for characterizing the de Bruijn conjugate. Then, for the objective of the MSE-PELVE (recall~\eqref{eq:mse}) we obtain
		\begin{align}\label{eq:extreme_MSE}
			\sqrt{\sum_{i=1}^{n}\omega_i\Big(\symbolExpectedShortfall{c\lambda}(X_i)-\symbolValueAtRisk{\lambda}(X_i)\Big)^2} \sim \sqrt{\sum_{i=1}^{n}\omega_i \Big(\lambda^{-\frac{1}{\gamma}}\tilde{L}_{i}(\lambda^{-1})\Big)^2\Big(\frac{\gamma}{\gamma-1}c^{-\frac{1}{\gamma}}-1\Big)^2},\quad \lambda\downarrow 0.
		\end{align}
		By~\cite[Theorem A3.3]{embrechts_modelling_1997} it holds that $\tilde{L}_{i}(\lambda^{-1})>0$ for small enough values of $\lambda$. Hence, the right-hand side in~\eqref{eq:extreme_MSE} attains its minimum for $c = \left(\frac{\gamma}{\gamma-1}\right)^{\gamma}$. \qedhere
	\end{enumerate}
\end{proof}

\section{Limiting behavior for gamma and Pareto distributions}\label{sec:gammaDistributions}

First, we illustrate the behavior of ES and VaR for a (negative) gamma distributed random variable, when the VaR level goes to zero, i.e.,~$\lambda\downarrow 0$. To do so, assume a random variable $X$ such that $-X\sim \Gamma(k,s)$, i.e.,~gamma distributed with shape $k$ and scale $s$. 

For $k=s=1$, $-X$ is exponentially distributed with parameter $1$. Then, by Example 5 (ii) in~\textcite{li_pelve_2023} it holds that $\pelve{\lambda}(X) = e$ if $\lambda \leq e^{-1}$. The heatmaps in  Figure~\ref{fig:ratios_gamma} demonstrate that the PELVE is also for other values of shape $k$, scale $s$ and level $\lambda$ close to the Euler number $e$. Also the values for the Multi-PELVE versions in the case study of model 1 are close to $e$, because the marginal liability distributions are given as (negative) gamma distributions, cf.~Figure~\ref{fig:pelve_lognormalGamma} in the main text. The next result shows that ES and VaR behave similar in case of a (negative) gamma distribution and small values of the level $\lambda$. 

\begin{prop}\label{prop:ratio_es_var_gamma}
	Let $X$ be a random variable with $-X\sim \Gamma(k,s)$ and $c>0$. Then,
	\[
	\lim\limits_{\lambda\downarrow 0}\frac{\expectedShortfall{c\lambda}{X}}{\valueAtRisk{\lambda}{X}} = 1.
	\]
\end{prop} 

\begin{proof}
	First, it holds that 
	\[
	\lim\limits_{\lambda\downarrow 0}\frac{\expectedShortfall{c\lambda}{X}}{\valueAtRisk{\lambda}{X}} = \lim\limits_{\lambda\downarrow 0}\left(\frac{\valueAtRisk{c\lambda}{X}}{\valueAtRisk{\lambda}{X}}\cdot\frac{\expectedShortfall{c\lambda}{X}}{\valueAtRisk{c\lambda}{X}}\right) = 1. 
	\]
	Therefore, it is enough to show that
	\[
	\lim\limits_{\lambda\downarrow 0}\frac{\expectedShortfall{\lambda}{X}}{\valueAtRisk{\lambda}{X}} = 1\quad\text{and}\quad \lim_{\lambda\downarrow 0}\frac{\valueAtRisk{c\lambda}{X}}{\valueAtRisk{\lambda}{X}} = 1.
	\]
	For the ES at level $\lambda$, we obtain
	\[
	\expectedShortfall{\lambda}{X} = E[-X|-X\geq \quantileFunction{-X}{1-\lambda}] = s\frac{\Gamma\left(k+1,\frac{\quantileFunction{-X}{1-\lambda}}{s}\right)}{\Gamma\left(k,\frac{\quantileFunction{-X}{1-\lambda}}{s}\right)},
	\]
	where $\Gamma(r,x) = \int_{x}^{\infty}u^{r-1}e^{-u}\diff u$ denotes the upper incomplete gamma function. For brevity, let us write $z(p):=\frac{\quantileFunction{-X}{1-p}}{s}$ for all $p\in(0,1)$. Then, by using the recursion formula $\Gamma(r+1,x) = x^{r}e^{-x}+r\Gamma(r,x)$ it holds that
	\[
	\expectedShortfall{\lambda}{X} = sk+s\frac{z(\lambda)^k e^{-z(\lambda)}}{\Gamma\left(k,z(\lambda)\right)}.
	\]

	Applying the asymptotic formula in~\textcite[Equation (1.01)]{Olver_1974} for $\lambda\downarrow 0$, we obtain
	\[
	\expectedShortfall{\lambda}{X} = sk + s\frac{z(\lambda)^k e^{-z(\lambda)}}{z(\lambda)^{k-1}e^{-z(\lambda)}(1+O(z(\lambda)^{-1}))} = sk + \frac{q_{-X}(1-\lambda)}{1+\mathcal{O}(z(\lambda)^{-1})}.
	\]
	Hence,
	\[
	\frac{\expectedShortfall{\lambda}{X}}{\valueAtRisk{\lambda}{X}} = \frac{\expectedShortfall{\lambda}{X}}{\quantileFunction{-X}{1-\lambda}}= \frac{sk}{\quantileFunction{-X}{1-\lambda}} + \frac{1}{1+\mathcal{O}(z(\lambda)^{-1})} \rightarrow 1,\quad \text{as }\lambda\downarrow 0.
	\]
	
	Now, denoting by $\Gamma(k) = \Gamma(k,0)$ the gamma function and using again~\textcite[Equation (1.01)]{Olver_1974}, we have for $p\downarrow 0$ that 
	\[
	p = P(X\geq \quantileFunction{-X}{1-p}) = \frac{z(p)^{k-1}}{\Gamma(k)}e^{-z(p)}(1+\mathcal{O}(z(p)^{-1})),
	\] 
	which gives us 
	\[
	\frac{\log(p)}{z(p)} = \frac{(k-1)\log(z(p))}{z(p)} - 1 -\frac{\log(\Gamma(k))}{z(p)} + \log(1+\mathcal{O}(z(p)^{-1}))\rightarrow -1,\quad\text{as }p\downarrow 0. 
	\]
	This implies the following asymptotic behavior for small $\lambda$: 
	\[
	\frac{\valueAtRisk{c\lambda}{X}}{\valueAtRisk{\lambda}{X}} = \frac{\quantileFunction{-X}{1-c\lambda}}{\quantileFunction{-X}{1-\lambda}}\sim \frac{-\log(c\lambda)}{-\log(\lambda)} = 1+\frac{\log(c)}{\log(\lambda)} \rightarrow 1,\quad \text{as }\lambda\downarrow 0, 
	\]
	which proves the claim.
\end{proof}

\begin{figure}
	\begin{subfigure}[b]{0.48\textwidth}
		\includegraphics[width=\linewidth]{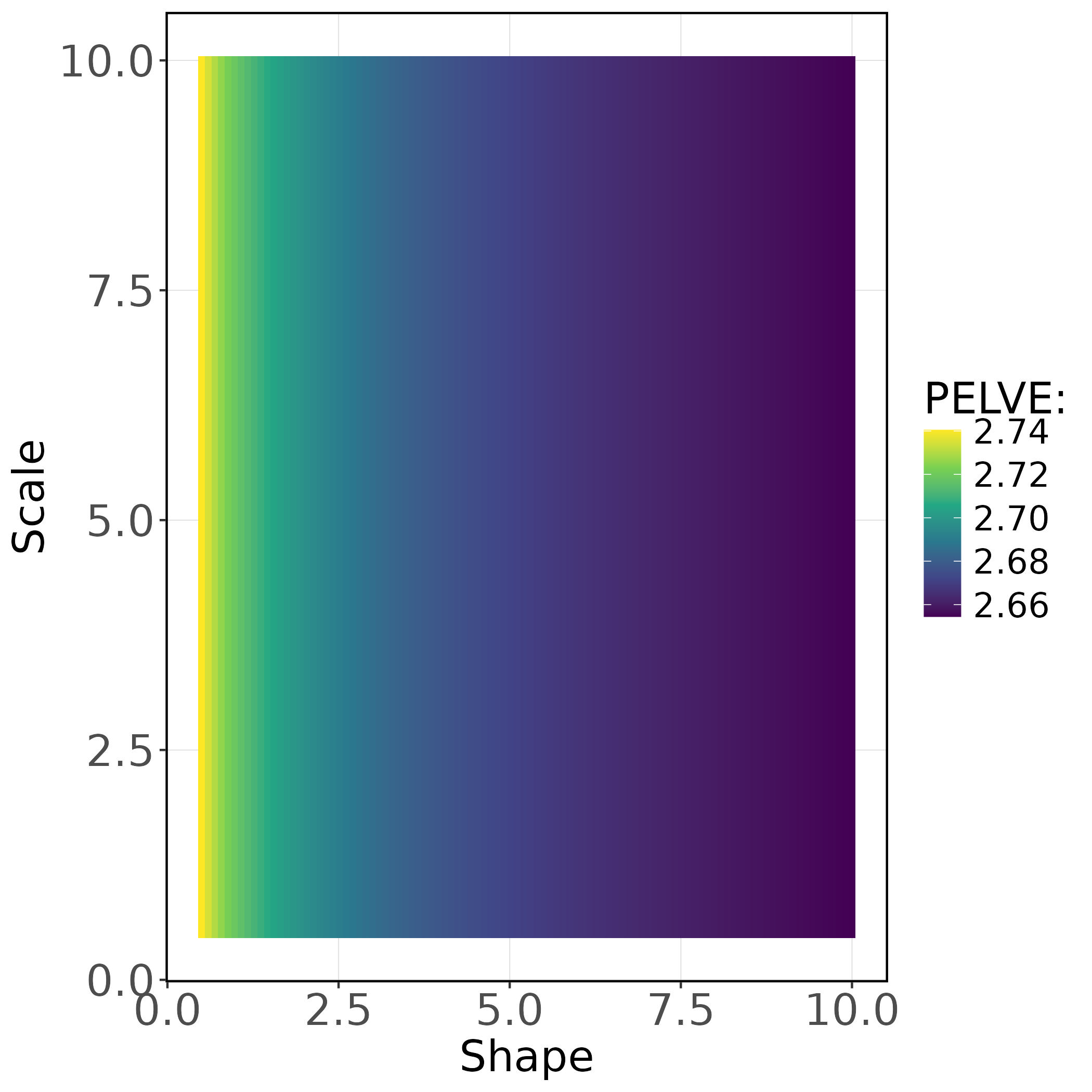}
		\caption{Heatmap of PELVE values ($\lambda = 0.5\%$)}
	\end{subfigure}
	\hfill
	\begin{subfigure}[b]{0.48\textwidth}
		\includegraphics[width=\linewidth]{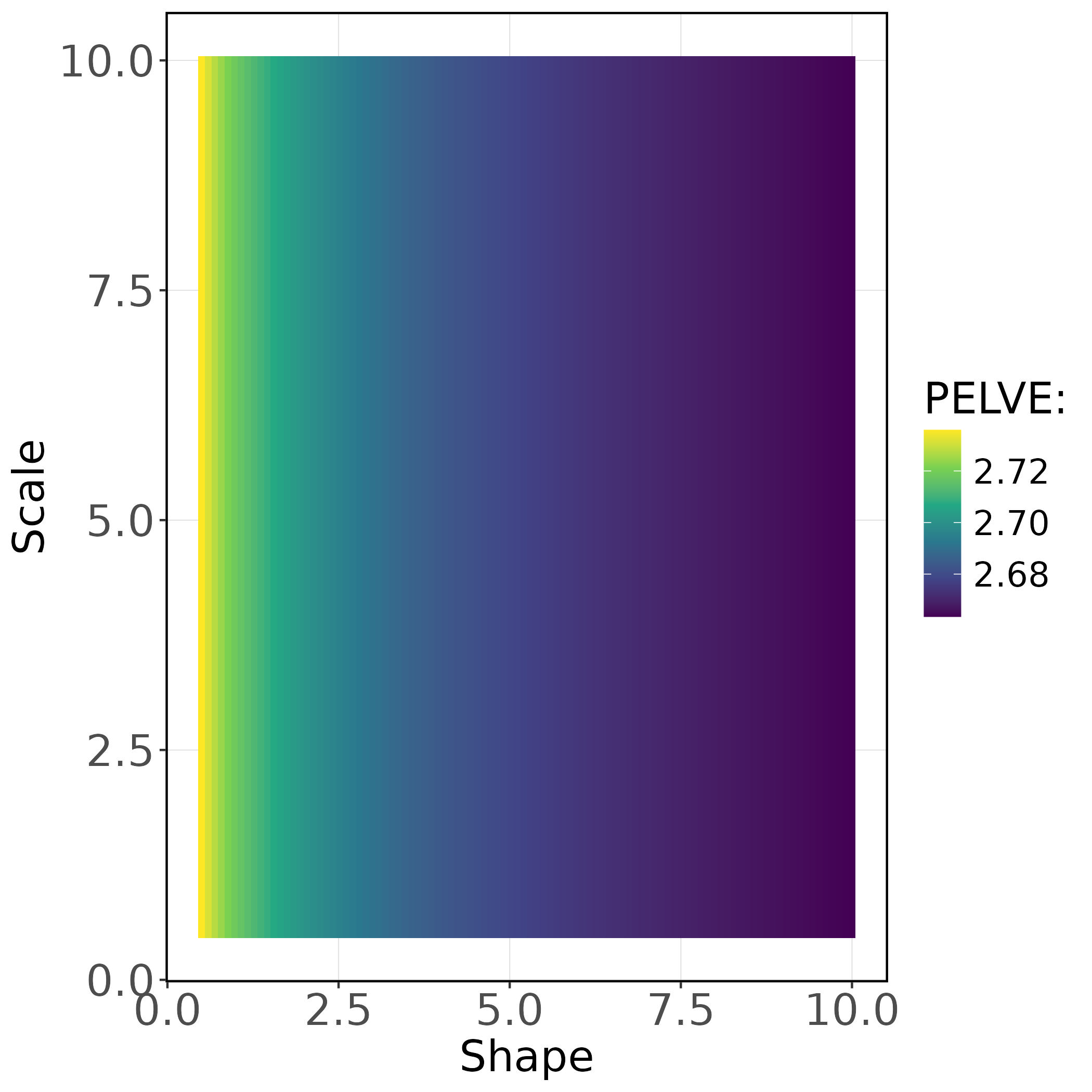}
		\caption{Heatmap of PELVE values ($\lambda = 0.25\%$)}
	\end{subfigure}
	\captionsetup{font=footnotesize}
	\caption{Heatmap of PELVE values for different shape and scale parameters.}
	\label{fig:ratios_gamma}
\end{figure}

\prettyref{prop:ratio_es_var_gamma} justifies why the sums of the absolute differences in the reserves in Figure~\ref{fig:comparing_diffs_in_sum_es_var} (B) converge to zero as $\lambda\downarrow 0$. The analogous conclusion fails for model 2, because other marginal liability distributions are applied. In particular, two  are of Pareto type. Indeed, by Proposition 3 in~\textcite{li_pelve_2023} we know for $-X$ being Pareto distributed with shape parameter $\alpha>1$ that 
\[
\lim\limits_{\lambda\downarrow 0}\frac{\expectedShortfall{\lambda}{X}}{\valueAtRisk{\lambda}{X}} = \frac{\alpha}{1-\alpha}.
\]
Then, the right-hand side goes to $\infty$ for $\alpha\downarrow 1$, i.e.,~for heavy-tailed distributions we cannot conclude that ES and VaR are asymptotically close to each other. For the Multi-PELVE versions, this implies that the sums of the absolute differences in the reserves (cf.~Figure~\ref{fig:comparing_diffs_in_sum_es_var_perturbed} (B)) do not converge to zero as $\lambda\downarrow 0$.